\documentclass[sigconf]{acmart}

\AtBeginDocument{%
  \providecommand\BibTeX{{%
    \normalfont B\kern-0.5em{\scshape i\kern-0.25em b}\kern-0.8em\TeX}}}

\copyrightyear{2021}
\acmYear{2021}
\setcopyright{acmlicensed}\acmConference[PODS '21]{Proceedings of the 40th ACM SIGMOD-SIGACT-SIGAI Symposium on Principles of Database Systems}{June 20--25, 2021}{Virtual Event, China}
\acmBooktitle{Proceedings of the 40th ACM SIGMOD-SIGACT-SIGAI Symposium on Principles of Database Systems (PODS '21), June 20--25, 2021, Virtual Event, China}
\acmPrice{15.00}
\acmDOI{10.1145/3452021.3458332}
\acmISBN{978-1-4503-8381-3/21/06}

\usepackage{pifont}
\usepackage{tikz}
\usetikzlibrary{positioning,arrows,fit,calc}
\usetikzlibrary{decorations.pathmorphing}
\usepackage{pgfplots}
\pgfplotsset{compat=newest} 

\newtheorem{theorem}{Theorem}

\newtheorem{example}{Example}
\newtheorem{lemma}[theorem]{Lemma}
\newtheorem{corollary}[theorem]{Corollary}
\newtheorem{proposition}[theorem]{Proposition}

\newtheorem{lclaim}{{\sc Claim}}[theorem]

\renewcommand{\qed}{\hfill\ding{113}}

\newcommand{\solitarypair}{solitary pair}
\newcommand{\formulablock}{main block}
\newcommand{\locality}{branch}

\newcommand{\tleft}{t_0}
\newcommand{\tright}{t_1}

\newcommand{\pbr}{properly branching}
\newcommand{\pinit}{properly initialising}
\newcommand{\itype}{input-types}
\newcommand{\pcomp}{properly computing}

\newcommand{\bform}[1]{\mbox{\sc #1}}
\newcommand{\fbr}{\bform{MustBranch}}
\newcommand{\fnobr}{\bform{NoBranch}}
\newcommand{\fgood}{\bform{Good}}
\newcommand{\freject}{\bform{Reject}}
\newcommand{\fdelta}{\bform{Step}}
\newcommand{\finit}{\bform{Init}}
\newcommand{\fhead}{\bform{Head}}
\newcommand{\fstate}{\bform{State}}
\newcommand{\fcell}{\bform{Cell}}
\newcommand{\fscell}{\bform{SameCell}}

\newcommand{\rr}{\mathfrak r}

\newcommand{\w}{\boldsymbol{w}}
\newcommand{\foc}{focus}
\newcommand{\rfoc}{root-focus}

\newcommand{\syc}{correct}
\newcommand{\syic}{incorrect}

\newcommand{\snode}{\mathfrak{s}}
\newcommand{\tnode}{\mathfrak{t}}
\newcommand{\unode}{\mathfrak{u}}
\newcommand{\vnode}{\mathfrak{v}}

\newcommand{\ct}{\mathsf{t}}
\newcommand{\cf}{\mathsf{f}}
\newcommand{\cc}{\mathsf{c}}

\newcommand{\shift}{g_\leftarrow}

\newcommand{\wormh}{{\mathcal H}^{(\ct,\cf)}}

\newcommand{\slin}[3]
{
\scriptsize
\draw[-, thick] ($({#1}) +(0,0.09)$) -- node[above=0.1] {#2}  node[below=0.1] {#3} ($({#1}) +(0,-0.09)$);
}

\newcommand{\OWLQL}{\textsl{OWL\,2\,QL}}

\newcommand{\NL}{\textsc{NL}}
\renewcommand{\L}{\textsc{L}}

\newcommand{\DL}{\textsl{DL-Lite}}

\newcommand{\DLb}{\textsl{DL-Lite}_\textit{bool}}

\newcommand{\q}{{\boldsymbol{q}}}

\renewcommand{\G}{{\boldsymbol{G}}}

\newcommand{\ACz}{{\ensuremath{\textsc{AC}^0}}}

\newcommand{\LogSpace}{\textsc{L}}

\newcommand{\coNP}{\textsc{coNP}}
\newcommand{\NP}{\textsc{NP}}
\newcommand{\PTime}{\textsc{P}}
\newcommand{\ExpTime}{\textsc{ExpTime}}
\newcommand{\Exp}{\textsc{ExpTime}}
\newcommand{\AExpSpace}{\textsc{AExpSpace}}
\newcommand{\NExpTime}{\textsc{NExpTime}}

\newcommand{\PSpace}{\textsc{PSpace}}

\newcommand{\atm}{\boldsymbol{M}}

\newcommand{\segs}{\mathfrak s}

\newcommand{\ind}{\mathsf{ind}}

\newcommand{\A}{\mathcal{D}}
\renewcommand{\C}{\mathcal{C}}
\newcommand{\T}{\mathcal{T}}

\newcommand{\I}{\mathcal{I}}

\newcommand{\avec}[1]{\boldsymbol{#1}}

\def\data{{\mathcal{D}}}

\newcommand{\inp}{\boldsymbol{\mathfrak i}}
\newcommand{\out}{\boldsymbol{\mathfrak o}}

\newcommand{\pre}{\boldsymbol{B}}
\newcommand{\per}{\boldsymbol{P}}
\newcommand{\post}{\boldsymbol{E}}

\def\D{{\data}}

\newcommand{\qmx}{\q}
\newcommand{\qmxfa}{\q^-_{TT}}
\newcommand{\omqmx}{\q}
\newcommand{\otree}{computation-tree}

\newcommand{\whatisit}{desired}

\newcommand{\dexp}{\boldsymbol{e}}

\newcommand{\good}{good}
\newcommand{\nodea}{\mathfrak a}
\newcommand{\mnode}{main node}

\newcommand{\upp}{P}

\newcommand{\inputgs}{\avec{b}_{\ga}^{\mathfrak s}}
\newcommand{\gga}{G_{\ga}}
\newcommand{\ga}{\mathfrak g}
\newcommand{\fga}{M_{\ga}}
\newcommand{\fgai}{M_{\ga_i}}
\newcommand{\foga}{\varphi_{\ga}}
\newcommand{\iga}{I_{\ga}}
\newcommand{\igai}{I_{\ga_i}}

\newcommand{\inode}{\iota_{\ga}}
\newcommand{\pnode}{\pi_{\ga}}
\newcommand{\inodei}{\iota_{\ga_i}}
\newcommand{\inodej}{\iota_{\ga_j}}
\newcommand{\fnode}{\tau_{\ga}}
\newcommand{\fnodei}{\tau_{\ga_i}}
\newcommand{\fnodej}{\tau_{\ga_j}}
\newcommand{\fnodeone}{\tau_{\ga_1}}
\newcommand{\fnodem}{\tau_{\ga_m}}

\newcommand{\rnode}{\varrho_{\ga}}
\newcommand{\rnodej}{\varrho_{\ga_j}}
\newcommand{\midnode}{\alpha}
\newcommand{\Unode}{U_{\ga}}
\newcommand{\Unodej}{U_{\ga_j}}

\newcommand{\itf}{\mathsf{m}}
 
\tikzstyle{or-gate}=[rectangle,draw,inner sep=4pt,thick]
\tikzstyle{and-gate}=[rectangle,draw,inner sep=4pt,thick]
\tikzstyle{input}=[circle,draw,minimum size=3mm]

\tikzset{>=latex, 
	point/.style = {circle,draw,thick,minimum size=2mm,inner sep=0pt},
	point1/.style = {circle,draw,thick,minimum size=6mm,inner sep=0pt},
	hm/.style = {dotted,semithick},
	role/.style = {thick},
	tree/.style = {rounded corners=10pt, dashed, fill opacity=0.5, fill=nullscolour},
	wiggly/.style={thick,
	},
	query/.style={thick},
	itria/.style={
  draw,dashed,shape border uses incircle,
  isosceles triangle,shape border rotate=90,yshift=-1.45cm},
  square/.style={regular polygon,regular polygon sides=4}
}

\newcommand{\qinit}{q_{\textit{init}}}
\newcommand{\qaccept}{q_{\textit{accept}}}
\newcommand{\qreject}{q_{\textit{reject}}}
\newcommand{\cinit}{c_{\textit{init}(\w)}}
\newcommand{\ppoly}{\boldsymbol{p}}
\newcommand{\dpoly}{\boldsymbol{d}}
\newcommand{\creject}{c_{\textit{reject}}}

\begin{document}
\fancyhead{}



\title{Deciding Boundedness of Monadic Sirups}

%
%
%
%

\author{Stanislav Kikot}
\affiliation{\institution{Institute for Information Transmission Problems}\city{Moscow}\country{ Russia}}
\email{staskikotx@gmail.com}

\author{Agi Kurucz}
\affiliation{\institution{King's College London}\city{London}\country{UK}}
\email{agi.kurucz@kcl.ac.uk}


\author{Vladimir V. Podolskii}
\affiliation{\institution{
HSE University}\city{Moscow}\country{Russia}}
\email{vpodolskii@hse.ru}

\author{Michael Zakharya\-schev}
\affiliation{\institution{Birkbeck, University of London, UK \& \\ HSE University, Moscow, Russia}\country{}}
\email{michael@dcs.bbk.ac.uk}

\renewcommand{\shortauthors}{Kikot, Kurucz, Podolskii and Zakharyaschev}

\begin{abstract}
We show that deciding boundedness (aka FO-rewritability) of mon\-adic single rule datalog programs (sirups) is 2\Exp-hard, which matches the upper bound known since 1988 and finally settles a long-standing open problem. We obtain this result as a byproduct of an attempt to classify monadic `disjunctive sirups'---Boolean conjunctive queries $\q$ with unary and binary predicates mediated by a disjunctive rule $T(x) \lor F(x) \leftarrow A(x)$---according to the data complexity of their evaluation. Apart from establishing that deciding FO-rewritability of disjunctive sirups with a dag-shaped $\q$ is also 2\Exp-hard, we make substantial progress towards obtaining a complete FO/\L-hardness dichotomy of disjunctive sirups with ditree-shaped $\q$.
\vspace*{1.3cm}
\end{abstract}


%
%

%
%

%
%



\maketitle

\section{Introduction}

There have been two waves in the investigation of boundedness or first-order rewritability of various types of recursive queries. The first one started in the mid 1980s, when the deductive database community was analysing recursion in datalog queries with the aim of optimising and parallelising their execution. One of the fundamental issues was the problem of deciding whether the depth of recursion required to evaluate a given datalog query could be bounded independently of the input data. By 2000, among other remarkable results, it had been discovered that 
\begin{itemize}
\item[--] boundedness of linear datalog queries with binary predicates and of ternary linear datalog queries with a single recursive rule is undecidable~\cite{DBLP:journals/jlp/HillebrandKMV95,DBLP:journals/siamcomp/Marcinkowski99};

\item[--] deciding program boundedness is 2\ExpTime-complete for monadic programs~\cite{DBLP:conf/stoc/CosmadakisGKV88,DBLP:conf/lics/BenediktCCB15}, \PSpace-complete for linear mon\-adic programs~\cite{DBLP:conf/stoc/CosmadakisGKV88,DBLP:journals/ijfcs/Meyden00}, and \NP-complete for linear monadic and dyadic single rule programs~\cite{DBLP:conf/pods/Vardi88}.
\end{itemize}
Interestingly, the exact complexity of deciding boundedness of monadic datalog programs with a single recursive rule, known as \emph{sirups} since~\cite{DBLP:conf/pods/CosmadakisK86}, has remained open so far, somewhere between \NP{} and 2\ExpTime, to be more precise. To clarify the `status [of boundedness] for sirups' is  part of Open Problem 4.2.10 in~\cite{DBLP:books/el/leeuwen90/Kanellakis90}. According to~\cite{DBLP:journals/jacm/AfratiP93}, Kanellakis and Papadimitriou, who were interested in datalog programs computable in \textsc{NC}, and so parallelisable, `have investigated the case of unary sirups, and have made progress towards a complete characterization'\!. Alas, that work appears to have never been completed and published. 

In this paper, we finally settle the boundedness problem for mon\-adic sirups by showing that it is 2\ExpTime-hard, which matches the upper bound for deciding boundedness of arbitrary monadic datalog  programs~\cite{DBLP:conf/stoc/CosmadakisGKV88} (and which should be compared with the \NP--\PSpace{} gap between deciding boundedness of \emph{linear} sirups and non-sirups.)

We obtained this result while surfing the second wave, which was triggered in the mid 2010s by the theory and practice of ontology-based data access (OBDA)~\cite{PLCD*08,CDLLR07,DBLP:conf/ijcai/XiaoCKLPRZ18} (recently rebranded to virtual knowledge graphs~\cite{DBLP:journals/dint/XiaoDCC19}). In OBDA, a typical ontology-mediated query (OMQ) takes the form $\boldsymbol{Q}= (\mathcal{O},\q)$ with a description logic (DL) ontology $\mathcal{O}$ and a conjunctive query (CQ) $\q$. A fundamental problem in this setting is to decide whether a given OMQ $\boldsymbol{Q}$ is FO-rewritable, in which case finding certain answers to $\boldsymbol{Q}$ can be done by evaluating a non-recursive SQL-query using a standard RDBMS. 

The ontology language \OWLQL{} for OBDA systems (such as Mastro\footnote{\url{https://www.obdasystems.com}} or Ontop\footnote{\url{https://ontopic.biz}}), standardised by the W3C in 2009, is based on \DL{} that uniformly guarantees FO-rewritability of all OMQs with an \OWLQL{} ontology. 
Uniformly FO-rewritable tgds, aka Datalog$^\pm$ or existential rules, have also been identified; see, e.g.,~\cite{DBLP:conf/datalog/CiviliR12,DBLP:journals/tods/GottlobOP14,DBLP:journals/semweb/KonigLMT15}. As an inevitable consequence, however, all of these ontology languages are very inexpressive. 

The FO-rewritability problem for OMQs in more expressive ontology languages was attacked in~\cite{DBLP:journals/tods/BienvenuCLW14} via a reduction to CSPs. It has been shown, among other results, that
\begin{itemize}
\item[--] deciding FO-rewritability of OMQs with ontologies in expressive DLs such as $\mathcal{ALC}$ (notational variant of multi-modal logic $\textbf{K}_n$) and atomic CQs is \NExpTime-complete~\cite{DBLP:journals/tods/BienvenuCLW14}, which becomes  2\NExpTime-complete in the case of (non-atomic) CQs and also monadic disjunctive datalog queries~\cite{DBLP:conf/kr/BourhisL16,DBLP:journals/lmcs/FeierKL19};

\item[--] any OMQ with a (Horn) $\mathcal{EL}$ ontology and a CQ is either FO-, or  linear-datalog-, or datalog-rewritable, and deciding this trichotomy is \textsc{ExpTime}-complete~\cite{DBLP:conf/ijcai/LutzS17,DBLP:journals/corr/abs-1904-12533}; see also~\cite{DBLP:conf/ijcai/Bienvenu0LW16,DBLP:conf/ijcai/BarceloBLP18} for complexity results on deciding FO-rewritability of OMQs with more expressive Horn description logic ontologies and frontier-guarded existential rules.
\end{itemize}
In~\cite{DBLP:conf/kr/GerasimovaKKPZ20}, aiming to single out and classify possible causes of non-FO- or non-(linear)-datalog-rewritability of OMQs, we considered (in the DL setting) a disjunctive analogue of monadic sirups, namely, monadic disjunctive datalog programs $\Delta_\q$ of the form
\begin{align}\label{d-sirup1}
T(x) \lor F(x) & \leftarrow A(x)\\\label{d-sirup2}
\boldsymbol{G} & \leftarrow \q
\end{align}
where $\q$ is a (Boolean) CQ with unary predicates $T(x)$, $F(y)$ and arbitrary binary predicates, and $\G$ is a nullary (goal) predicate. 
In DL and conceptual modelling, rule~\eqref{d-sirup1} is known as a \emph{covering axiom} (or constraint) $A \sqsubseteq T \sqcup F$ (as in `class Animal is covered by classes Male and Female'). 
We illustrate the zoo of `monadic disjunctive sirups' by an example, where CQs are given as digraphs with labelled edges and (partially) labelled nodes. 
\begin{example}\label{ex-comple}\em
Consider the CQs $\q_1,\dots,\q_5$ shown below:
\begin{center}
\begin{tikzpicture}[>=latex,line width=0.8pt, rounded corners,scale = 1.3]
		\node (0) at (-0.3,0) {$\q_1$};
		\node[point,scale=0.6,label=above:{\scriptsize $F$}] (1) at (0,0) {};
		\node[point,scale=0.6,label=above:\scriptsize$F$] (m) at (1,0) {};
		\node[point,scale=0.6,label=above:\scriptsize$T$] (2) at (2,0) {};
		\node[point,scale=0.6,label=above:\scriptsize$T$] (3) at (3,0) {};
		\draw[->,right] (1) to node[below] {\scriptsize $R$}  (m);
		\draw[->,right] (m) to node[below] {\scriptsize $R$} (2);
		\draw[->,right] (2) to node[below] {\scriptsize $R$} (3);
		\end{tikzpicture}\\[-2pt]
		\begin{tikzpicture}[>=latex,line width=0.8pt, rounded corners,scale = 1.3]
		\node (0) at (-0.3,0) {$\q_2$};
		\node[point,scale=0.6,label=above:{\scriptsize $T$}] (1) at (0,0) {};
		\node[point,scale=0.6,label=above:\scriptsize$T$] (m) at (1,0) {};
		\node[point,scale=0.6,label=above:\scriptsize$F$] (2) at (2,0) {};
		\draw[->,right] (1) to node[below] {\scriptsize $S$}  (m);
		\draw[->,right] (m) to node[below] {\scriptsize $R$} (2);
		\end{tikzpicture}\\[-2pt]
\begin{tikzpicture}[>=latex,line width=0.8pt, rounded corners,scale = 1.3]
		\node (0) at (-0.3,0) {$\q_3$};
		\node[point,scale=0.6,label=above:{\scriptsize $T$}] (1) at (0,0) {};
		\node[point,scale=0.6,label=above:{\scriptsize$T$}] (m) at (1,0) {};
		\node[point,scale=0.6,label=above:{\scriptsize$F$}] (2) at (2,0) {};
		\draw[->,right] (1) to node[below] {\scriptsize $R$}  (m);
		\draw[->,right] (m) to node[below] {\scriptsize $R$} (2);
		\end{tikzpicture}\\[-2pt]
\begin{tikzpicture}[>=latex,line width=0.8pt, rounded corners, scale = 1.3]
\node (0) at (-0.3,0) {$\q_4$};
		\node[point,scale=0.6,label=above:{},label=above:{\scriptsize $T$},label=below:{\scriptsize$z$}] (1) at (0,0) {};
		\node[point,scale=0.6,label=below:{\scriptsize$y$}] (m) at (1,0) {};
		\node[point,scale=0.6,label=above:{\scriptsize $F$},label=below:{\scriptsize$x$}] (2) at (2,0) {};
		\draw[->,right] (m) to node[below] {\scriptsize $R$}  (1);
		\draw[->,right] (m) to node[below] {\scriptsize $R$} (2);
		\end{tikzpicture}\\[-2pt]
\begin{tikzpicture}[>=latex,line width=0.8pt,rounded corners, scale = 0.9]
\node (0) at (-1.9,0) {$\q_5$};
\node[point,scale = 0.6] (0) at (-1.5,0) {};
\node[point,scale = 0.6] (1) at (0,0) {};
\node[point,scale = 0.6,label=above:{\scriptsize $T$}] (m) at (1.5,0) {};
\node[point,scale = 0.6] (2) at (3,0) {};
\node[point,scale = 0.6,label=above:{\scriptsize $FT$}] (3) at (4.5,0) {};
\node[point,scale = 0.6,label=above:{\scriptsize $F$}] (4) at (6,0) {};
\draw[<-,right] (0) to node[below] {\scriptsize $R$}  (1);
\draw[<-,right] (1) to node[below] {\scriptsize $R$}  (m);
\draw[<-,right] (m) to node[below] {\scriptsize $R$} (2);
\draw[->,right] (2) to node[below] {\scriptsize $R$} (3);
\draw[->,right] (3) to node[below] {\scriptsize $R$} (4);
\end{tikzpicture}
\end{center}
For instance, in full, rule~\eqref{d-sirup2} in the program $\Delta_{\q_4}$ looks as
$$
\G \leftarrow F(x), R(y,x), R(y,z), T(z).
$$
Intuitively, the certain answer to the Boolean query $(\Delta_{\q_4},\G)$ over a data instance $\D$ (given in the form of a labelled graph) is `yes' iff we can find the pattern $\q_4$ in every graph obtained by labelling each of the $A$-nodes in $\D$ with either $T$ or $F$. 
As shown in~\cite{DBLP:conf/kr/GerasimovaKKPZ20}, answering $(\Delta_{\q_i}, \G)$ is \coNP-complete for $\q_1$, \PTime-complete for $\q_2$, \NL-complete for $\q_3$, \LogSpace-complete for $\q_4$, and, in view of Example~\ref{ex:one more} below,  $\q_5$ is FO-rewritable and so in \ACz{}.
\end{example}

Every disjunctive sirup $\Delta_\q$, in which $\q$ has a single `solitary' $F$-node (like in $\q_2$--$\q_5$), is equivalent to a monadic datalog program $\Pi_\q$. 
For instance, $\Delta_{\q_4}$ is equivalent to $\Pi_{\q_4}$ with three rules
\begin{align*}
& \G \leftarrow F(x), R(y,x), R(y,z), P(z)\\
& P(x) \leftarrow T(x)\\
& P(x) \leftarrow A(x), R(y,x), R(y,z), P(z)
\end{align*}
Furthermore, for certain CQs $\q$, boundedness of $\Pi_\q$ coincides with boundedness of a sirup sub-program of $\Pi_\q$ (see Sec.~\ref{sec:prelims}). 
In the above example, this sirup, $\Sigma_{\q_4}$, comprises the last two rules of $\Pi_{\q_4}$, and neither $(\Delta_{\q_4},\G)$ nor $(\Sigma_{\q_4},P)$ is FO-rewritable.

On the other hand, every disjunctive sirup $\Delta_\q$ can be encoded as a CQ mediated by a Schema.org\footnote{\url{https://schema.org}: `Many applications from Google, Microsoft, Pinterest, Yandex and others already use these vocabularies to power rich, extensible experiences'\!.} ontology. Deciding FO-rewritability of UCQs mediated by Schema.org is known to be \PSpace-hard~\cite{DBLP:conf/ijcai/HernichLOW15}.


Our first result in this paper establishes 2\ExpTime-hardness of deciding FO-rewritability in all of these cases. In Sec.~\ref{sec:hardness}, we show how a computation of an alternating Turing machine 
can be captured in terms of boundedness of the disjunctive sirup $\Delta_\q$, datalog program $\Pi_\q$ or its sirup sub-program $\Sigma_\q$, for some CQ $\q$. 
Compared to known techniques, which require multiple rules in a program or a union of multiple CQs to check properties of Turing machine computations, we achieve the same aim by means of polynomially-many small Boolean circuits that are `implemented' by a \emph{single} CQ $\q$ and check local properties of binary trees representing the expansions of $\Pi_\q$.



What causes such high computational costs of recognising FO-rewritab\-ility of seemingly very primitive programs? Are there any natural classes of monadic (disjunctive) sirups whose boundedness can be checked by tractable algorithms? The 2\ExpTime-hardness proof provides three clues: first, the CQs $\q$ used in it are dags; second, each of them has two $T$-nodes; and, third, they contain many twin $FT$-nodes (as in $\q_5$ above). In~\cite{DBLP:conf/kr/GerasimovaKKPZ20}, we gave a complete classification of monadic disjunctive sirups $\Delta_\q$ with a path CQ $\q$ and an extra disjointness constraint
\begin{equation}\label{disjoint}
\bot \leftarrow T(x), F(x)
\end{equation}
(as in `classes Male and Female are disjoint') according to their data complexity (\ACz/\NL/\PTime/\coNP{}) and rewritability type (FO/linear datalog/datalog/disjunctive datalog). 

Here, in Sec.~\ref{sec:ditrees}, we make significant progress towards a complete  understanding of FO-rewrit\-ability of disjunctive sirups $\Delta_\q$ with a ditree-shaped CQ $\q$. First, we prove that twin-free CQs $\q$ as well as those that contain comparable (w.r.t.\ the tree order in $\q$) solitary $F$- and $T$-nodes (like in $\q_1$--$\q_3$ but not $\q_4$ and $\q_5$) give rise to \NL-hard disjunctive sirups. In particular, this yields a tractable FO/\NL-hardness dichotomy of the ditree disjunctive sirups with disjointness~\eqref{disjoint}. Second, we obtain a tractable FO/\L/\NL-completeness trichotomy of ditree disjunctive sirups with one solitary $F$, one solitary $T$ and any number of $FT$-twins. (This case corresponds to linear ditree sirups.) Finally, we establish an FO/\L-hardness dichotomy for ditree disjunctive sirups with one solitary $F$ and show that this dichotomy can be decided in polynomial time if the number of solitary $T$s in the CQs is bounded (like in our 2\ExpTime-hardness proof) and in exponential time otherwise. It follows  that deciding FO-rewritability of such disjunctive sirups is fixed-parameter tractable if the number of solitary $T$s is regarded as a parameter.




\section{Preliminaries}\label{sec:prelims}


We remind the reader (who can consult~\cite{Abitebouletal95} for details) that a datalog program is a finite set, $\Pi$, of rules of the form
\begin{equation}\label{rule}
\forall \avec{x}\, (\gamma_0 \leftarrow \gamma_1 \land \dots \land \gamma_m)
\end{equation}
where each $\gamma_i$ is a (constant- and function-free) atom $Q(\avec{y})$ with $\avec{y} \subseteq  \avec{x}$. As usual, we omit $\forall \avec{x}$ and replace $\land$ with a comma. The atom $\gamma_0$ is the \emph{head} of the rule, and $\gamma_1,\dots,\gamma_m$ comprise its \emph{body}. The variables in the head must also occur in the body. The predicate in the head of a rule in $\Pi$ is called an \emph{IDB predicate}; non-IDB predicates in $\Pi$ are \emph{EDB predicates}. 
We call a rule \emph{recursive} if its body has at least one IDB predicate; otherwise, it is an \emph{initialisation rule}. The \emph{arity} of $\Pi$ is the maximum arity of its IDB predicates. Here, we only consider \emph{monadic} datalog programs with at most \emph{binary} EDBs. 
A \emph{monadic sirup} is a monadic program 
with a single recursive rule.

A \emph{data instance} for $\Pi$ is any finite set $\A$ of ground atoms with EDB predicates in $\Pi$. The set of constants in $\A$ is denoted by $\ind(\A)$. 
For a unary IDB predicate $P$, a \emph{certain answer} to the \emph{datalog query} $(\Pi,P)$ \emph{over} $\A$ is any $a \in \ind(\A)$ such that $\I \models P[a]$, for every model $\I$ of $\Pi$ and $\A$, or, in other words, $P(a)$ is in the closure $\Pi(\A)$ of $\A$ under the rules in $\Pi$. For a 0-ary IDB $\G$ (goal), a \emph{certain answer} to $(\Pi,\G)$ over $\A$ is `yes' if $\G \in \Pi(\A)$, and `no' otherwise.


A typical monadic datalog program, $\Pi_\q$, we deal with in this paper is associated with a \emph{conjunctive query} (CQ) $\q$, which in our context is just a set of atoms with unary predicates $F$, $T$ and arbitrary binary predicates. An atom $F(z) \in \q$ is \emph{solitary} if $T(z) \notin \q$, and symmetrically for $T(z)$; a pair $T(z),F(z) \in \q$ is referred to as \emph{twins}. 

For a CQ $\q$ with a single solitary $F(x)$, possibly multiple solitary $T(y_1),\dots,T(y_n)$, arbitrary twins $T(z)$, $F(z)$ and binary atoms, the program $\Pi_\q$ comprises the following rules with 0-ary goal $\G$:
\begin{align}
\G &\leftarrow F(x), \q^-, P(y_1), \dots, P(y_n)\label{one}\\
P(x)  &\leftarrow T(x)\label{two} \\ 
P(x) &\leftarrow A(x), \q^-, P(y_1), \dots, P(y_n) \label{three}
\end{align}
where $\q^- = \q \setminus \{F(x),T(y_1), \dots, T(y_n)\}$, and $A$ and $P$ are fresh unary EDB and IDB  predicates, respectively. 
One can show (see~\cite{DBLP:journals/ai/KaminskiNG16,DBLP:conf/kr/GerasimovaKKPZ20} for details) that, for any such $\q$, called a \emph{1-CQ} henceforth, $(\Pi_\q,\G)$ is equivalent to the \emph{disjunctive datalog program} $(\Delta_\q, \G)$ with rules \eqref{d-sirup1} and \eqref{d-sirup2} in the sense that they return the same answer over any data instance $\A$. Here, as usual, a \emph{certain answer} to $(\Delta_\q, \G)$ over $\A$ is `yes' iff $\I \models \G$, for every model $\I$ of $\Delta_\q$ and $\A$. 


The monadic sirups, deciding boundedness of which is proved to be 2\Exp-hard in Sec.~\ref{sec:hardness}, take the form $\Sigma_\q = \{\eqref{two},\eqref{three}\}$ with a 1-CQ $\q$ and goal predicate $P$.
Adapting a similar terminology, we refer to disjunctive datalog programs $\Delta_\q = \{\eqref{d-sirup1}, \eqref{d-sirup2}\}$ and queries $(\Delta_\q, \G)$, where $\q$ may contain multiple $T$ and $F$ in general, as \emph{monadic disjunctive sirups} or \emph{d-sirups\/}, for short. 

\begin{example}\label{recursion}\em
Note that recursion in d-sirups is implicit and originates in `proof by exhaustion' or `case distinction'\!, which can be seen by evaluating $(\Delta_{\q_1}, \G)$ and $(\Delta_{\q_2}, \G)$ (or the corresponding $(\Pi_{\q_2}, \G)$), with the $\q_i$ from Example~\ref{ex-comple}, over the data instances $\A_1$ and, respectively $\A_2$ below.\\
\centerline{
\begin{tikzpicture}[line width=0.8pt,scale = 0.85]
\node (d) at (-1.8,0) {$\A_1$};
\node[point,scale=0.6,label=below:{\scriptsize $F$}] (1) at (-2,-1) {};
\node[point,scale=0.6,label=below:{\scriptsize $F$}] (2) at (-1,-1) {};
\node[point,scale=0.6,label=below:{\scriptsize $A$}] (3) at (0,-1) {};
\node[point,scale=0.6,label=below:{\scriptsize $A$}] (4) at (1,-1) {};
\node[point,scale=0.6,label=above:{\scriptsize $T$}] (5) at (2,-1) {};
\node[point,scale=0.6,label=above:{\scriptsize $T$}] (6) at (3,-1) {};
\node[point,scale=0.6,label=left:{\scriptsize $T$}] (7) at (0,0) {};
\draw[->,right] (1) to node[below] {\scriptsize $R$} (2);
\draw[->,right] (2) to node[below] {\scriptsize $R$} (3);
\draw[->,right] (3) to node[below] {\scriptsize $R$} (4);
\draw[->,right] (4) to node[below] {\scriptsize $R$} (5);
\draw[->,right] (5) to node[below] {\scriptsize $R$} (6);
\draw[->,right] (3) to node[left] {\scriptsize $R$} (7);
\end{tikzpicture}
\hspace{0.07cm} 
\begin{tikzpicture}[line width=0.8pt,scale = 0.85]
\node (d) at (0.2,-0.9) {$\A_2$};
\node (a) at (3.19,0.15) {\scriptsize$a$};
\node (b) at (2.19,0.15) {\scriptsize$b$};
\node[point,scale=0.6,label=above:{\scriptsize $T$}] (1) at (0,0) {};
\node[point,scale=0.6,label=above:\scriptsize$T$] (2) at (1,0) {};
\node[point,scale=0.6,label=above:\hspace{-0.5mm}\scriptsize$A$] (3) at (2,0) {};
\node[point,scale=0.6,label=above:\hspace{-0.5mm}{\scriptsize $A$}] (4) at (3,0) {};
\node[point,scale=0.6,label=above:\scriptsize$F$] (5) at (4,0) {};
\node[point,scale=0.6,label=right:\scriptsize$T$] (6) at (3,-1) {};
\node[point,scale=0.6,label=above:\scriptsize$T$] (7) at (2,-1) {};
\draw[->,right] (1) to node[below] {\scriptsize $S$}  (2);
\draw[->,right] (2) to node[below] {\scriptsize $R$}  (3);
\draw[->,right] (3) to node[below] {\scriptsize $S$}  (4);
\draw[->,right] (4) to node[below] {\scriptsize $R$} (5);
\draw[->,right] (6) to node[right] {\scriptsize $R$}  (4);
\draw[->,right] (7) to node[above] {\scriptsize $S$}  (6);
\end{tikzpicture}}
For instance, let $\I$ be any model of $\Delta_{\q_2}$ and $\D_2$. By rule~\eqref{d-sirup1}, each of the $A$-nodes $a$ and $b$ in $\I$ is labelled by $F$ or $T$. If $a$ is an $F$-node, $\q_2$ is embeddable in $\I$ via the vertical $R$-arrow. So let $a$ be a $T$-node. If $b$ is a $T$-node, $\q_2$ is embeddable in $\I$ starting from $a$, and if $b$ is an $F$-node, there is an embedding starting from $b$. Thus, $\I \models \q_2$.
\end{example}

A monadic (disjunctive) datalog query $(\Pi,Q)$ is \emph{bounded} or \emph{FO-rewri\-table} if there is a first-order formula $\Phi(x)$ (a sentence $\Phi$ if $Q$ is 0-ary) such that, for any data instance $\A$, a constant \mbox{$a \in \ind(\A)$} (or `yes') is a certain answer to $(\Pi,Q)$ over $\A$ iff $\A \models \Phi[a]$ (respectively, $\A \models \Phi$), where $\A$ is regarded as an FO-structure. It is known (see, e.g.,~\cite{DBLP:journals/tods/BienvenuCLW14,DBLP:journals/lmcs/FeierKL19}) that in this case $(\Pi,Q)$ is rewritable into a union of conjunctive queries (UCQ). It is also known~\cite{DBLP:conf/pods/Naughton86} that FO-rewritability of datalog queries $(\Pi,Q)$ can be characterised in terms of $Q$-\emph{expansions}, which are defined inductively below for our special queries $(\Pi_\q,\G)$ under the moniker `cactuses'\!.

To begin with, we set $\C_{\G} = \{F(x), \q^-, T(y_1),\dots,T(y_n)\} = \{\q\}$ and $\mathfrak K_\q = \{\C_\G\}$. Then we take the closure of $\mathfrak K_\q$ under the rule
\begin{itemize}
\item[\bf(bud)] if $T(y) \in \mathcal{C} \in \mathfrak K_\q$ is solitary, then we add to $\mathfrak K_\q$ the set of atoms obtained from $\C$ by replacing $T(y)$ with the atoms $A(x), \q^-, T(y_1),\dots,T(y_n)$, in which $x$ is renamed to $y$ and all other variables are given \emph{fresh} names. 
\end{itemize}
The elements of the resulting (infinite if $n \ge 1$) set $\mathfrak K_\q$ are called \emph{cactuses} for $(\Pi_\q,\G)$. 
We represent cactuses as labelled digraphs. 


For $\C \in \mathfrak K_\q$, we refer to the copies $\mathfrak s$ of (maximal subsets of) $\q$ that comprise $\C$ as \emph{segments\/} and to the copy of the solitary $F$-node in $\mathfrak s$ as its \emph{\foc}. The \emph{skeleton} $\C^s$ of $\C$ is the ditree whose nodes are the segments $\mathfrak s$ of $\C$ and edges $(\mathfrak s, \mathfrak s')$ mean that $\mathfrak s'$ was attached to $\mathfrak s$ by budding. The \emph{depth of $\mathfrak s$ in} $\C$ 
(or \emph{in} $\C^s$) is the number of edges on the branch from the root of $\C^s$ to $\mathfrak s$. The \emph{depth of} $\C$ is the maximum depth of its segments. 

\begin{example}\label{cac-ill}\em
The data instance $\D_2$ from Example~\ref{recursion} is (isomorphic to) a cactus from $\mathfrak K_{\q_2}$ obtained by applying {\bf (bud)} to $\q_2$ twice. The skeleton $\D_2^s$ along with its three segments 
$\mathfrak s_0,\mathfrak s_1,\mathfrak s_2$ and their respective focuses $z_0,z_1,z_2$ are illustrated below:\\
\centerline{
\begin{tikzpicture}[line width=0.8pt,scale = 0.85]
\node (d) at (0,.5) {$\D_2^s$};
\node[point,scale=0.5,fill=black,label=below:{$\mathfrak s_1$}] (2) at (-.5,-1) {};
\node[point,scale=0.5,fill=black,label=below:{$\mathfrak s_2$}] (4) at (.5,-1) {};
\node[point,scale=0.5,fill=black,label=left:{$\mathfrak s_0$}] (7) at (0,0) {};
\draw[->,right] (7) to node[below] {} (2);
\draw[->,right] (7) to node[below] {} (4);
\end{tikzpicture}
\hspace{0.5cm} 
\begin{tikzpicture}[line width=0.8pt,scale = 1.3]
\node (d) at (0.9,-0.9) {$\A_2$};
\node[point,scale=0.6,label=above:{\scriptsize $T$}] (1) at (0,0) {};
\node[point,scale=0.6,label=above:\scriptsize$T$] (2) at (1,0) {};
\node[point,scale=0.6,label=above:\scriptsize$A$,label=below:\scriptsize$z_2$] (3) at (2,0) {};
\node[point,scale=0.6,label=above:{\scriptsize $A$},label=below right:\!\scriptsize$z_1$] (4) at (3,0) {};
\node[point,scale=0.6,label=above:\scriptsize$F$,label=below:\scriptsize$z_0$] (5) at (4,0) {};
\node[point,scale=0.6,label=right:\!\scriptsize$T$] (6) at (3,-1) {};
\node[point,scale=0.6,label=above:\scriptsize$T$] (7) at (2,-1) {};
\draw[->,right] (1) to node[below] {\scriptsize $S$}  (2);
\draw[->,right] (2) to node[below] {\scriptsize $R$}  (3);
\draw[->,right] (3) to node[below] {\scriptsize $S$}  (4);
\draw[->,right] (4) to node[below] {\scriptsize $R$} (5);
\draw[->,right] (6) to node[right] {\scriptsize $R$}  (4);
\draw[->,right] (7) to node[above] {\scriptsize $S$}  (6);
\draw[thin,dashed,rounded corners=10] (4.3,.3) -- (4.3,-.3) -- (1.7,-.3) -- (1.7,.3) -- cycle;
\node[] at (3.7,.45) {$\mathfrak s_0$}; 
\draw[thin,dashed,rounded corners=10] (2.3,.4) -- (2.3,-.4) -- (-.3,-.4) -- (-.3,.4) -- cycle;
\node[] at (0,-.6) {$\mathfrak s_2$}; 
\draw[thin,dashed,rounded corners=10] (3.3,.4) -- (3.3,-1.3) -- (1.8,-1.3) -- (1.8,-.6) -- (2.7,-.6) -- (2.7,.4) -- cycle;
\node[] at (3.5,-1) {$\mathfrak s_1$}; 
\end{tikzpicture}}
\end{example}

%
In the remainder of this section, we establish a connection between boundedness of $(\Pi_\q,\G)$ and $(\Sigma_\q,P)$, for a 1-CQ $\q$, which requires a few  definitions.
Every cactus $\C \in \mathfrak K_\q$ has exactly one $F$-node. We call it the \emph{\rfoc} of $\C$ and denote it by $r$. By replacing the $F$-label of $r$ in $\C$ with $A$, we obtain a digraph $\C^\circ$; the set of all such $\C^\circ$, for $\C \in \mathfrak K_\q$, is denoted by $\mathfrak K_\q^\circ$. 
The following proposition is proved by a standard induction on the derivation length:
\begin{proposition}\label{cactuses}
For any data instance $\A$ and any $a \in \ind(\A)$,
\begin{itemize}
\item[--] $\G \in \Pi_\q(\A)$ iff there is a homomorphism from some cactus $\C \in \mathfrak K_\q$ to $\A$\textup{;}

\item[--] $P(a) \in \Sigma_\q(\A)$ iff either $T(a) \in \A$ or there is a homomorphism $h$ from some $\C^\circ \in \mathfrak K_\q^\circ$ to $\A$ such that $h(r) = a$.
\end{itemize}
\end{proposition}

A 1-CQ $\q$ is called \emph{\foc ed} if the following condition holds:
\begin{itemize}
\item[\bf(foc)] for any cactuses $\C,\C'\in \mathfrak K_\q$, if there is a homomorphism $h \colon \C \to \C'$,  then $h(r) = r$.   
\end{itemize}
%
The significance of this notion is shown by Example~\ref{ex:one more} below, and by the following characterisation of boundedness; cf.~\cite{DBLP:conf/pods/Naughton86}:

\begin{proposition}\label{thmequi}
For every \foc{}ed 1-CQ $\q$ with solitary $F(x)$, $T(y_1),\dots,T(y_n)$, the following conditions are equivalent:
\begin{itemize}
\item[$(a)$] $(\Sigma_\q,P)$ is bounded;

\item[$(b)$] $(\Pi_\q,\G)$ is bounded; 

\item[$(c)$] there exists $d < \omega$ such that, for every $\mathcal{C} \in \mathfrak K_\q$, there is a homomorphism $h\colon \mathcal{C}' \to \mathcal{C}$, for some $\mathcal{C}' \in \mathfrak K_\q$ of depth $\le d$.
\end{itemize}
Conditions $(b)$ and $(c)$ are equivalent for every \textup{(}not necessarily focused\textup{)} 1-CQ $\q$, in which case $(a)$ is equivalent to $(c)$ with an additional requirement that $h(r) = r$.
\end{proposition}
\begin{proof}
$(a) \Rightarrow (b)$ If $\Phi(x)$ is an FO-rewriting of $(\Sigma_\q,P)$, then 
$$
\exists x,y_1,\dots,y_n,\avec{z} \, \big( F(x) \land \q' \land \Phi(y_1) \land \dots \land \Phi(y_n) \big)
$$
is an FO-rewriting of $(\Pi_\q,\G)$, where $\avec{z}$ comprises the variables in $\q'$ different from $x,y_1,\dots,y_n$.

$(b) \Rightarrow (c)$ Let $\exists \avec{y}\, (\q_1 \lor \dots \lor \q_m)$ be a UCQ-rewriting of $(\Pi_\q,\G)$, where the $\q_i$ are CQs and $\avec{y}$ comprises their variables. Treating the $\q_i$ as data instances, we obviously have $\G \in \Pi_\q(\q_i)$, and so, for every $i$, $1 \le i \le m$, there is a homomorphism from some $\C_i \in \mathfrak K_\q$ to $\q_i$. Let $d$ be the maximum depth of the $\C_i$, $i=1,\dots,m$. Consider any $\C \in \mathfrak K_\q$. Then there are homomorphisms $\C_i \to \q_i \to \C$, for some $i$, $1 \le i \le m$, the composition of which is the required $h$. 

$(c) \Rightarrow (a)$ By Prop.~\ref{cactuses} and (c), the sentence $\exists r,\avec{y}\,(\C_1 \lor \dots \lor \C_m)$, where the $\C_i$ are all of the cactuses of depth $\le d$ with root-\foc{} $r$ and the remaining variables $\avec{y}$, is an FO-rewriting of $(\Pi_\q,\G)$. We show that the formula
$
\Phi(r) = T(r) \lor \exists \avec{y}\, \big(\C^\circ_1 \lor \dots \lor \C^\circ_m\big)
$
is an FO-rewriting of $(\Sigma_\q,P)$. Let $P(a) \in \Sigma_\q(\A)$, for some $\A$ and $a \in \ind(\A)$. By Prop.~\ref{cactuses}, either $T(a) \in \A$, in which case $\A \models \Phi[a]$, or there is a homomorphism $h$ from some $\C^\circ \in \mathfrak K_\q$ to $\A$ such that $h(r) = a$. By (c), there is a homomorphism $g \colon \C_i \to \C$, for some $i \le m$. As $\q$ is focused, $g(r) = r$, and so we can regard $g$ as a $\C^\circ_i \to \C^\circ$ homomorphism. But then we obtain a homomorphism $hg \colon \C_i^\circ \to \D$ with $hg(r) = a$, from which $\D \models \exists \avec{y}\,  \C_i^\circ[a]$. That $\D \models \Phi[a]$ implies $P(a) \in \Sigma_\q(\A)$ is trivial.
\end{proof}

The next example illustrates the difference between focused and unfocused 1-CQs $\q$ as far as boundedness of $(\Pi_\q,\G)$ and $(\Sigma_\q,P)$ is concerned.

\begin{example}\label{ex:one more}\em
Consider the 1-CQ $\q_5$ from Example~\ref{ex-comple}. Let $\C_k$ be the cactus obtained by applying {\bf (bud)} $k$-times to $\C_0 = \q_5$. There are homomorphisms $h\colon \C_1 \to \C_k$, for $k \ge 2$, and so both $(\Pi_{\q_5},\G)$ and $(\Delta_{\q_5},\G)$ are rewritable to the UCQ $\C_0 \lor \C_1$. For each such $h$, we have $h(r) = r$, so $\q_5$ is focused and the sirup $(\Sigma_{\q_5},P)$ is bounded.

Now, consider the 1-CQ $\q_6$ below, where all of the arrows are labelled by $R$. 
It is not hard to see that, for every $\C' \in \mathfrak K_{\q_6}$ of depth $\ge 2$, there exist $\C\in \mathfrak K_{\q_6}$ of depth $\le 1$ and a homomorphism $h \colon \C \to \C'$,
so $(\Pi_{\q_6},\G)$ and $(\Delta_{\q_6},\G)$ are FO-rewritable. However, every such $h$ maps the root-focus $F$-node 
$r$ to an $FT$-node, and so $\q_6$ is not focused. In the picture below, $\C$ is obtained by budding at $t_0$,
and $\C'$ by budding first at $t_1$ and then at $t_0$.
Using Prop.~\ref{thmequi}, one can show that $(\Sigma_{\q_6},P)$ is not bounded.\\
\centerline{
\begin{tikzpicture}[>=latex,line width=0.8pt,rounded corners, scale = 0.6]
\node[] at (-3,0) {$\q_6$};
\node[point,scale = 0.6] (0) at (-1.5,0) {};
\node[point,scale = 0.6] (1) at (0,0) {};
\node[point,scale = 0.6] (m) at (1.5,0) {};
\node[point,scale = 0.6,label=above:{\scriptsize $F$}] (2) at (3,0) {};
\node[point,scale = 0.6] (3) at (4.5,0) {};
\node[point,scale = 0.6,label=above:{\scriptsize $FT$}] (4) at (6,0) {};
\node[point,scale = 0.6,label=above:{\scriptsize $T$},label=below:{\scriptsize $t_0$}] (5) at (7.5,0) {};
\node[point,scale = 0.6,label=above:{\scriptsize $T$},label=below:{\scriptsize $t_1$}] (6) at (9,0) {};
\draw[<-,right] (0) to node[below] {}  (1);
\draw[<-,right] (1) to node[below] {}  (m);
\draw[<-,right] (m) to node[below] {} (2);
\draw[<-,right] (2) to node[below] {} (3);
\draw[->,right] (3) to node[below] {} (4);
\draw[->,right] (4) to node[below] {} (5);
\draw[->,right] (5) to node[below] {} (6);
\end{tikzpicture}}

\centerline{
\begin{tikzpicture}[>=latex,line width=0.8pt,rounded corners, xscale = 0.7,yscale = 0.6]
\node[point,scale = 0.6,label=above left:{\scriptsize $F$}\!,label=below left:{\scriptsize $r$}\!] (aa8) at (4,3) {};
\node[point,scale = 0.6] (bb3) at (3,3) {};
\node[point,scale = 0.6] (bb2) at (2,3) {};
\node[point,scale = 0.6] (bb1) at (1,3) {};
\node[point,scale = 0.6] (bb4) at (5,4) {};
\node[point,scale = 0.6,label=above:{\scriptsize $FT$}] (bb5) at (6,3) {};
\node[point,scale = 0.6,label=above:{\scriptsize $A$}] (bb6) at (7,2) {};
\node[point,scale = 0.6,label=right:\!{\scriptsize $T$}] (bb7) at (8,2) {};
\node[point,scale = 0.6] (cc3) at (6,2) {};
\node[point,scale = 0.6] (cc2) at (5,2) {};
\node[point,scale = 0.6] (cc1) at (4,2) {};
\node[point,scale = 0.6] (cc4) at (8,3) {};
\node[point,scale = 0.6,label=above:{\scriptsize $FT$}] (cc5) at (9,3) {};
\node[point,scale = 0.6,label=above:{\scriptsize $T$}] (cc6) at (10,3) {};
\node[point,scale = 0.6,label=above:{\scriptsize $T$}] (cc7) at (11,3) {};
\draw[->] (bb4) to (aa8);
\draw[->] (aa8) to (bb3);
\draw[->] (bb3) to (bb2);
\draw[->] (bb2) to (bb1);
\draw[->] (bb4) to (bb5);
\draw[->] (bb5) to (bb6);
\draw[->] (bb6) to (bb7);
\draw[thin,dashed,rounded corners=10] (.7,2.35) -- (4,2.35) -- (5,3.5) -- (6,2.5) -- (7,1.5) -- (8.5,1.5) -- (8.5,2.7) -- (7,2.7) -- (5,4.7) -- (3.8,3.5) -- (.7,3.5) -- cycle;
\draw[->] (cc4) to (bb6);
\draw[->] (bb6) to (cc3);
\draw[->] (cc3) to (cc2);
\draw[->] (cc2) to (cc1);
\draw[->] (cc4) to (cc5);
\draw[->] (cc5) to (cc6);
\draw[->] (cc6) to (cc7);
\draw[dotted,rounded corners=6] (3.7,2.2) -- (6.3,2.2) -- (7.6,3.7) -- (11.4,3.7) -- (11.4,2.6) -- (8,2.6) -- (7,1.6) -- (3.7,1.6) -- cycle;

\node[point,scale = 0.6] (a1) at (0,0) {};
\node[point,scale = 0.6] (a2) at (1,0) {};
\node[point,scale = 0.6] (a3) at (2,0) {};
\node[point,scale = 0.6,label=above:{\scriptsize $F$},label=below:{\scriptsize $r$}] (a4) at (3,0) {};
\node[point,scale = 0.6] (a5) at (4,0) {};
\node[point,scale = 0.6,label=left:{\scriptsize $FT$}] (a6) at (4,-1) {};
\node[point,scale = 0.6,label=left:{\scriptsize $T$}] (a7) at (4,-2) {};
\node[point,scale = 0.6,label=above left:{\scriptsize $A$}\!] (a8) at (4,-3) {};
\node[point,scale = 0.6] (b3) at (3,-3) {};
\node[point,scale = 0.6] (b2) at (2,-3) {};
\node[point,scale = 0.6] (b1) at (1,-3) {};
\node[point,scale = 0.6] (b4) at (5,-2) {};
\node[point,scale = 0.6,label=above:{\scriptsize $FT$}] (b5) at (6,-3) {};
\node[point,scale = 0.6,label=above:{\scriptsize $A$}] (b6) at (7,-4) {};
\node[point,scale = 0.6,label=right:\!{\scriptsize $T$}] (b7) at (8,-4) {};
\node[point,scale = 0.6] (c3) at (6,-4) {};
\node[point,scale = 0.6] (c2) at (5,-4) {};
\node[point,scale = 0.6] (c1) at (4,-4) {};
\node[point,scale = 0.6] (c4) at (8,-3) {};
\node[point,scale = 0.6,label=above:{\scriptsize $FT$}] (c5) at (9,-3) {};
\node[point,scale = 0.6,label=above:{\scriptsize $T$}] (c6) at (10,-3) {};
\node[point,scale = 0.6,label=above:{\scriptsize $T$}] (c7) at (11,-3) {};
\draw[->] (a5) to (a4);
\draw[->] (a4) to (a3);
\draw[->] (a3) to (a2);
\draw[->] (a2) to (a1);
\draw[->] (a5) to (a6);
\draw[->] (a6) to (a7);
\draw[->] (a7) to (a8);
\draw[dotted,rounded corners=10] (-.3,.6) -- (4.3,.6) -- (4.3,-3.7) -- (3.3,-3.7) -- (3.3,-.6) -- (-.3,-.6) -- cycle;
\draw[->] (b4) to (a8);
\draw[->] (a8) to (b3);
\draw[->] (b3) to (b2);
\draw[->] (b2) to (b1);
\draw[->] (b4) to (b5);
\draw[->] (b5) to (b6);
\draw[->] (b6) to (b7);
\draw[thin,dashed,rounded corners=10] (.7,-3.5) -- (4,-3.5) -- (5,-2.5) -- (6,-3.5) -- (7,-4.5) -- (8.5,-4.5) -- (8.5,-3.3) -- (7,-3.3) -- (5,-1.3) -- (3.8,-2.5) -- (.7,-2.5) -- cycle;
\draw[->] (c4) to (b6);
\draw[->] (b6) to (c3);
\draw[->] (c3) to (c2);
\draw[->] (c2) to (c1);
\draw[->] (c4) to (c5);
\draw[->] (c5) to (c6);
\draw[->] (c6) to (c7);
\draw[dotted,rounded corners=6] (3.7,-3.8) -- (6.3,-3.8) -- (7.6,-2.3) -- (11.4,-2.3) -- (11.4,-3.4) -- (8,-3.4) -- (7,-4.4) -- (3.7,-4.4) -- cycle;
\draw[->,gray,very thick] (4.1,2.9) to [out=-55,in=90] (6,-2.5);
\draw[->,gray,very thick] (9,2.8) to [out=-80,in=80] (9,-2.5);
\draw[->,gray,very thick] (11,2.8) to [out=-80,in=80] (11,-2.5);
\node[gray] at (10,0) {$h$};
\node[] at (0,1) {$\C'$};
\node[] at (.2,3) {$\C$};
\end{tikzpicture}}
\end{example}


%
%



\section{Deciding boundedness of sirups}\label{sec:hardness}

In this section, we prove the following:
\begin{theorem}\label{thm:2-exp}
The problems of deciding boundedness of monadic sirups $(\Sigma_\q,P)$ and monadic d-sirups $(\Delta_\q,\G)$ are both 2\Exp-hard. 
\end{theorem}

Before diving into technical details, we put this theorem into the context of related work.


\subsection{Related results}

That deciding program boundedness of arbitrary monadic datalog queries can be done in 2\Exp{} was shown in 1988 using an automata-theoretic technique~\cite{DBLP:conf/stoc/CosmadakisGKV88}. A matching lower bound for monadic queries with multiple recursive rules was finally settled in 2015~\cite{DBLP:conf/lics/BenediktCCB15} using a construction from~\cite{DBLP:conf/icalp/BenediktBS12}, which is based on the encoding 
  of Turing machine computations from~\cite{DBLP:conf/mfcs/BjorklundMS08,DBLP:journals/acta/BjorklundMS18}. 
For monadic sirups, the \NP{} lower bound for the linear case~\cite{DBLP:conf/pods/Vardi88} has remained so far the best known result (though, in view of Prop.~\ref{thmequi} and the proof of~\cite[Theorem~9]{DBLP:conf/kr/GerasimovaKKPZ20}, it can be raised to \PSpace).

Establishing a higher lower bound for monadic sirups is difficult for two obvious reasons: monadicity  and singularity. The impact of arity and the number of recursive rules on deciding boundedness of datalog programs has been studied in great detail; see~\cite{DBLP:conf/pods/HillebrandKMV91,DBLP:journals/siamcomp/Marcinkowski99} and further references therein. For example, boundedness was shown to be undecidable first for linear datalog programs of arity 4~\cite{DBLP:conf/lics/GaifmanMSV87}, then for those of arity 2 with multiple recursive rules~\cite{DBLP:conf/pods/Vardi88}, which were encoded in a single rule at the expense of higher arity~\cite{DBLP:journals/ipl/Abiteboul89}; finally, boundedness was proved to be undecidable already for linear sirups of arity 3~\cite{DBLP:journals/siamcomp/Marcinkowski99}. 

Intuitively, the proofs of the lower bounds mentioned above use different rules in a datalog program in order to detect and exclude different `defects' in possible computations of a Turing machine. Our task in the proof of Theorem~\ref{thm:2-exp} will be to design such an encoding of computations that can be verified by a single CQ.


\subsection{Proof idea}\label{pidea}

To achieve this, similarly to~\cite{DBLP:conf/mfcs/BjorklundMS08, DBLP:journals/acta/BjorklundMS18, DBLP:conf/lics/BenediktCCB15,DBLP:conf/icalp/BenediktBS12, DBLP:journals/tocl/BenediktBGS20}, we represent computations of a Turing machine by means of annotated binary trees. The design of the tree-representation of computations is such that its structure can be connected with expansions (cactuses) of a given sirup via a series of small Boolean circuits, which is the main innovation of our construction.

%

More precisely, we use the criterion of Prop.~\ref{thmequi} for testing boundedness.
%
Our aim is, given any alternating Turing machine (ATM) $\atm$ deciding a language in $\AExpSpace=2\Exp$ and an input $\w$, to construct a (dag-shaped) \foc{}ed 1-CQ $\q$ of polynomial size such that the following holds:

\begin{lemma}\label{l:dagq}
$\atm$ rejects $\w$ iff there is $K<\omega$ such that every cactus $\C\in \mathfrak K_{\omqmx}$ contains a homomorphic image of some $\C^-\in \mathfrak K_{\omqmx}$ of depth at most $K$.
\end{lemma}

We represent both the computation space of $\atm$ on $\w$ and $\q$-cactuses by $01$-\emph{trees\/}:
binary ditrees whose edges are labelled by $0$ or $1$, with siblings having different labels.
On the one hand, we encode the computation space of $\atm$ in such a way that checking whether an arbitrary $01$-tree represents a rejecting computation-tree on $\w$ can be done by means of polynomially-many polynomial-size Boolean circuits (in fact, formulas). On the other hand, the 1-CQ $\q$ we associate with $\atm$ and $\w$ has two solitary $T$-nodes, $\tleft$ and $\tright$. Thus, we can regard the skeleton $\C^s$ of any cactus $\C \in \mathfrak K_{\omqmx}$ as a $01$-tree, indicating which of  $\tleft$ or $\tright$ were budded. The 1-CQ $\q$ is assembled from gadgets implementing the Boolean circuits used for checking the above properties of computations. 


%



\subsection{Connecting computations and cactuses}\label{strees}


\subsubsection{Encoding computations by $01$-trees}\label{comptrees}

We assume that we are given an ATM $\atm = (Q ,\Gamma ,\delta ,\qinit,\qaccept,\qreject ,g)$ with states $Q$ including $\qinit$, $\qaccept$, $\qreject$, tape alphabet $\Gamma$, transition function $\delta$, and $g\colon Q\to \{\land,\lor\}$.
For any input $\w\in\Gamma^\ast$,
a \emph{configuration} of $\atm$ is a triple containing information about the current state, the current position of the head, and the current content of the $2^{\ppoly(|\w|)}=2^{\ppoly}$ tape-cells, for some polynomial  $\ppoly$.
If its current state is $q$, then we call the configuration a $q$-\emph{configuration}. 
The \emph{full computation space} $\T_{\atm,\w}$ is a finite tree whose nodes are (labelled by) configurations, with its
root being the initial configuration $\cinit$ (in state $\qinit$ reading the leftmost symbol of $\w$),
the descendants generated by $\delta$,
and each leaf being either a $\qaccept$- or  a $\qreject$-configuration (a \emph{halting configuration}).
We assume that the depth of $\T_{\atm,\w}$ is $2^{\ppoly(|\w|)}$,
$\qinit,\qaccept,\qreject$ are $\lor$-states, every non-leaf has branching $2$, and $\land$- and $\lor$-configurations alternate on each branch.
A \emph{\otree{}} (\emph{of $\atm$ on $\w$}) is a substructure $\T$ of $\T_{\atm,\w}$, which is a tree with root $\cinit$ such that every non-leaf $\land$-node ($\lor$-node) in $\T$ has both   (respectively, exactly one) of its children from $\T_{\atm,\w}$ in $\T$.
The tree $\T$ is \emph{rejecting} if it has a $\qreject$-leaf and \emph{accepting} otherwise. 
$\atm$ \emph{rejects} $\w$ iff all \otree{}s of $\atm$ on $\w$ are rejecting, and \emph{accepts} $\w$ otherwise.

%

We encode a \otree{} $\T$ by an \emph{infinite} $01$-tree $\beta_\T^+$ via a series of steps as follows. 
First, by our assumption on binary branching, $\T$ can be considered as a (finite) $01$-tree $\beta_\T^0$
(with its nodes still labelled by configurations).
Next, we take the full binary `substructure' $\beta_\T^1$ of the $\lor$-configurations in $\beta_\T^0$ as shown below:\\
%
\centerline{
\setlength{\unitlength}{.075cm}
\begin{picture}(80,53)(-6,0)
\thicklines
\multiput(0,10)(10,0){4}{\circle*{.5}}
\multiput(5,20)(20,0){2}{\circle*{.5}}
\multiput(5,30)(20,0){2}{\circle*{.5}}
\put(15,40){\circle*{.5}}
\put(15,50){\circle*{.5}}
\put(15,41){\line(0,1){8}}
\put(6,31){\line(1,1){8}}
\put(24,31){\line(-1,1){8}}
\multiput(5,21)(20,0){2}{\line(0,1){8}}
\multiput(.5,11)(20,0){2}{\line(1,2){4}}
\multiput(9.5,11)(20,0){2}{\line(-1,2){4}}
\put(-1,13){$0$}
\put(9,13){$1$}
\put(19,13){$0$}
\put(29,13){$1$}
\put(6,34){$0$}
\put(22,34){$1$}
\put(12,44){$0$}
\put(6,24){$1$}
\put(22,24){$0$}
\put(-9,9){$\lor$}
\put(-9,19){$\land$}
\put(-9,29){$\lor$}
\put(-9,39){$\land$}
\put(-9,49){$\lor$}
\put(17,50){$\cinit$}
\put(17,40){$c_1$}
\put(-.5,30){$c_2$}
\put(27,30){$c_3$}
\put(-.5,20){$c_4$}
\put(27,20){$c_5$}
\put(-2,6.5){$c_6$}
\put(8,6.5){$c_7$}
\put(18,6.5){$c_8$}
\put(28,6.5){$c_9$}
\put(12,2){$\beta_\T^0$}
\put(38,25){\Large $\leadsto$}
\multiput(50,15)(10,0){4}{\circle*{.5}}
\multiput(55,25)(20,0){2}{\circle*{.5}}
\put(65,35){\circle*{.5}}
\put(56,26){\line(1,1){8}}
\put(74,26){\line(-1,1){8}}
\multiput(50.5,16)(20,0){2}{\line(1,2){4}}
\multiput(59.5,16)(20,0){2}{\line(-1,2){4}}
\put(49,18){$0$}
\put(59,18){$1$}
\put(69,18){$0$}
\put(79,18){$1$}
\put(56,29){$0$}
\put(72,29){$1$}
\put(66,37){$\cinit$}
\put(49.5,25){$c_2$}
\put(77,25){$c_3$}
\put(48,11.5){$c_6$}
\put(58,11.5){$c_7$}
\put(68,11.5){$c_8$}
\put(78,11.5){$c_9$}
\put(62,6){$\beta_\T^1$}
\end{picture}
}\\ 
(So the depth  of $\beta_\T^1$ is $2^{\ppoly-1}$.) 
%
The information about which child of each $\lor$-configuration is taken in $\beta_\T^0$ is provided in the encoding of
the subsequent $\lor$-configuration. To achieve this,
we fine-tune the `configurations-as-binary-tree-leaves' representation 
of~\cite{DBLP:conf/mfcs/BjorklundMS08,DBLP:journals/acta/BjorklundMS18} for our purpose. 
Let $\dpoly=\dpoly(|\w|)>\ppoly(|\w|)=\ppoly$ be a polynomial in $|\w|$ such that configurations can be encoded by a $01$-sequence of length $2^{\dpoly}$. We represent each $\lor$-configuration $c$ by the $01$-sequence\\
\centerline{
\begin{tikzpicture}[nd/.style={draw,thick,circle,inner sep=0pt,minimum size=1.5mm,fill=white},xscale=1.15,>=latex]
\node (1) at (0.5,0.15) {\scriptsize state $q$};
\node (1') at (0.5,-0.2) {\scriptsize $\log |Q|$};
\node (2) at (1.8,0.15) {\scriptsize cell content $t_1$};
\node (2b) at (1.07,0.15) {\scriptsize \textbf{0}};
\node (2') at (1.8,-0.2) {\scriptsize $\log |\Gamma|$};
\node (3) at (3.25,0.15) {\scriptsize cell content $t_2$};
\node (3b) at (2.55,0.15) {\scriptsize \textbf{0}};
\node (3') at (3.2,-0.2) {\scriptsize $\log |\Gamma|$};
\node (1) at (4.3,0.15) {\dots};
\node (4) at (5.4,0.15) {\scriptsize active cell $t_k$};
\node (4b) at (4.75,0.15) {\scriptsize \textbf{1}};
\node (2') at (5.4,-0.2) {\scriptsize $\log |\Gamma|$};
\node (1) at (6.4,0.15) {\dots};
\node (last) at (6.9,0.15) {\scriptsize \textbf{0}/\textbf{1}};
\node (0) at (0.5,0.35) {\ };
\draw[thick,gray,-] (0,0) -- (7.1,0);
\slin{0,0}{}{};
\slin{0.95,0}{}{};
\slin{2.42,0}{}{};
\slin{3.9,0}{}{};
\slin{4.65,0}{}{};
\slin{6,0}{}{};
\slin{6.7,0}{}{};
\slin{7.1,0}{}{};
\end{tikzpicture}}\\
where the last bit 
is 0 (1) iff $c$'s parent $\land$-configuration is a $0$-child ($1$-child) of its parent. (By imposing some restrictions on $Q$ and $\Gamma$, one can  ensure that, given a $2^{\dpoly}$-long $01$-sequence, it is `easy' to locate the 
first bit of each `cell-representation' in it.) 
We encode the digits of this sequence as the leaves of a $01$-tree $\gamma_c^0$ of depth $\dpoly+1$ by taking first a full binary
tree of depth $d$, and for each of its $2^{\dpoly}$ leaves, taking a $\ast$-child whenever the corresponding digit in
the sequence is $\ast$. 
(Throughout, we use $\ast$ in $01$-sequences as a wildcard for $0$ or $1$.)
Finally, we turn $\gamma_c^0$ to a $01$-tree $\gamma_c$ of depth $4\dpoly+4$
by adding an incoming edge-pattern $111$ above each node:\\
%
%
\centerline{
\setlength{\unitlength}{.05cm}
\begin{picture}(165,125)(5,-2)
\thicklines
\put(5,50){$c=0110\quad\leadsto$}
\multiput(55,35)(10,00){4}{\circle*{.5}}
\multiput(55,45)(10,00){4}{\circle*{.5}}
\multiput(60,55)(20,00){2}{\circle*{.5}}
\put(70,65){\circle*{.5}}
\multiput(55,36)(10,0){4}{\line(0,1){8}}
\multiput(55.5,46)(20,0){2}{\line(1,2){4}}
\multiput(64.5,46)(20,0){2}{\line(-1,2){4}}
\put(61,56){\line(1,1){8}}
\put(79,56){\line(-1,1){8}}
\put(50.5,38){$0$}
\put(60.5,38){$1$}
\put(70.5,38){$1$}
\put(80.5,38){$0$}
\put(52.5,48){$0$}
\put(64,48){$1$}
\put(72.5,48){$0$}
\put(84,48){$1$}
\put(59,59){$0$}
\put(78,59){$1$}
\put(67,22){$\gamma_c^0$}
\put(100,50){$\leadsto$}
\put(115,15){$\gamma_c$}
\multiput(120,90)(0,10){4}{\circle*{.5}}
\multiput(130,80)(5,-10){8}{\circle*{.5}}
\multiput(110,80)(5,-10){8}{\circle*{.5}}
\multiput(120,40)(5,-10){4}{\circle*{.5}}
\multiput(140,40)(5,-10){4}{\circle*{.5}}
\multiput(135,0)(10,0){4}{\circle*{.5}}
\multiput(135,1)(10,0){4}{\line(0,1){8}}
\multiput(134.5,11)(10,0){4}{\line(-1,2){4}}
\multiput(129.5,21)(10,0){4}{\line(-1,2){4}}
\multiput(124.5,31)(10,0){4}{\line(-1,2){4}}
\multiput(120.5,41)(20,0){2}{\line(1,2){4}}
\multiput(129.5,41)(20,0){2}{\line(-1,2){4}}
\multiput(124.5,51)(20,0){2}{\line(-1,2){4}}
\multiput(119.5,61)(20,0){2}{\line(-1,2){4}}
\multiput(114.5,71)(20,0){2}{\line(-1,2){4}}
\put(111,81){\line(1,1){8}}
\put(129,81){\line(-1,1){8}}
\multiput(120,91)(0,10){3}{\line(0,1){8}}
\put(121,113){$1$}
\put(121,103){$1$}
\put(121,93){$1$}
\put(127,84){$1$}
\put(134,74){$1$}
\put(139,64){$1$}
\put(144,54){$1$}
\put(149,44){$1$}
\put(154,34){$1$}
\put(159,24){$1$}
\put(164,14){$1$}
\put(160.5,4){$0$}
\put(114,74){$1$}
\put(119,64){$1$}
\put(124,54){$1$}
\put(129,44){$1$}
\put(134,34){$1$}
\put(139,24){$1$}
\put(144,14){$1$}
\put(140.5,4){$1$}
\put(110,84){$0$}
\put(118,44){$0$}
\put(138,44){$0$}
\put(124,34){$1$}
\put(129,24){$1$}
\put(134,14){$1$}
\put(130.5,4){$0$}
\put(144,34){$1$}
\put(149,24){$1$}
\put(154,14){$1$}
\put(150.5,4){$1$}
\end{picture}}
%
%
We call $\gamma_c$ a $c$-\emph{tree} (or, a \emph{configuration-tree\/}, in general).
%

Next, we take the full binary $01$-tree $\beta_\T^1$ above (whose nodes are labelled by $\lor$-configurations),
and turn it to a $01$-tree $\beta_\T$ (now without node labels) as follows. 
We add an incoming edge-pattern $0010$ above the root, stick the root of a $c$-tree to each node labelled by some $c$,
and add an outgoing edge-pattern $001$ below each node before branching; see Fig.~\ref{fig:comptree}.
Note that $\beta_\T$ is of depth $\dexp=\dexp(|\w|)$, for some exponential function $\dexp$. 
For any configuration $c$, if the $c$-tree $\gamma_c$ is a substructure of $\beta_{\T}$,
then we call the root node of $\gamma_c$ a \emph{\mnode{}} (\emph{of $c$}) and say that it \emph{represents $c$ in} $\beta_{\T}$; see $\bullet$-nodes in Fig.~\ref{fig:comptree}.
%

\begin{figure}[ht]
\centering
\setlength{\unitlength}{.045cm}
\begin{picture}(135,140)
\multiput(5,50)(10,0){4}{\circle*{.5}}
\multiput(10,65)(20,0){2}{\circle*{.5}}
\put(20,80){\circle*{.5}}
\put(22,80){$\cinit$}
\put(2,65){$c_2$}
\put(32,65){$c_3$}
\put(3,44){$c_6$}
\put(13,44){$c_7$}
\put(23,44){$c_8$}
\put(33,44){$c_9$}
\put(15,30){$\beta_\T^1$}
\put(9,70){$0$}
\put(28,70){$1$}
\put(3,55){$0$}
\put(14,55){$1$}
\put(22,55){$0$}
\put(34,55){$1$}
\put(55,65){\Large $\leadsto$}
\multiput(70,20)(15,0){2}{\circle*{3}}
\multiput(105,20)(15,0){2}{\circle*{3}}
\multiput(80,60)(30,0){2}{\circle*{3}}
\multiput(80,30)(0,10){3}{\circle*{.5}}
\multiput(110,30)(0,10){3}{\circle*{.5}}
\multiput(95,70)(0,10){8}{\circle*{.5}}
\put(95,100){\circle*{3}}
\put(100,80){\line(1,0){35}}
\put(100,80){\line(-1,4){5}}
\put(135,80){\line(-2,1){40}}
\put(101,85){$\gamma_{\cinit}$}
\multiput(60,40)(55,0){2}{\line(1,0){15}}
\put(60,40){\line(1,1){20}}
\put(75,40){\line(1,4){5}}
\put(115,40){\line(-1,4){5}}
\put(130,40){\line(-1,1){20}}
\put(67,44.5){$\gamma_{c_2}$}
\put(115,44.5){$\gamma_{c_3}$}
\multiput(70,0)(50,0){2}{\line(0,1){20}}
\multiput(55,0)(20,0){2}{\line(1,0){15}}
\multiput(100,0)(20,0){2}{\line(1,0){15}}
\put(55,0){\line(3,4){15}}
\put(75,0){\line(1,2){10}}
\put(90,0){\line(-1,4){5}}
\put(100,0){\line(1,4){5}}
\put(115,0){\line(-1,2){10}}
\put(135,0){\line(-3,4){15}}
\put(62,5){$\gamma_{c_6}$}
\put(80,5){$\gamma_{c_7}$}
\put(103,5){$\gamma_{c_8}$}
\put(121,5){$\gamma_{c_9}$}
\put(90,134){$0$}
\put(90,124){$0$}
\put(90,114){$1$}
\put(90,104){$0$}
\put(90,94){$0$}
\put(90,84){$0$}
\put(90,74){$1$}
\put(82,64){$0$}
\put(105,64){$1$}
\put(82,53){$0$}
\put(82,43){$0$}
\put(82,33){$1$}
\put(105,53){$0$}
\put(105,43){$0$}
\put(105,33){$1$}
\put(69,24){$0$}
\put(84,24){$1$}
\put(102,24){$0$}
\put(117,24){$1$}
\thicklines
\multiput(5.3,51)(20,0){2}{\line(1,3){4}}
\multiput(14.7,51)(20,0){2}{\line(-1,3){4}}
\put(10.5,65.75){\line(2,3){9}}
\put(29.5,65.75){\line(-2,3){9}}
\multiput(95,71)(0,10){7}{\line(0,1){8}}
\multiput(80,31)(0,10){3}{\line(0,1){8}}
\multiput(110,31)(0,10){3}{\line(0,1){8}}
\put(80,60){\line(3,2){14}}
\put(110,60){\line(-3,2){14}}
\put(70,20){\line(1,1){9}}
\put(85,20){\line(-1,2){4.5}}
\put(105,20){\line(1,2){4.5}}
\put(120,20){\line(-1,1){9}}
\put(135,25){$\beta_\T$}
\end{picture}
\caption{The $01$-tree $\beta_\T$.}\label{fig:comptree}
\end{figure}
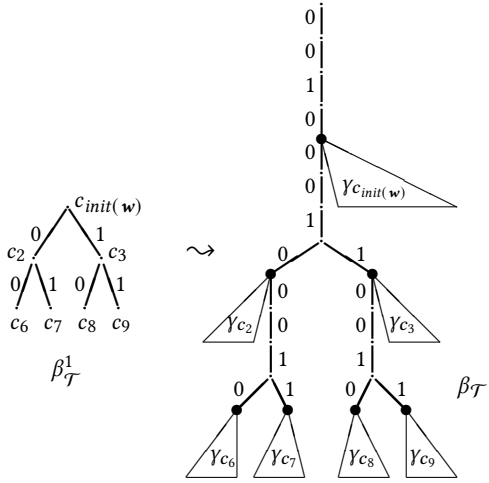

%

We also consider an infinite `version' of $\beta_\T$. 
We obtain the infinite $01$-tree $\beta_\T^+$ from $\beta_\T$ by repeatedly sticking the following pattern to the \mnode{} of each halting configuration $c$:\\
%
%
%
\centerline{
\setlength{\unitlength}{.045cm}
\begin{picture}(40,63)
\multiput(15,20)(10,0){2}{\circle*{3}}
\put(20,60){\circle*{3}}
\multiput(20,30)(0,10){3}{\circle*{.5}}
\multiput(15,0)(10,0){2}{\line(0,1){20}}
\multiput(0,0)(25,0){2}{\line(1,0){15}}
\put(0,0){\line(3,4){15}}
\put(40,0){\line(-3,4){15}}
\put(7,5){$\gamma_c$}
\put(27,5){$\gamma_c$}
\put(15,53){$0$}
\put(15,43){$0$}
\put(15,33){$1$}
\put(12,23){$0$}
\put(24,23){$1$}
\thicklines
\multiput(20,31)(0,10){3}{\line(0,1){8}}
\put(15,20){\line(1,2){4.5}}
\put(25,20){\line(-1,2){4.5}}
\end{picture}}
%
In other words, we assume $\delta$ to be such that after reaching a halting configuration $c$, $c$ is repeated forever on every branch of $\T_{\atm,\w}$.

An infinite $01$-tree $\beta$ is \emph{ideal} if it can be constructed by starting with $\beta_{\T_0}^+$, for some \otree{} $\T_0$, and then by repeatedly attaching to each of the leaves (that must be leaves of some configuration-tree) the root of some $\beta_{\T}^+$, where each $\T$ can be any computation-tree.
%

Observe that each branch of an ideal tree is infinite.
We are interested in finite `middle-bits' of ideal trees.
We call a subtree of an ideal tree having a \mnode{} (of not necessarily $\cinit$) as root a \emph{\whatisit{} tree\/}. 
Given some $M<\omega$ and a $01$-tree $\beta$, by an $M$-\emph{cut} of $\beta$ we mean
the $01$-tree obtained by cutting all longer than $M$ branches in $\beta$ at depth $M$.
%
The pretty baroque design above ensures that there is a polynomial list of polynomially detectable properties that identify 
\whatisit{} trees
up to isomorphism (see Claim~\ref{c:syc} below). In the next subsection, we discuss these properties.


\subsubsection{Characterising exponential computations polynomially}\label{ss:check}

We investigate certain polynomial neighbourhoods of nodes in $01$-trees, and collect a polynomial list of their properties that fully characterise those situations that can occur in a \whatisit{} tree.
For each of the properties $P$ below, in Sec.~\ref{s:bf}
we describe in detail how to give a small Boolean circuit $\varphi_P$ having specific \emph{\itype} such that, when evaluated at a node $\nodea$ of some $01$-tree $\beta$,
$P$ fails at $\nodea$ iff there is some $01$-sequence $\avec{b}$ such that $\avec{b}$ is gathered from the neighbourhood of $\nodea$ in $\beta$ according to the \itype{} of $\varphi_P$ and $\varphi_P[\avec{b}]=1$.

Given $n<\omega$ and a node $\nodea$ of depth $\leq n$ in a $01$-tree $\beta$, for any $k\leq n$, we denote by 
$\upp_\nodea^k$ 
the $k$-long suffix of the path ending at $\nodea$ in $\beta$.
To begin with,
observe that every path  in a \whatisit{} tree that is longer than $4\dpoly+6$ must contain a \mnode, and \mnode s can be identified by the 
property `the path leading to the node ends with a $001{\ast}$-pattern'\!. 
So,
given a node $\nodea$ in a $01$-tree $\beta$, we say that $\nodea$ \emph{is \good{} in} $\beta$,
if either the depth of $\nodea$ in $\beta$ is $<4\dpoly+11$, or 
$\upp_\nodea^{4\dpoly+11}$ contains a $001{\ast}$-pattern;
see Sec.~\ref{sgood}.

Next, we describe proper branching-patterns in a \whatisit{} tree. It is easy to see that if
the path leading to a node $\nodea$ does contain a $001{\ast}$-pattern, then 
there exist unique $k$, $\ell$ and $w$ such that $4\le k\le 4\dpoly+11$, $\upp_\nodea^k=001{\ast}(111{\ast})^\ell w$, and
either $\ell\le\dpoly$ and $w$ is a prefix of $001$, or $\ell<\dpoly$ and $w$ is a prefix of $111$.
Moreover, $\ell$ and $w$ characterise the children of $\nodea$.
We call $\nodea$ \emph{\pbr{} in} $\beta$ if the following conditions {\bf (pb1)}--{\bf (pb4)}
hold:
\begin{enumerate}
\item[{\bf (pb1)}]
if either $w$ is empty and $\ell=0$, 
or $w=001$, 
or  $w=111$ and $\ell<\dpoly-1$, then $\nodea$ has two children;

\item[{\bf (pb2)}]
if either $w$ is empty and $0<\ell<\dpoly$, 
or $w=1$, or $w=11$, or $w=00$, 
then $\nodea$ has no $0$-child;

\item[{\bf (pb3)}]
if either $w$ is empty and $\ell=\dpoly$, 
or $w=0$, 
then $\nodea$ has no $1$-child;

\item[{\bf (pb4)}]
if $w=111$ and $\ell=\dpoly-1$, 
then $\nodea$ has only one child;
\end{enumerate}
see Sec.~\ref{sbr}.
Note that leaves are never \pbr.

Next, we ensure that the `building-block' computation-trees in an ideal tree are properly represented in a $01$-tree $\beta$
(provided that all of its nodes are properly branching).
First, after each leaf of a configuration-tree, the representation of a proper computation-tree from $\cinit$ should start.
In such a case, the \mnode{} $\nodea$ of 
$\cinit$ can be identified by an incoming path ending with a $111{\ast} 001{\ast}$-pattern.
Then detecting whether the $c$-tree with root $\nodea$  is not a $\cinit$-tree requires checking polynomial information.
We call $\nodea$ \emph{\pinit{} in} $\beta$ if whenever the depth of $\nodea$ in $\beta$ is $\geq 8$,
$P_\nodea^8$ is of the form $111{\ast} 001{\ast}$, and $\nodea$ is the root of a $c$-tree, then
$c=\cinit$;
see Sec.~\ref{sinit}.
%
%
Second, the computation steps described by $\delta$ should be properly represented.
We call $\nodea$ \emph{\pcomp{} in} $\beta$ if, whenever
the following pattern is present at $\nodea$ in $\beta$\\
%
%
%
\centerline{
\setlength{\unitlength}{.045cm}
\begin{picture}(40,70)
\put(13,62){$\nodea$}
\put(25,40){\line(1,0){15}}
\put(25,40){\line(-1,4){5}}
\put(40,40){\line(-1,1){20}}
\put(25,45){$\gamma_{c}$}
\multiput(15,20)(10,0){2}{\circle*{3}}
\put(20,60){\circle*{3}}
\multiput(20,30)(0,10){3}{\circle*{.5}}
\multiput(15,0)(10,0){2}{\line(0,1){20}}
\multiput(0,0)(25,0){2}{\line(1,0){15}}
\put(0,0){\line(3,4){15}}
\put(40,0){\line(-3,4){15}}
\put(6.5,5){$\gamma_{c_0}$}
\put(27,5){$\gamma_{c_1}$}
\put(15,53){$0$}
\put(15,43){$0$}
\put(15,33){$1$}
\put(12,23){$0$}
\put(24,23){$1$}
\thicklines
\multiput(20,31)(0,10){3}{\line(0,1){8}}
\put(15,20){\line(1,2){4.5}}
\put(25,20){\line(-1,2){4.5}}
\end{picture}}
%
then the triple $(c,c_0,c_1)$ of $\lor$-configurations `matches' the transition function $\delta$ of $\atm$.
In order to detect that this is not the case, one needs to check polynomial information `around' $\nodea$ in $\beta$;
see Sec.~\ref{sdelta}.

We call $\nodea$ \emph{\syc{} in} $\beta$ if $\nodea$ is \good{}, \pbr, \pinit{} and \pcomp{} in $\beta$.
Otherwise, $\nodea$ is called \emph{\syic{} in} $\beta$.
%
Now it is straightforward to show that the collected properties of $\nodea$-neighbourhoods characterise \whatisit{} trees:

\begin{lclaim}\label{c:syc}
For any $M<\omega$, any $01$-tree $\beta$ and any node $\nodea$ with $\upp_\nodea^4=001{\ast}$,
the $M$-cut $\beta_\nodea^M$ of the subtree of $\beta$ with root $\nodea$ is isomorphic to the $M$-cut of a \whatisit{} tree iff every node of depth $<M$ in $\beta_\nodea^M$ is \syc{} in $\beta_\nodea^M$. 
\end{lclaim} 

We also need to detect the presence of nodes representing $\qreject$-configurations in computation-trees;
see Sec.~\ref{sreject}.
%
%


\subsubsection{Cactus homomorphisms vs rejecting computations}\label{s:connect}

As our \mbox{1-CQ} $\q$ will have one solitary $F$-node and two solitary $T$-nodes $\tleft$ and $\tright$, 
there are four possible kinds of non-root segments in any cactus $\C\in \mathfrak K_{\omqmx}$ denoted $\qmxfa$, $\q^-_{AT}$,  $\q^-_{TA}$ and $\q^-_{AA}$. For example, $\qmxfa$ is obtained by replacing the $F$-label of the solitary $F$-node in $\qmx$ by $A$; in cactuses different from $\q$, leaf segments  are of this form. The segment $\q^-_{TA}$ is obtained by replacing both the $F$-label of the solitary $F$-node and the $T$-label of the $\tright$-node by $A$. As $\q$ itself does not contain $A$,
if $h \colon \C \to \C'$ is a homomorphism, for some $\C,\C'\in \mathfrak K_{\omqmx}$, then the \foc{} of every
non-root
segment $\mathfrak s$ in $\C$ (labelled by $A$) is mapped by $h$ to the \foc{} of 
some non-root segment $\mathfrak s'$ in $\C'$.
Our $\q$ will also satisfy {\bf (foc)}: for every homomorphism
$h \colon \C \to \C'$ between cactuses $\C,\C'\in \mathfrak K_{\omqmx}$, 
$h$ maps the only solitary $F$-node in $\C$ (the \foc{} of its root segment) to the only solitary $F$-node in $\C'$.
So we say that \emph{$h$ maps} a segment $\mathfrak s$ into a segment $\mathfrak s'$ if $h$ maps the \foc{} of $\mathfrak s$ to the \foc{} of $\mathfrak s'$.

Now the proof of Theorem~\ref{thm:2-exp} can be completed as follows: Using Claim~\ref{c:syc}, we prove in Appendix~\ref{a:lproof}  that  to obtain Lemma~\ref{l:dagq} it suffices to construct a 1-CQ $\q$ such that {\bf (foc)} holds  and, for every $\C\in \mathfrak K_{\omqmx}$,
\begin{description}
\item[(leaf)] there is a homomorphism $h \colon \qmxfa \to \C$ mapping $\qmxfa$ into some non-leaf segment $\mathfrak s$ of $\C$ iff either 
$\mathfrak s$ is \syic{} or $\mathfrak s$ 
represents a $\qreject$-configuration in the skeleton $\C^s$ of $\C$;

\item[(\locality)] if $h$ maps $\qmxfa$ into a non-leaf segment $\mathfrak s$ that is not \pbr{} in $\C^s$ due to violating 
{\bf (pb1)},
but $\mathfrak s$ is 
\syc{} in $\C^s$ according to the other properties, then
\begin{itemize}
\item[--]
$h(\tleft)=\tleft$ and $h(\tright) \ne \tleft$, if $\mathfrak s=\q^-_{TA}$;
\item[--]
$h(\tright)=\tright$ and $h(\tleft) \ne \tright$, if $\mathfrak s=\q^-_{AT}$.
\end{itemize}
\end{description}
%
After defining the focused 1-CQ $\q$ in Secs.~\ref{qstructure}--\ref{iblocks}, we show in Sec.~\ref{s:qok} that, for every $\C\in \mathfrak K_{\omqmx}$, {\bf (leaf)} and {\bf (\locality)} are satisfied, completing the proof of Lemma~\ref{l:dagq}.


\subsection{Boolean formulas}\label{s:bf}

We describe polynomially-many polynomial-size Boolean circuits (in fact, Boolean formulas) that
test the (failure) of the properties of a node $\nodea$ in a $01$-tree $\beta$, given in Sec.~\ref{ss:check}.
For each such formula, we also define some \emph{\itype\/}, describing where around the tested node $\nodea$ the input $01$-sequence for the formula should be `gathered' from. In defining the \itype{} we use the following terminology:
for $n<\omega$, the $n$-\emph{long uppath} (\emph{of $\nodea$ in $\beta$}) is the reverse of the $n$-long suffix of the path ending at $\nodea$ in $\beta$; 
while an $n$-\emph{long downpath} is the $n$-long prefix of some path starting at $\nodea$ in $\beta$.

\subsubsection{Checking \good ness}\label{sgood}

One can clearly define a Boolean formula $\fgood(x_1,\dots,x_{4\dpoly+11})$ such that, for any $4\dpoly+11$-long $01$-sequence $\avec{b}$, $\fgood[\avec{b}]=1$ iff $\avec{b}$ does not contain the reverse of a $001{\ast}$-pattern.
The input should be gathered from the $4\dpoly+11$-long uppath.


\subsubsection{Checking proper branching-patterns}\label{sbr}

For each of conditions 
{\bf (pb1)}--{\bf (pb4)} in Sec.~\ref{ss:check}, we have a different family of formulas.

{\bf (pb1)}
For every $k$ with $4\le k\le 4\dpoly+11$,
we define a Boolean formula 
$\fbr^k(x_1,\dots,x_k)$ such that, 
for any $k$-long $01$-sequence $\avec{b}=(b_1,\dots,b_k)$, 
$\varphi^k[\avec{b}]=1$ iff 
$\avec{b}$ is the reverse of a sequence of the form $001{\ast}(111{\ast})^\ell w$, where
either $w$ is empty and $\ell=0$, or $w=001$, or  $w=111$ and $\ell<\dpoly-1$.
For example, 
if $k=4$ then
we have
\[
\fbr^{4}[\avec{b}]=1\quad  \mbox{iff}\quad \avec{b}\mbox{ is the reverse of }001{\ast}.
\]
The input should be gathered from the $k$-long uppath.

{\bf (pb2)}
For every $k$ with $4\le k\le 4\dpoly+11$, we define a Boolean formula 
$\fnobr^k_0(x_1,\dots,x_{k+1})$ such that, 
for any $k+1$-long $01$-sequence $\avec{b}=(b_1,\dots,b_{k+1})$,
$\fnobr^k_0[\avec{b}]=1$ iff $b_{k+1}=0$ and 
$(b_1,\dots,b_k)$ is the reverse of a sequence of the form $001{\ast}(111{\ast})^\ell w$, where 
either $w$ is empty and $0<\ell<\dpoly$, or $w=1$, or $w=11$, or $w=00$.
The input for $(x_1,\dots,x_k)$ should be gathered from the $k$-long uppath, and for $x_{k+1}$ from a $1$-long downpath.

{\bf (pb3)}
For every $k$ with $4\le k\le 4\dpoly+11$, we define a polynomial size Boolean formula 
$\fnobr^k_1(x_1,\dots,x_{k+1})$ such that, 
for any $k+1$-long $01$-sequence $\avec{b}=(b_1,\dots,b_{k+1})$,
$\fnobr^k_1[\avec{b}]=1$ iff $b_{k+1}=1$ and 
$(b_1,\dots,b_k)$ is the reverse of a sequence of the form $001{\ast}(111{\ast})^\ell w$, where
either $w$ is empty and $\ell=\dpoly$, or $w=0$.
The input for $(x_1,\dots,x_k)$ should be gathered from the $k$-long uppath, and for $x_{k+1}$ from a $1$-long downpath.

{\bf (pb4)}
For every $k$ with $4\le k\le 4\dpoly+11$, we define 
a polynomial size Boolean formula $\fnobr^k(x_1,\dots,x_{k+2})$ such that, 
for any $k+2$-long $01$-sequence $\avec{b}=(b_1,\dots,b_{k+2})$, 
$\fnobr^k[\avec{b}]=1$ iff $b_{k+1}\ne b_{k+2}$ and 
$(b_1,\dots,b_k)$ is the reverse of a sequence of the form $001{\ast}(111{\ast})^\ell w$, where
$w=111$ and $\ell=\dpoly-1$.
The input for $(x_1,\dots,x_k)$ should be gathered from the $k$-long uppath, and for each of $x_{k+1}$ and $x_{k+2}$ from a $1$-long downpath.


\subsubsection{Checking proper computation steps}\label{sdelta}


This is an adaptation of the technique 
of ~\cite{DBLP:conf/mfcs/BjorklundMS08,DBLP:journals/acta/BjorklundMS18} to our representation. Suppose that
\begin{itemize}
\item[--]
$n_Q$ is such that $n_Q$-long $01$-sequences are in one-to-one correspondence with the states in $Q$,
\item[--]
$n_\Gamma$ is such that $(n_\Gamma-1)$-long $01$-sequences are in one-to-one correspondence with the symbols in $\Gamma$,
\item[--]
$n_Q+2^{\ppoly}\cdot n_\Gamma+1=2^{\dpoly}$
\end{itemize}
(see the picture in Sec.~\ref{comptrees} on representing configurations with $2^{\dpoly}$-long $01$-sequences).

First, we define a Boolean formula $\fhead(x_1,\dots,x_{4(\dpoly+1)})$
such that, for any $4(\dpoly+1)$-long $01$-sequence $\avec{b}$, 
$\fhead[\avec{b}]=1$ iff $\avec{b}$ describes a path in a \whatisit{} tree starting at a \mnode{} with $1$,
and ending at the first bit of the representation of some cell-content of some configuration $c$
(that is, it is the $(n_Q+i\cdot n_\Gamma+1)$th bit of the $01$-sequence representing $c$, for some $i<2^{\ppoly}$).
Similarly, for ${\ast}=0,1$, we define a Boolean formula $\fhead^\ast(x_1,\dots,x_{4(\dpoly+1)+4})$
such that, for any $4(\dpoly+1)+4$-long $01$-sequence $\avec{b}$, 
$\fhead^\ast[\avec{b}]=1$ iff $\avec{b}$ describes a path in a \whatisit{} tree starting at a \mnode{} with $001{\ast}1$,
and ending at the first bit of the representation of some cell-content of some configuration.

Next,  we define a Boolean formula 
\[
\fscell(x_1,\dots,x_{4(\dpoly+1)},y_1,\dots,y_{4(\dpoly+1)+4},z_1,\dots,z_{4(\dpoly+1)+4})
\]
such that, for any $4(\dpoly+1)$-long $01$-sequence $\avec{b}$ and $4(\dpoly+1)+4$-long $01$-sequences $\avec{b}^0$, 
$\avec{b}^1$, we have 
$\fscell[\avec{b},\avec{b}^0,\avec{b}^1]=1$ iff $\fhead[\avec{b}]=1$, $\fhead^\ast[\avec{b}^\ast]=1$, and the three paths $\avec{b},\avec{b}^0,\avec{b}^1$ end at the same cell of a configuration and its two children-configurations. (The number $i$ of this cell, for some $i<2^{\ppoly}$, can be identified from $\avec{b}$.)
 
Next,  we define a Boolean formula $\fstate(x_1,\dots,x_{4(\dpoly+1)\cdot n_Q})$ such that, for any $4(\dpoly+1)$-long $01$-sequences $\avec{b}_1,\dots,\avec{b}_{n_Q}$, 
\[
\fstate[\avec{b}_1,\dots,\avec{b}_{n_Q}]=1
\]
 iff for every $j\le n_Q$, $\avec{b}_j$ describes a $4(\dpoly+1)$-long path in a \whatisit{} tree starting at a \mnode{} with $1$,
and ending at the $j$th bit of the representation of some configuration $c$.
(So, whenever $\avec{b}_j=(b_j^1,\dots,b_j^{4(\dpoly+1)})$ for each $j\le n_Q$, then 
$(b_1^{4(\dpoly+1)},\dots,b_{n_Q}^{4(\dpoly+1)})$ encodes a state in $Q$.)
Similarly, for ${\ast}=0,1$, we define a Boolean formula $\fstate^\ast(x_1,\dots,x_{(4(\dpoly+1)+4)\cdot n_Q})$
such that, for any $4(\dpoly+1)+4$-long $01$-sequences $\avec{b}_1\dots,\avec{b}_{n_Q}$, 
$\fstate^\ast[\avec{b}_1\dots,\avec{b}_{n_Q}]=1$ iff for every $j\le n_Q$, $\avec{b}_j$ describes a path in a \whatisit{} tree starting at a \mnode{} with $001{\ast}1$,
and ending at the $j$th bit of the representation of some configuration.

Next,  we define a Boolean formula $\fcell(x_1,\dots,x_{4(\dpoly+1)\cdot n_\Gamma})$
such that, for any $4(\dpoly+1)$-long $01$-sequences $\avec{b}_1\dots,\avec{b}_{n_\Gamma}$, 
\[
\fcell[\avec{b}_1\dots,\avec{b}_{n_\Gamma}]=1
\]
 iff there is $i<2^{\ppoly}$ such that, for every $j\le n_\Gamma$, $\avec{b}_j$ describes a path in a \whatisit{} tree starting at a \mnode{} with $1$,
and ending at the $j$th bit of the representation of the $i$th cell's content in some configuration.
(So, $\fhead[\avec{b}_1]=1$, and
if $\avec{b}_j=(b_j^1,\dots,b_j^{4(\dpoly+1)})$ for each $j\le n_\Gamma$, then 
$(b_2^{4(\dpoly+1)},\dots,b_{n_\Gamma}^{4(\dpoly+1)})$ encodes a symbol in $\Gamma$.)
Similarly, for ${\ast}=0,1$, we define a polynomial size Boolean formula 
$\fcell^\ast(x_1,\dots,x_{(4(\dpoly+1)+4)\cdot n_\Gamma})$
such that, for any $4(\dpoly+1)+4$-long $01$-sequences $\avec{b}_1\dots,\avec{b}_{n_\Gamma}$, 
$\fcell^\ast[\avec{b}_1\dots,\avec{b}_{n_\Gamma}]=1$ iff there is $i<2^{\ppoly}$ such that,  for every $j\le n_\Gamma$, $\avec{b}_j$ describes a path in a \whatisit{} tree starting at a \mnode{} with $001{\ast}1$,
and ending at the $j$th bit of the representation of the $i$th cell's content in some configuration
(In particular, $\fhead^\ast[\avec{b}_1]=1$.)

Next, for $z\in\{0,1\}$, we take the following tuples of variables:
\begin{itemize}
\item[--]
$\avec{s}=(s_1,\dots,s_{4(\dpoly+1)\cdot n_Q})$, which is to be gathered from
$n_Q$-many \mbox{$4(\dpoly+1)$}-long downpaths (representing the $\lor$-state in $c$);

\item[--]
$\avec{v}=(v_1,\dots,v_{4(\dpoly+1)\cdot n_\Gamma})$, which needs to be gathered from $n_\Gamma$-many $4(\dpoly+1)$-long downpaths (representing the active cell's content in $c$);

\item[--]
$\avec{s}^0=(s_1^0,\dots,s_{(4(\dpoly+1)+4)\cdot n_Q}^0)$, $\avec{s}^1=(s_1^1,\dots,s_{(4(\dpoly+1)+4)\cdot n_Q}^1)$, each of which needs to be gathered from $n_Q$-many $4(\dpoly+1)+4$-long downpaths (representing the $\lor$-states in $c_0,c_1$);

\item[--]
$\avec{t}^z=(t_1^z,\dots,t_{4(\dpoly+1)\cdot n_\Gamma}^z)$, $\avec{t}^{z0}=(t_1^{z0},\dots,t_{(4(\dpoly+1)+4)\cdot n_\Gamma}^{z0})$, and $\avec{t}^{z1}=(t_1^{z1},\dots,t_{(4(\dpoly+1+4))\cdot n_\Gamma}^{z1})$, for $z\in\{0,1\}$, where
$\avec{t}^z$ is to be gathered from $n_\Gamma$-many $4(\dpoly+1)$-long downpaths, and each of $\avec{t}^{z0}$ and $\avec{t}^{z1}$
is to be gathered from $n_\Gamma$-many $4(\dpoly+1)+4$-long downpaths
($\avec{t}^z$, $\avec{t}^{z0}$, $\avec{t}^{z1}$ represent the $i$th cell's contents in $c,c_0,c_1$, for some $i<2^{\ppoly}$, when the $z$ $\land$-child of $c$ is taken in $\T_{\atm,x}$);

\item[--]
$\avec{z}^0=(z_1^0,\dots,z_{4(\dpoly+1)+4}^0)$ and $\avec{z}^1=(z_1^1,\dots,z_{4(\dpoly+1)+4}^1)$, each of which needs to be gathered from
a $4(\dpoly+1)+4$-long downpath (representing the respective bits identifying the parent $\land$-configuration of $c_0$ and $c_1$).
\end{itemize}
For $z\in\{0,1\}$,
we can define a Boolean formula $\fdelta^z$ such that,
 for any $01$-sequence
$\avec{b}=(\avec{s},\avec{v},\avec{s}^0,\avec{s}^1,\avec{t},\avec{t}^{0},\avec{t}^{1},\avec{z}^0,\avec{z}^1)$,
$\fdelta^z[\avec{b}]=1$ iff 
$\avec{z}^0=001011\dots 1z$, $\avec{z}^1=001111\dots 1z$,  
$\fstate[\avec{s}]=1$, 
$\fcell[\avec{v}]=1$, 
$\fstate^0[\avec{s}^0]=1$, 
$\fstate^1[\avec{s}^1]=1$, 
$\fcell[\avec{t}]=1$, 
$\fcell^0[\avec{t}^{0}]=1$, 
$\fcell^1[\avec{t}^{1}]=1$, and
\[
\fscell[t_1,\dots,t_{4(\dpoly+1)},t_1^{0},\dots,t_{4(\dpoly+1)+4}^{0},t_1^{1},\dots,t_{4(\dpoly+1)+4}^{1}]=1,
\]
%
%
but the information provided by $\avec{b}$ is inconsistent with the transition function $\delta$ in the sense that 
when the $z$ $\land$-child of $c$ is chosen  as the common parent of $c_0$ and $c_1$ in the computation-tree $\T$, the content-triple of the $i$th cells of $c,c_0,c_1$ is wrong, where $i$ is identified from the input in 
$(t_1,\dots,t_{4(\dpoly+1)},t_1^{0},\dots,t_{4(\dpoly+1)+4}^{0},t_1^{1},\dots,t_{4(\dpoly+1)+4}^{1})$.

Finally, we define
$\fdelta(\avec{s},\avec{v},\avec{s}^0,\avec{s}^1,\avec{t}^0,\avec{t}^{00},\avec{t}^{01},\avec{t}^1,\avec{t}^{10},\avec{t}^{11},\avec{z}^0,\avec{z}^1)$
as the disjunction 
\begin{multline*}
\fdelta^0(\avec{s},\avec{v},\avec{s}^0,\avec{s}^1,\avec{t}^0,\avec{t}^{00},\avec{t}^{01},\avec{z}^0,\avec{z}^1)\;\lor\\
\fdelta^1(\avec{s},\avec{v},\avec{s}^0,\avec{s}^1,\avec{t}^1,\avec{t}^{10},\avec{t}^{11},\avec{z}^0,\avec{z}^1).
\end{multline*}
It is not hard to see that there is $\avec{b}$ with $\fdelta[\avec{b}]=1$ iff the information about the configuration-triple $(c,c_0,c_1)$ encoded in $\avec{b}$ is inconsistent with $\delta$.


\subsubsection{Checking proper initialisation}\label{sinit}

We take the following tuples of variables:
\begin{itemize}
\item[--]
$\avec{y}=(y_1,\dots,y_8)$, which is to be gathered from the $8$-long uppath (representing the last $8$-bits of the path leading to the \mnode{} of a configuration $c$);

\item[--]
$\avec{s}=(s_1,\dots,s_{4(\dpoly+1)\cdot n_Q})$, to be gathered from $n_Q$-many $4(\dpoly+1)$-long downpaths (representing the state in $c$);
\item[--] 
$\avec{w}^j=(w_1^j,\dots,w_{4(\dpoly+1)\cdot n_\Gamma}^j)$, for $j<|\w|$,  each of which needs to be gathered from $n_\Gamma$-many $4(\dpoly+1)$-long downpaths
(representing the contents of the first $|\w|$-many cells in $c$);

\item[--]
$\avec{t}=(t_1,\dots,t_{4(\dpoly+1)\cdot n_\Gamma})$, which needs to be gathered from $n_\Gamma$-many $4(\dpoly+1)$-long downpaths (representing the contents of some cell in $c$).
\end{itemize}
%
Then we can define  a Boolean formula $\finit$ such that,
 for any $01$-sequence $\avec{b}=(\avec{y},\avec{s},\avec{w}^1,\dots,\avec{w}^{|\w|},\avec{t})$,
 $\finit[\avec{b}]=1$ iff 
 $\avec{y}$ is the reverse of some pattern $111{\ast} 001{\ast}$,
 $\fstate[\avec{s}]=1$, 
  for all $1\le j\le |x|$, 
 $\fcell[\avec{w}^j]=1$ and $(w_1^j,\dots,w_{4(\dpoly+1)}^j)$ ends at the $(n_Q+(j-1)\cdot n_\Gamma+1)$th bit of the $01$-sequence representing configurations, 
$\fcell[\avec{t}]=1$,
 but the information provided by $\avec{b}$ is inconsistent with $c$ being $\cinit$ (at the cell identified by the prefix 
 $(t_1,\dots t_{4(\dpoly+1)})$ of $\avec{t}$). 


\subsubsection{Representing $\qreject$-configurations}\label{sreject}

This can clearly be done by a formula $\freject(s_1,\dots,s_{4(\dpoly+1)\cdot n_Q})$ for which 
$\freject[\avec{s}]=1$ iff $\fstate[\avec{s}]=1$ and the sequence 
\[
(s_{4(\dpoly+1)},s_{4(\dpoly+1)\cdot 2},...,s_{4(\dpoly+1)\cdot n_Q})
\]
encodes $\qreject$.
The input should be gathered from $n_Q$-many $4(\dpoly+1)$-long downpaths.


\subsection{Query design}\label{qdesign}

The dag-shaped $1$-CQ $\q$ having one solitary $F$-node, two solitary $T$-nodes and many $FT$-twins will be such that 
properties {\bf (foc)}, {\bf (leaf)} and {\bf (\locality)} given in Sec.~\ref{s:connect} hold, for all $\C,\C'\in \mathfrak K_{\omqmx}$.


\subsubsection{Overall query structure}\label{qstructure}

\begin{figure*}[ht]
   \centering
%
%
%
\begin{tikzpicture}[decoration={brace,mirror,amplitude=7},line width=1pt,scale =.54]
\node[point,scale=0.6,label = left:$F$] (1) at (0,0) {};
\node[point,scale=0.6,label = left:$\xi$] (2) at (0,-2) {};
\node[point,scale=0.6,label = above left:$\alpha$] (3) at (0,-4) {};
\node[point,scale=0.6] (4) at (-.5,-6) {};
\node[point,scale=0.6,label = left:$\tleft$,label = below:$T$] (5) at (-1,-8) {};
\node[point,scale=0.6] (6) at (.5,-6) {};
\node[point,scale=0.6,label = right:$\tright$,label = below:$T$] (7) at (1,-8) {};
\node[point,scale=0.6,label = above:$W$] (8) at (2,0) {};
\node[point,scale=0.6,label = above:$\ \ \xi'$] (9) at (4,-2) {};
\draw[->,line width=.4pt] (1) to (2);
\draw[->,line width=.4pt] (2) to (3);
\draw[->,line width=.4pt] (3) to (4);
\draw[->,line width=.4pt,bend left=10] (4) to (5);
\draw[->,line width=.4pt,bend right=20] (4) to node[left] {$S$}  (5);
\draw[->,line width=.4pt,bend left=20] (3) to node[right] {$S$}  (6);
\draw[->,line width=.4pt,bend right=10] (3) to (6);
\draw[->,line width=.4pt] (6) to (7);
\draw[->,line width=.4pt] (1) to (8);
\draw[->,line width=.4pt] (2) to (8);
\draw[->,line width=.4pt,bend right=10] (9) to (8);
\draw[->,line width=.4pt] (1) to (9);
\node (m1) at (1.5,-3) {};
\draw[->] (m1) to node[right] {$R_\ga$}  (3);
\draw[ultra thick,rounded corners=9] (2) -- (1.5,-1.3) -- (2.5,-2) -- (2,-3) -- (m1) -- (.3,-2.6) -- (2); 
\node[point,fill=white,scale=0.6,label = right:$\ \rnode$] (m1) at (1.5,-3) {};
\node[]  at (1.3,-2.3) {$\fga$};
\node (m2) at (5.5,-3) {};
\draw[ultra thick,rounded corners=9] (9) -- (5.5,-1.3) -- (6.5,-2) -- (6,-3) -- (m2) -- (4.3,-2.6) -- (9); 
\node[point,fill=white,scale=0.6,label = right:$\ \rnode'$] (m1) at (5.5,-3) {};
\node[]  at (5.3,-2.3) {$\fga'$};
\node[point,scale=0.6,label = above left:$\fnode$] (g3) at (4,-4) {};
\draw[->] (m2) to node[right] {$R_\ga$}  (g3);
\draw[->] (9) to (g3);
\node[point,scale=0.6] (g4) at (2.5,-6) {};
\node[point,scale=0.6] (g6) at (4.5,-6) {};
\node[point,scale=0.6,label = below:$FT$] (g7) at (5,-8) {};
\node[point,scale=0.6,label = below:$\Unode$] (i4) at (2,-4.5) {};
\draw[->] (3) to (i4);
\draw[->] (g3) to (i4);
\draw[->,bend right=10] (g3) to (g4);
\draw[->,bend left=20] (g3) to node[right] {$S$}  (g4);
\draw[->] (g4) to (7);
\draw[->] (g3) to (g6);
\draw[->,bend right=10] (g6) to (g7);
\draw[->,bend left=20] (g6) to node[right] {$S$}  (g7);
\node[point,scale=0.6,label = below:$\Unode$] (i3) at (-1.5,-4.5) {};
\node[point,scale=0.6,label = below left:$\inode\!\!$] (i2) at (-3,-4) {};
\node (i1) at (-3,-2) {};
\draw[->] (i1) to node[right] {$\!R_\ga$}  (i2);
\draw[->] (i2) to (i3);
\draw[->] (3) to (i3);
\draw [ultra thick] (-3,-1.2) circle [radius=.8];
\node[]  at (-3,-1.2) {$\iga$};
\node[point,fill=white,scale=0.6,label = below left:$\pnode\!\!$] (i1) at (-3,-2) {};
\node[]  at (-3,-1.2) {$\iga$};
\draw [thin,dashed] (-4,-10.2) rectangle (6.6,1.8);
\node[]  at (-2.5,1) {type $AT$};
\end{tikzpicture}
\begin{tikzpicture}[decoration={brace,mirror,amplitude=7},line width=1pt,scale =.54]
\node[point,scale=0.6,label = left:$F$] (1) at (0,0) {};
\node[point,scale=0.6,label = left:$\xi$] (2) at (0,-2) {};
\node[point,scale=0.6,label = above left:$\alpha$] (3) at (0,-4) {};
\node[point,scale=0.6] (4) at (-.5,-6) {};
\node[point,scale=0.6,label = left:$\tleft$,label = right:$\!T$] (5) at (-1,-8) {};
\node[point,scale=0.6] (6) at (.5,-6) {};
\node[point,scale=0.6,label = right:$\tright$,label = left:$T\!$] (7) at (1,-8) {};
\node[point,scale=0.6,label = above:$W$] (8) at (2,0) {};
\node[point,scale=0.6,label = above:$\ \ \xi'$] (9) at (4,-2) {};
\draw[->,line width=.4pt] (1) to (2);
\draw[->,line width=.4pt] (2) to (3);
\draw[->,line width=.4pt] (3) to (4);
\draw[->,line width=.4pt,bend left=10] (4) to (5);
\draw[->,line width=.4pt,bend right=20] (4) to node[left] {$S$}  (5);
\draw[->,line width=.4pt,bend left=20] (3) to node[right] {$S$}  (6);
\draw[->,line width=.4pt,bend right=10] (3) to (6);
\draw[->,line width=.4pt] (6) to (7);
\draw[->,line width=.4pt] (1) to (8);
\draw[->,line width=.4pt] (2) to (8);
\draw[->,line width=.4pt,bend right=10] (9) to (8);
\draw[->,line width=.4pt] (1) to (9);
\node (m1) at (1.5,-3) {};
\draw[->] (m1) to node[right] {$R_\ga$}  (3);
\draw[ultra thick,rounded corners=9] (2) -- (1.5,-1.3) -- (2.5,-2) -- (2,-3) -- (m1) -- (.3,-2.6) -- (2); 
\node[point,fill=white,scale=0.6,label = right:$\ \rnode$] (m1) at (1.5,-3) {};
\node[]  at (1.3,-2.3) {$\fga$};
\node (m2) at (5.5,-3) {};
\draw[ultra thick,rounded corners=9] (9) -- (5.5,-1.3) -- (6.5,-2) -- (6,-3) -- (m2) -- (4.3,-2.6) -- (9); 
\node[point,fill=white,scale=0.6,label = right:$\ \rnode'$] (m1) at (5.5,-3) {};
\node[]  at (5.3,-2.3) {$\fga'$};
\node[point,scale=0.6,label = above left:$\fnode$] (g3) at (4,-4) {};
\draw[->] (m2) to node[right] {$R_\ga$}  (g3);
\draw[->] (9) to (g3);
\node[point,scale=0.6] (g4) at (3,-8) {};
\node[point,scale=0.6] (g6) at (4.5,-6) {};
\node[point,scale=0.6,label = below:$FT$] (g7) at (5,-8) {};
\node[point,scale=0.6,label = below:$\Unode$] (i4) at (2,-4.5) {};
\draw[->] (3) to (i4);
\draw[->] (g3) to (i4);
\draw[->] (g3) to (g4);
\draw[->,bend right=10] (g3) to (g6);
\draw[->,bend left=20] (g3) to node[right] {$S$} (g6);
\draw[->] (g6) to (g7);
\node[point,scale=0.6,label = below:$\Unode$] (i3) at (-1.5,-4.5) {};
\node[point,scale=0.6,label = below left:$\inode\!\!$] (i2) at (-3,-4) {};
\node (i1) at (-3,-2) {};
\draw[->] (i1) to node[right] {$\!R_\ga$}  (i2);
\draw[->] (i2) to (i3);
\draw[->] (3) to (i3);
\draw [ultra thick] (-3,-1.2) circle [radius=.8];
\node[]  at (-3,-1.2) {$\iga$};
\node[point,fill=white,scale=0.6,label = below left:$\pnode\!\!$] (i1) at (-3,-2) {};
\draw[->,bend left=50] (g4) to (5);
\draw[->,bend left=90] (g4) to node[below] {$S$} (5);
\draw [thin,dashed] (-4,-10.2) rectangle (6.6,1.8);
\node[]  at (-2.5,1) {type $TA$};
\end{tikzpicture}
\begin{tikzpicture}[decoration={brace,mirror,amplitude=7},line width=1pt,scale =.55]
\node[point,scale=0.6,label = left:$F$] (1) at (0,0) {};
\node[point,scale=0.6,label = left:$\xi$] (2) at (0,-2) {};
\node[point,scale=0.6,label = above left:$\alpha$] (3) at (0,-4) {};
\node[point,scale=0.6] (4) at (-.5,-6) {};
\node[point,scale=0.6,label = left:$\tleft$,label = below:$T$] (5) at (-1,-8) {};
\node[point,scale=0.6] (6) at (.5,-6) {};
\node[point,scale=0.6,label = right:$\tright$,label = below:$T$] (7) at (1,-8) {};
\node[point,scale=0.6,label = above:$W$] (8) at (2,0) {};
\node[point,scale=0.6,label = above:$\ \ \xi'$] (9) at (4,-2) {};
\draw[->,line width=.4pt] (1) to (2);
\draw[->,line width=.4pt] (2) to (3);
\draw[->,line width=.4pt] (3) to (4);
\draw[->,line width=.4pt,bend left=10] (4) to (5);
\draw[->,line width=.4pt,bend right=20] (4) to node[left] {$S$}  (5);
\draw[->,line width=.4pt,bend left=20] (3) to node[right] {$S$}  (6);
\draw[->,line width=.4pt,bend right=10] (3) to (6);
\draw[->,line width=.4pt] (6) to (7);
\draw[->,line width=.4pt] (1) to (8);
\draw[->,line width=.4pt] (2) to (8);
\draw[->,line width=.4pt,bend right=10] (9) to (8);
\draw[->,line width=.4pt] (1) to (9);
\node (m1) at (1.5,-3) {};
\draw[->] (m1) to node[right] {$R_\ga$}  (3);
\draw[ultra thick,rounded corners=9] (2) -- (1.5,-1.3) -- (2.5,-2) -- (2,-3) -- (m1) -- (.3,-2.6) -- (2); 
\node[point,fill=white,scale=0.6,label = right:$\ \rnode$] (m1) at (1.5,-3) {};
\node[]  at (1.3,-2.3) {$\fga$};
\node (m2) at (5.5,-3) {};
\draw[ultra thick,rounded corners=9] (9) -- (5.5,-1.3) -- (6.5,-2) -- (6,-3) -- (m2) -- (4.3,-2.6) -- (9); 
\node[point,fill=white,scale=0.6,label = right:$\ \rnode'$] (m1) at (5.5,-3) {};
\node[]  at (5.3,-2.3) {$\fga'$};
\node[point,scale=0.6,label = above left:$\fnode$] (g3) at (4,-4) {};
\draw[->] (m2) to node[right] {$R_\ga$}  (g3);
\draw[->] (9) to (g3);
\node[point,scale=0.6] (g6) at (4.5,-6) {};
\node[point,scale=0.6,label = below:$FT$] (g7) at (5,-8) {};
\node[point,scale=0.6,label = below:$\Unode$] (i4) at (2,-4.5) {};
\draw[->] (3) to (i4);
\draw[->] (g3) to (i4);
\draw[->,bend right=10] (g3) to (g6);
\draw[->,bend left=20] (g3) to node[right] {$S$}  (g6);
\draw[->,bend right=10] (g6) to (g7);
\draw[->,bend left=20] (g6) to node[right] {$S$}  (g7);
\node[point,scale=0.6,label = below:$\Unode$] (i3) at (-1.5,-4.5) {};
\node[point,scale=0.6,label = below left:$\inode\!\!$] (i2) at (-3,-4) {};
\node (i1) at (-3,-2) {};
\draw[->] (i1) to node[right] {$\!R_\ga$}  (i2);
\draw[->] (i2) to (i3);
\draw[->] (3) to (i3);
\draw [ultra thick] (-3,-1.2) circle [radius=.8];
\node[]  at (-3,-1.2) {$\iga$};
\node[point,fill=white,scale=0.6,label = below left:$\pnode\!\!$] (i1) at (-3,-2) {};
\draw [thin,dashed] (-4,-10.2) rectangle (6.6,1.6);
\node[]  at (-2.5,1) {type $AA$};
\end{tikzpicture}
   \caption{Frames of type $AT$, $TA$ and $AA$.}
   \label{fig:gadgets}
\end{figure*}
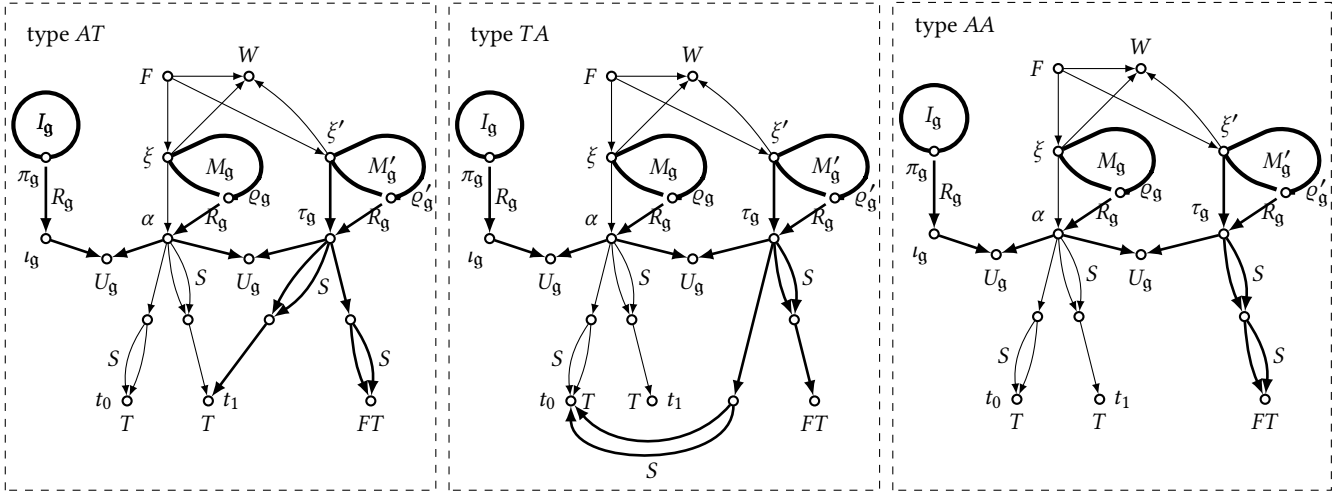

To simplify notation, in our pictures we omit the $R$-labels from $R$-arrows, and use extra labels (different from $F$, $T$) on nodes, say $B$ on $a$, as a shorthand for a $B$-arrow $(a,a')$ to a \emph{fresh} node $a'$. Letters other than upper case italics (greek, lower case italics and bold) are used as pointers and are not part of $\qmx$.

The 1-CQ $\qmx$ has the following simple \emph{base block} containing all of the solitary $F$- and $T$-nodes of $\qmx$:\\
%
%
\centerline{
\begin{tikzpicture}[decoration={brace,mirror,amplitude=7},line width=1pt,scale =.57]
\node[point,scale=0.6,label = left:$F$] (1) at (0,0) {};
\node[point,scale=0.6,label = left:$\xi$] (2) at (0,-2) {};
\node[point,scale=0.6,label = left:$\alpha$] (3) at (0,-4) {};
\node[point,scale=0.6] (4) at (-.5,-6) {};
\node[point,scale=0.6,label = left:$\tleft$,label = below:$T$] (5) at (-1,-8) {};
\node[point,scale=0.6] (6) at (.5,-6) {};
\node[point,scale=0.6,label = right:$\tright$,label = below:$T$] (7) at (1,-8) {};
\node[point,scale=0.6,label = above:$W$] (8) at (2,0) {};
\node[point,scale=0.6,label = right:$\xi'$] (9) at (4,-2) {};
\draw[->] (1) to (2);
\draw[->] (2) to (3);
\draw[->] (3) to (4);
\draw[->,bend left=10] (4) to (5);
\draw[->,bend right=20] (4) to node[left] {$S$}  (5);
\draw[->,bend left=20] (3) to node[right] {$S$}  (6);
\draw[->,bend right=10] (3) to (6);
\draw[->] (6) to (7);
\draw[->] (1) to (8);
\draw[->] (2) to (8);
\draw[->,bend right=10] (9) to (8);
\draw[->] (1) to (9);
\end{tikzpicture}}
%
Wired to the base in $\qmx$ are \emph{gadgets} $\ga$ that  \emph{implement} the Boolean formulas $\foga$ defined in Sec.~\ref{s:bf}. 
Each gadget $\ga$ has four components: two isomorphic copies of its \emph{\formulablock} $\fga$ and $\fga'$, an \emph{input block} $\iga$ and a \emph{frame}. The frame wires the gadget to the base and can be of one of the three \emph{types} $AT$, $TA$ and $AA$, which are shown in Fig.~\ref{fig:gadgets} 
(with the base block being indicated in each case by thin lines). We say that $\ga$ is of type $Z$ if its frame is of type $Z$.
The frame of $\ga$ has a few distinguished nodes:
$\pnode$ and $\inode$ via which $\iga$ is $R_\ga$-wired to the base block,
$\rnode$ via which $\fga$ is $R_\ga$-wired to the base block,
$\rnode'$ and $\fnode$ via which $\fga'$ is $R_\ga$-wired to the base block 
(where the edge-labelling $R_\ga$ is also unique for gadget $\ga$),
the single $FT$-twin of $\ga$ (none of $\iga$, $\fga$ and $\fga'$ contains any $FT$-twins), 
and two nodes labelled by $\Unode$.

It is easy to see that $\q$ satisfies {\bf (foc)}: its
$F$-node has successors, while none of the $FT$-nodes does. Further, we
observe that if $h \colon \qmxfa \to \C$  is a homomorphism mapping $\qmxfa$ into some non-leaf segment $\mathfrak s$, then there must be a gadget $\ga$ such that $h(\midnode)=\fnode$, for the $\midnode$-node in $\qmxfa$ and
the $\fnode$-node in $\mathfrak s$. Then, because of the $U_\ga$-nodes, 
$h(\inode)=\alpha$ must hold, for the $\inode$-node in $\qmxfa$ and the
$\alpha$-node in $\mathfrak s$. Therefore, $h(\pnode)=\rnode$ and the $\iga$-block of $\qmxfa$ must be mapped by $h$ to the $\fga$-block of $\mathfrak s$, forcing the input to interact with the formula.

Given a gadget $\ga$ and a homomorphism $h \colon \qmxfa \to \C$ mapping $\qmxfa$ into a non-leaf segment $\mathfrak s$ of some cactus $\C$,
we say that $\ga$ \emph{is triggered by $h$ at $\mathfrak s$} if $h(\inode)=\midnode$,  for the $\inode$-node in $\qmxfa$ and
the $\midnode$-node in $\mathfrak s$. 
We say that $\ga$ \emph{is triggered at $\mathfrak s$} if there is a homomorphism $h$ triggering $\ga$ at $\mathfrak s$. 
%
%
%
%
Observe that 
if $\mathfrak s$ is of the form $\q^-_Z$, for some $Z\in\{AT,TA,AA\}$, and $\ga$ is triggered at $\mathfrak s$, then $\ga$ is either of type $AA$ or of type $Z$.

Each gadget $\ga$ in $\q$ `implements' some Boolean formula $\foga(\avec{y})$ checking some property of
\whatisit{} trees at node $\mathfrak s$ in the skeleton $01$-tree $\C^s$ of the cactus $\C$. 
So the input values for the variables in $\avec{y}$ are `collected' from an environment of $\mathfrak s$
in $\C^s$. This collection process is regulated by the 
\itype{} of each $\foga$.
We have the following gadgets in $\qmx$,
each implementing some formula described in Sec.~\ref{s:bf}:
\begin{enumerate}
\item[{\bf (g1)}]
a type $AA$ gadget implementing $\fgood$;
\item[{\bf (g2)}]
for every $k$ with $4\le k\le 4\dpoly+11$ and every $Z\in\{AT,TA\}$, a type $Z$ gadget implementing $\fbr^k$;

\item[{\bf (g3)}]
for every $k$ with $4\le k\le 4\dpoly+11$ and every ${\ast}\in\{0,1\}$, a type $AA$-gadget implementing $\fnobr^k_\ast$;
\item[{\bf (g4)}]
for every $k$ with $4\le k\le 4\dpoly+11$, a type $AA$-gadget implementing $\fnobr^k$;

\item[{\bf (g5)}]
a type $AA$ gadget implementing $\fdelta$;

\item[{\bf (g6)}]
a type $AA$ gadget implementing $\finit$;

\item[{\bf (g7)}]
a type $AA$ gadget implementing $\freject$.
\end{enumerate}
%
%

Now suppose $\ga_1,\dots,\ga_m$ are all of the gadgets in $\qmx$.
We want to ensure that when a gadget $\ga_j$  is triggered by some $h$, then the other gadgets are not triggered (that is, the $\inodei$-node for every $i\ne j$ can be mapped by $h$ to itself). So, in addition to the above,
for every $j$, we connect the $\inodej$-node via an $\Unodej$-labelled node to the $\fnodei$-node, for all $i\ne j$.
%
We also want to ensure that when $\ga_j$  is triggered by some $h$, then 
the $\fgai$-block can be $h$-mapped to the $\fgai'$-block for every $i$ (not just for $j$).
So we not only $R_{\ga_j}$-connect $\rnodej'$ with $\fnodej$, but also add $R_{\ga_j}$-arrows connecting $\rnodej'$ with all of the $\fnodei$:

\centerline{
   \begin{tikzpicture}[decoration={brace,mirror,amplitude=7},line width=1pt,scale =.55]
   \node[point,scale=0.6,label = left:$\inodej$] (1) at (1,0) {};
   \node[point,scale=0.6,label = above:$\fnodeone\quad$] (2) at (4,0) {};
   \node[point,scale=0.6,label = below:$\fnodei$] (3) at (7,0) {};
   \node[point,scale=0.6,label = below:$\fnodej$] (33) at (10,0) {};
   \node[point,scale=0.6,label = below:$\fnodem$] (4) at (13,0) {};
   \node[point,scale=0.6,label = below:$\Unodej$] (5) at (2.5,-1) {};
   \node[point,scale=0.6,label = left:$\xi'$] (6) at (6,4) {}; 
   \node[point,scale=0.6,label = above:$\rnodej'$] (7) at (11.5,2.5) {};
   \draw[->] (2) to (5);
   \draw[->,bend left=10] (3) to (5);
   \draw[->,bend left=25] (4) to (5);
   \draw[->] (7) to node[left] {$R_{\ga_j}\quad\ \ \ $}  (2);
   \draw[->] (7) to node[right] {$\ R_{\ga_j}$}  (3);
   \draw[->] (7) to node[right] {$R_{\ga_j}$}  (33);
   \draw[->] (7) to node[right] {$\!R_{\ga_j}$}  (4);
   \draw[->,line width=.4pt] (6) to (2);
   \draw[->,line width=.4pt] (6) to (3);
   \draw[->,line width=.4pt] (6) to (33);
   \draw[->,line width=.4pt] (6) to (4);
    \draw[->,line width=.4pt] (1) to (5);
   \node[point,scale=0.6,draw=white] (c1) at (5.5,0) {\Large $\dots$};  
   \node[point,scale=0.6,draw=white] (c2) at (8.5,0) {\Large $\dots$};  
   \node[point,scale=0.6,draw=white] (c3) at (11.5,0) {\Large $\dots$};  
    \end{tikzpicture}}

The proof of the following claim is provided in Appendix~\ref{a:triggerproof}:

\begin{lclaim}\label{c:trigger}
A gadget $\ga$ in $\q$ is triggered at $\mathfrak s$ iff there is $\inputgs$ such that
$\inputgs$ is gathered from `around' $\mathfrak s$ in $\C^s$ according to the
\itype{} for $\foga$ and $\foga[\inputgs]=1$.
\end{lclaim}

Claim~\ref{c:trigger} will be used in Sec.~\ref{s:qok} to show
that every $\C\in \mathfrak K_{\omqmx}$ satisfies properties 
{\bf (leaf)} and {\bf (\locality)} given in Sec.~\ref{s:connect}.






\subsubsection{Main blocks in gadgets}\label{fblocks}

Here we give a uniform description of the \formulablock{} $\fga$ of each gadget $\ga$ in $\q$.
Apart from the label $W$, which is uniform through the gadgets, for each particular $\ga$, there are a few additional labels on some nodes in $\fga$, $\fga'$ and $\iga$. These are always specific to $\ga$, but we omit indicating this to simplify notation. 

A Boolean formula $\foga(\avec{y})$ is regarded as a ditree whose vertices are called \emph{gates\/}. Leaf gates are labelled by the variables from the list $\avec{y}=(y_1,\dots,y_n)$, with each $y_i$ labelling $k_i$-many leaves of $\foga(\avec{y})$.  Each non-leaf $g$ is either an AND-gate (having $2$ children) or a NOT-gate (having $1$ child), with the outgoing edge(s) leading to the \emph{input}(\emph{s}) \emph{of} $g$. Given an assignment $\avec{b}$ of $0$ or $1$ to the \emph{input-variables} $y_i$, we compute the value of each gate in $\foga$ under $\avec{b}$ as usual in Boolean logic.

We encode the gate-structure of $\foga(\avec{y})$ by the $\fga$-block (and also by its copy $\fga'$) as follows.
With each non-leaf gate $g$ in $\foga$ we associate a fresh copy of its gadget shown below (where $D$ in brackets means that $D$ is only present when the gate in question is the root gate of $\foga$):\\
%
\begin{tikzpicture}[decoration={brace,mirror,amplitude=7},line width=0.8pt,scale =.5]
			\node[] at (5,10) {NOT-gate gadget};
			\node[point,scale=0.6,label=above:{$\inp$}] (1) at (0,10) {};
			\node[point,scale=0.6] (2) at (-1.2,9.35) {};
			\node[point,scale=0.6] (3) at (-1.5,8) {};
			\node[point,scale=0.6] (4) at (1.5,8) {};
			\node[point,scale=0.6,label=below left:{\hspace*{2mm}\scriptsize $(D)$}] (5) at (-1.2,6.65) {};
			\node[point,scale=0.6,label=below:{$\out$}] (6) at (0,6) {};
			%
			\draw[->,bend right =15] (1) to node[above left] {} (2);
			\draw[->,bend right =15] (2) to node[above left] {\scriptsize $S$} (3);
			\draw[->, bend left=40] (1) to node[above right] {\scriptsize $S$} (4);
			\draw[->,bend right=15] (3) to node[above left] {} (5);
			\draw[->,bend right=15] (5) to node[above left] {} (6);
			\draw[->, bend left=40] (4) to node[above left] {} (6);
			\end{tikzpicture}\\[-1.5cm]
			%
			%
\hspace*{1.5cm}
\begin{tikzpicture}[decoration={brace,mirror,amplitude=7},line width=0.8pt,scale =.65]
			\node[]  at (3,-4.5) {AND-gate gadget};
			\node[point,scale=0.6] (01) at (-2,2) {};
			\node[point,scale=0.6] (02) at (2,2) {};
			\node[point,scale=0.6,label=above:{$\ \ \inp_1$}] (11) at (-2.5,1) {};
			\node[point,scale=0.6,label=above:{$\inp_2\ \ $}] (12) at (2.5,1) {};
			\node[point,scale=0.6,label=above:{\scriptsize $b$}] (1) at (0,0) {};
			\node[point,scale=0.6,label=left:{\scriptsize $(D)$}] (2) at (0,-2.5) {};
			\node[point,scale=0.6,label=below:{$\out$}] (3) at (0,-4.1) {};
			\node[point,scale=0.6,label=right:{\scriptsize $c_3$}] (4) at (1.25,-2.25) {};
			\node[point,scale=0.6,label=left:{\scriptsize $c_1$}] (21) at (-4,-1) {};
			\node[point,scale=0.6] (22) at (-2,-1) {};
			\node[point,scale=0.6] (23) at (2,-1) {};
			\node[point,scale=0.6,label=right:{\scriptsize $c_2$}] (24) at (4,-1) {};
			\node[point,scale=0.6,label=right:{\scriptsize $E$}] (e1) at (.5,-1.25) {};
			\node[point,scale=0.6,label=right:{\scriptsize $E$}] (e4) at (1.75,-3) {};
			\node[point,scale=0.6,label=left:{\scriptsize $E$}] (e21) at (-4.5,-2) {};
			\node[point,scale=0.6,label=right:{\scriptsize $E$}] (e24) at (4.5,-2) {};
			\draw[->,right] (1) to node[above left] {} (e1);
			\draw[->,right] (4) to node[above left] {} (e4);
			\draw[->,right] (21) to node[above left] {} (e21);
			\draw[->,right] (24) to node[above left] {} (e24);
			\draw[->,right] (11) to node[above] {\scriptsize $S$} (1);
			\draw[->,right] (12) to node[above] {\scriptsize $S$} (1);
			\draw[->,right] (12) to node[above left] {} (01);
			\draw[->,right] (11) to node[above left] {} (02);
			\draw[->,right] (11) to node[below right = -1.6mm] {\scriptsize $S$} (21);
			\draw[->,right] (11) to node[above left] {} (22);
			\draw[->,right] (12) to node[above left] {} (23);
			\draw[->,right] (12) to node[below left = -1.6mm] {\scriptsize $S$} (24);
			\draw[->,bend right] (01) to node[above left = -1.5mm] {\scriptsize $S$} (21);
			\draw[->,bend left] (02) to node[above right = -1.5mm] {\scriptsize $S$} (24);
			\draw[->, bend right] (21) to node[above left] {} (3);
			\draw[->, bend left] (24) to node[above left] {} (3);
			\draw[->] (1) to node[above left] {} (2);
			\draw[->] (2) to node[above left] {} (3);
			\draw[->] (23) to node[right ] {\scriptsize $S$} (4);
			\draw[->] (4) to node[above left] {} (3);
			\draw[->] (22) to node[left = 2.5mm] {\scriptsize $S$} (4);
			\end{tikzpicture}

\noindent
Each branch of $\foga$ is characterised by a pair $(i,j)$ such that the leaf node of the branch is labelled by the $j$th copy $y_i^j$ of the variable $y_i$, for some $i,j$ with $1\le i\le n$ and $1\le j\le k_i$. For each pair $(i,j)$, we introduce a label $B_{ij}$.
 Suppose that $g_1$ and $g_2$ are the inputs of an AND-gate $g$. Then, for each $m = 1, 2$, if $g_m$ is a non-leaf gate, then we merge node $\out$ of the $g_m$-gadget with node $\inp_m$ of the $g$-gadget; and if $g_m$ is labelled by $y_i^j$, we merge node $\inp_m$ of the $g$-gadget with the lower $B_{ij}$ -node in $\fga$. We proceed similarly with NOT-gates. The picture below shows how $\fga$ (and its copy $\fga'$) looks like (where, apart from the $B_{ij}$, we also label some nodes with $B_i$, for $1\le i\le n$):\\

\centerline{
 \begin{tikzpicture}[decoration={brace,mirror,amplitude=7},line width=1pt,scale =.52]
\node[point,scale=0.6,label = left:$\alpha$ or $\fnode$] (1) at (-1,0) {};
\node[point,scale=0.6,label = left:$\xi$ or $\xi'$] (2) at (-1,5) {};
\node[point,scale=0.6,label = above right:${\!\!\!\!B_1,\dots,B_n}$,label = below:$\beta^F$] (3) at (7,5) {};
\node[point,scale=0.6,label = below:$\rnode$ or $\rnode'$] (b4) at (3,0) {};
\draw [thin,dashed] (7,1) -- (12,1) -- (9.5,-1) -- (7,1);
\node[]  at (9.5,.5) {{\small gate gadgets}};
\node[]  at (9.5,-.1) {{\small of $\foga$}};
\node[point,scale=0.6,label = below:$B_n$] (b3) at (3,2.5) {};
\node[point,scale=0.6,label = above left:$B_i$,label = above right:$\beta^T_i$] (b2) at (3,5) {};
\node[point,scale=0.6,label = above:$B_1$] (b1) at (3,7.5) {};
\node[point,scale=0.6,label = left:$B_{11}$] (d1) at (8,3) {};
\node[point,fill=white,scale=0.6,label = above left:$B_{11}$] (e1) at (8,1) {};
\node[point,scale=0.6,label = right:$B_{nk_n}$] (d2) at (11,3) {};
\node[point,fill=white,scale=0.6,label = above right:$B_{nk_n}$] (e2) at (11,1) {};
\node[point,scale=0.6] (4) at (5.5,2) {};
\draw[->,line width=.4pt] (b4) to node[above] {$R_\ga$} (1);
\draw[->,line width=.4pt] (2) to (1);
\draw[->] (2) to (b1);
\draw[->] (2) to (b2);
\draw[->] (2) to (b3);
\draw[->] (b1) to (3);
\draw[->] (b2) to (3);
\draw[->] (b3) to (3);
\draw[->] (d1) to (e1);
\draw[->] (d2) to (e2);
\draw[->] (3) to (b4);
\draw[->] (3) to (4);
\draw[->] (4) to (b4);
\draw[->,bend left=20] (3) to (d1);
\draw[->,bend left=30] (3) to (d2);
 \node[point,scale=0.6,draw=white] (c1) at (9.5,3) {\Large $\dots$};  
  \node[point,scale=0.6,draw=white] (c1) at (9.5,1.5) {\Large $\dots$};  
  \node[]  at (3,4) {$\vdots$};
  \node[]  at (3,6.5) {$\vdots$};
\draw [thin,dashed] (0,-1.2) -- (0,8.5);  
\node[]  at (9.5,8) {$\fga$ or $\fga'$};
 \end{tikzpicture}}


\subsubsection{Input blocks in gadgets}\label{iblocks}

Given a Boolean formula $\foga(\avec{y})$ with $\avec{y}=(y_1,\dots,y_n)$,
its input block $\iga$ consists of a uniformly describable part (depending on $\foga$ and $n$) and a \emph{gathering block} $\gga^i$, for each $i$ with $1\le i\le n$ (depending on the \itype{} of $\foga(\avec{y})$).
For each branch of $\foga$ characterised by $(i,j)$, let $g_{ij}^{1}, \dots, g_{ij}^{d_{ij}}$ be the sequence of non-leaf gates from leaf to root on the branch with leaf $y_i^j$. The structure of the input block $\iga$ is shown below:
%

\begin{center}			
			\begin{tikzpicture}[decoration={brace,mirror,amplitude=7},line width=0.8pt,scale =.68]
			\node[point,scale=1,draw=white,label = left:$\iga$] (a1) at (6.8,-3.5) {};
			\draw[dashed,thin] (a1) to (6.8,-9) {};	
			%
			%
			%
			%
			\node[point,scale=0.6,label = right:$\eta_i$] (65) at (0,-3.5) {};
			\node[point,scale=0.6,draw=white] (vv) at (0,-4.2) {\Large $\vdots$};
			\node[point,scale=0.6,label = right:$\gamma_i$] (75) at (0,-5) {};
			\draw [decorate] ([xshift = -1mm, yshift = -1mm]65.west) --node[left=2mm]{$\gga^i$} ([xshift = -3mm]75.east);
			\node[point,scale=0.6] (85) at (0,-6) {};
			\node[point,scale=0.6,label = above right:$B_i$] (95) at (0,-6.75) {};
			\node[point,scale=0.6,label = below right:$B_n$,label = below:$\vdots$] (r95) at (2,-6.75) {};
			\node[point,scale=0.6] (rr95) at (3,-6.38) {};
			\node[point,scale=0.6] (r105) at (2,-6) {};
			\node[point,scale=0.6] (r115) at (2,-5.25) {};
			\node[point,scale=0.6,label = left:$B_1$,label = below:$\vdots$] (l95) at (-2,-6.75) {};
			\node[point,scale=0.6,label = above:$\pnode$] (pg) at (4,-6) {};
			\node[point,scale=0.6,label = below:$\inode$] (ig) at (8,-6) {};
			\draw[->,right] (r95) to node[above left] {} (rr95);	
			\draw[->,right] (rr95) to node[above left] {} (pg);	
			\draw[->,right] (r105) to node[above left] {} (pg);	
			\draw[->,right] (r115) to node[above left] {} (pg);	
			\draw[->,right] (95) to node[above left] {} (r105);	
			\draw[->,bend left=23] (l95) to node[above left] {} (r115);
			\draw[->,right] (pg) to node[below left] {$R_\ga$} (ig);
			%
			%
			%
			%
			\draw[->,right] (75) to node[above left] {} (85);
			\draw[->,right] (85) to node[above left] {} (95);
			
			
			\node[point,scale=0.6] (01) at (-1.5,-7.75) {};
			\node[point,scale=0.6] (03) at (0,-7.75) {};
			\node[point,scale=0.6] (05) at (1.5,-7.75) {};
			\node[point,scale=0.6, label = left:$B_{i1}$] (11) at (-1.75,-8.5) {};
			\node[point,scale=0.6, label = right:$B_{ij}$] (13) at (0,-8.5) {};
			\node[point,scale=0.6, label = right:$B_{ik_i}$] (015) at (1.75,-8.5) {};
			\node[point,scale=0.6] (111) at (0,-9.25) {};
			\node[point,scale=0.6] (112) at (0,-10) {};
			\node[point,scale=0.6,label = right:$p_{ij}^1$] (113) at (0,-10.75) {};
			\node[label = below:{ $\vdots$}] (120) at (0,-10.75) {};
			\node[point,scale=0.6] (121) at (0,-12.25) {};
			\node[point,scale=0.6] (122) at (0,-13) {};
			\node[point,scale=0.6] (123) at (0,-13.75) {};
			\node[point,scale=0.6,label = right:$p_{ij}^\ell$] (124) at (0,-14.5) {};
			\node[label = below:{ $\vdots$}] (130) at (-1.75,-9) {};
			\node[label = below:{ $\vdots$}] (131) at (1.75,-9) {};
			\node[point,scale=0.6,label = above:$E$] (140) at (1.5,-13.75) {};
			\node[point,scale=0.6] (151) at (3,-12.25) {};
			\node[point,scale=0.6] (152) at (3,-13) {};
			\node[point,scale=0.6] (153) at (3,-13.75) {};
			\node[point,scale=0.6] (154) at (3,-14.5) {};
			\node[label = below:{ $\vdots$}] (160) at (0,-14.5) {};
			\node[point,scale=0.6] (161) at (0,-16) {};
               		\node[point,scale=0.6] (162) at (0,-16.75) {};
			\node[point,scale=0.6] (163) at (0,-17.5) {};
			\node[point,scale=0.6,label=right:$D$] (164) at (0,-18.25) {};

			%
			\draw[->,bend right] (95) to node[above left] {} (01);
			\draw[->, right] (95) to node[above left] {} (03);
			\draw[->,bend left] (95) to node[above left] {} (05);
			\draw[->,right] (01) to node[above left] {} (11);
			\draw[->,right] (03) to node[above left] {} (13);
			\draw[->,left] (05) to node[above left] {} (015);
			\draw[->,right] (13) to node[above left] {} (111);
			\draw[->,right] (111) to node[right] {\scriptsize $S$} (112);
			\draw[->,right] (112) to node[above left] {} (113);
			\draw [decorate] ([xshift = -1mm, yshift = -1mm]13.west) --node[left=2mm]{$g_{ij}^1$} ([xshift = -3mm]113.east);
			\draw[->,right] (121) to node[above left] {} (122);
			\draw[->,right] (122) to node[right] {\scriptsize $S$} (123);
			\draw[->,right] (123) to node[above left] {} (124);
			\draw [decorate] ([xshift = -1mm, yshift = -1mm]121.west) --node[left=2mm]{$g_{ij}^{\ell}$} ([xshift = -3mm]124.east);
			\draw[->,right] (161) to node[above left] {} (162);
			\draw[->,right] (162) to node[right] {\scriptsize $S$} (163);
			\draw[->,right] (163) to node[above left] {} (164);
			\draw [decorate] ([xshift = -1mm, yshift = -1mm]161.west) --node[left=2mm]{$g_{ij}^{d_{ij}}$} ([xshift = -3mm]164.east);
			\draw[->,right] (151) to node[above left] {} (152);
			\draw[->,right] (152) to node[left] {\scriptsize $S$} (153);
			\draw[->,right] (153) to node[above left] {} (154);
			\draw [decorate] ([xshift = 2.5mm, yshift = 1mm]154.west) --node[right=2mm]{$g_{i'j'}^{\ell'}\quad$\parbox[c]{3.5cm}{if $g_{ij}^\ell$ and $g_{i'j'}^{\ell'}$ are the\\[1pt] same AND-gate $g$}} ([xshift = 1mm]151.east);
			\draw[->,right] (123) to node {} (140);
			\draw[->,right] (153) to node {} (140);
			\node[point,scale=0.6,draw=white] (c1) at (-1,-8.5) {\Large $\dots$};
			\node[point,scale=0.6,draw=white] (c2) at (1.2,-8.5) {\Large $\ldots$};
			\node[point,scale=0.6,draw=white] (c3) at (-1.2,-6.8) {\Large $\dots$};
			\node[point,scale=0.6,draw=white] (c4) at (1.2,-6.8) {\Large $\ldots$};
			\end{tikzpicture}
\end{center}

Finally, we describe the gathering blocks $\gga^i$ in $\iga$. The Boolean formula in $\ga$ takes the form 
$\foga(\avec{y})=\foga(\avec{x}^1,\dots,\avec{x}^m)$ where each tuple $\avec{x}^j=(x^j_1,\dots,x^j_{n_j})$ of variables can be of two \itype:
\begin{description}
\item[(up)]
 either $\avec{x}^j$ is gathered from the 
 (unique) $n_j$-\emph{long uppath} (the reverse of the suffix of the path ending at $\mathfrak s$ in $\C^s$);
 \item[(down)]
 or
 $\avec{x}^j$ is gathered from an  $n_j$-\emph{long downpath} (the prefix of a path starting at $\mathfrak s$ in $\C^s$).
 \end{description}
 So suppose $y_{k+1},\dots,y_{k+n_j}$ are among the variables of $\foga$ such that 
 $(y_{k+1},\dots,y_{k+n_j})=\avec{x}^j$ for some $j$.
 Then, for each $i$ with $1\leq i\leq n_j$, $\gga^{k+i}$ is shown below:\\
 %
%
\centerline{   
\begin{tikzpicture}[decoration={brace,mirror,amplitude=7},line width=0.8pt,scale =.6]
\node[point,scale=0.6,label = left:$\eta_{k+i}$,label =above:if $\avec{x}^j$ is {\bf (up)}] (1) at (0,0) {};
\node[point,scale=0.6] (2) at (0,-1) {};
\node[point,scale=0.6] (3) at (0,-2) {};
\node[point,scale=0.6] (4) at (0,-3) {};
\node[point,scale=0.6] (5) at (0,-4) {};
\node[point,scale=0.6] (6) at (0,-5) {};
\node[point,scale=0.6] (7) at (0,-6) {};
\node[point,scale=0.6] (8) at (0,-7) {};
\node[point,scale=0.6] (9) at (0,-8) {};
\node[point,scale=0.6] (10) at (0,-9) {};
\node[point,scale=0.6] (11) at (0,-10) {};
\node[point,scale=0.6] (12) at (0,-11) {};
\node[point,scale=0.6,label = left:$\gamma_{k+i}$] (13) at (0,-12) {};
\draw [decorate] ([xshift = 3.5mm, yshift = 1mm]5.west) --node[right=8mm]{$n_j-i$ times} ([xshift = 1mm]1.east);
\draw [decorate] ([xshift = 3.5mm, yshift = 1mm]9.west) --node[right=13mm]{once} ([xshift = 1mm]5.east);
\draw [decorate] ([xshift = 3.5mm, yshift = 1mm]13.west) --node[right=10mm]{$i-1$ times} ([xshift = 1mm]9.east);
\draw[->] (1) to (2);
\draw[->] (2) to (3);
\draw[->] (3) to (4);
\draw[->] (4) to (5);
\draw[->] (5) to (6);
\draw[->] (6) to (7);
\draw[->,line width=.6mm] (7) to node[left] {$S$} (8);
\draw[->] (8) to (9);
\draw[->] (9) to (10);
\draw[->] (10) to (11);
\draw[->] (11) to (12);
\draw[->] (12) to (13);

\node[point,scale=0.6,label = right:$\eta_{k+i}$,label =above:if $\avec{x}^j$ is {\bf (down)}] (r1) at (6,0) {};
\node[point,scale=0.6] (r2) at (6,-1) {};
\node[point,scale=0.6] (r3) at (6,-2) {};
\node[point,scale=0.6] (r4) at (6,-3) {};
\node[point,scale=0.6] (r5) at (6,-4) {};
\node[point,scale=0.6] (r6) at (6,-5) {};
\node[point,scale=0.6] (r7) at (6,-6) {};
\node[point,scale=0.6] (r8) at (6,-7) {};
\node[point,scale=0.6] (r9) at (6,-8) {};
\node[point,scale=0.6] (r10) at (6,-9) {};
\node[point,scale=0.6] (r11) at (6,-10) {};
\node[point,scale=0.6] (r12) at (6,-11) {};
\node[point,scale=0.6,label = right:$\gamma_{k+i}$] (r13) at (6,-12) {};
\draw [decorate] ([xshift = -1mm, yshift = -1mm]r1.west) --node[left=2mm]{} ([xshift = -4mm]r5.east);
\draw [decorate] ([xshift = -1mm, yshift = -1mm]r5.west) --node[left=2mm]{} ([xshift = -4mm]r9.east);
\draw [decorate] ([xshift = -1mm, yshift = -1mm]r9.west) --node[left=2mm]{} ([xshift = -4mm]r13.east);
\draw[->] (r13) to (r12);
\draw[->] (r12) to (r11);
\draw[->] (r11) to (r10);
\draw[->] (r10) to (r9);
\draw[->] (r9) to (r8);
\draw[->] (r8) to (r7);
\draw[->,line width=.6mm] (r7) to node[right] {$S$} (r6);
\draw[->] (r6) to (r5);
\draw[->] (r5) to (r4);
\draw[->] (r4) to (r3);
\draw[->] (r3) to (r2);
\draw[->] (r2) to (r1);

\node[point,scale=0.6,label = right:$W$] (w) at (9,-1) {};
\draw[->,bend right=23] (r1) to (w);
\end{tikzpicture}}
%
In case $\avec{x}^j$ is like in {\bf (down)}, the $W$-node (of the base block) is a common successor of the $\eta_{k+i}$-nodes, for every $i=1,\dots,n_j$,
which ensures that the input bits for $y_{k+1},\dots,y_{k+n_j}$ are all gathered from the same $n_j$-long downpath;
see Appendix~\ref{a:triggerproof} for an example. 


\subsubsection{Proving that $\q$ satisfies {\bf (leaf)} and {\bf (\locality)}}\label{s:qok}

Suppose $\C$ is a  
cactus in $\mathfrak K_{\omqmx}$ and $\mathfrak s$ is a non-leaf segment in the skeleton $\C^s$ of $\C$. Then $\mathfrak s$ is
of the form $\q^-_{Z_\mathfrak s}$ for some $Z_{\mathfrak s}\in\{AT,TA,AA\}$.

{\bf (leaf)} $(\Rightarrow)$ 
Suppose $h \colon \qmxfa \to \C$ is  a homomorphism mapping $\qmxfa$ into $\mathfrak s$. 
Then there is a gadget $\ga$ 
that is triggered by $h$ at $\mathfrak s$ (that is,
the $h(\inode)=\midnode$ for the $\inode$-node in $\qmxfa$ and the $\midnode$-node in $\mathfrak s$). 
By Claim~\ref{c:trigger}, there is $\inputgs$ such that
$\inputgs$ is gathered from `around' $\mathfrak s$ in $\C^s$ according to the
\itype{} for $\foga$ and $\foga[\inputgs]=1$.
Now we have a case distinction, depending on the gadget $\ga$, as listed in Sec.~\ref{qstructure}.
Each gadget implements a formula whose \itype{} and behaviour are described in 
Sec.~\ref{s:bf}:
%

{\bf (g1)}
$\ga$ is the type $AA$ gadget implementing $\fgood$ (cf.\ Sec.~\ref{sgood}). 
By the \itype{} of $\fgood$, $\inputgs$ is the $4\dpoly+11$-long uppath, and it
does not contain the reverse of a $001{\ast}$-pattern. 
Thus, $\mathfrak s$ is not \good{} in $\C^s$.

{\bf (g2)}
There exist some $k$ with $4\le k\le 4\dpoly+11$ and $Z\in\{AT,TA\}$ such that $\ga$ is the type $Z$ gadget implementing $\fbr^k$
(cf.\ Sec.~\ref{sbr}). 
By the \itype{} of $\fbr^k$, $\inputgs$ is the $k$-long uppath, and it is
 the reverse of a sequence of the form $001{\ast}(111{\ast})^\ell w$, where
 either $w$ is empty and $\ell=0$, or $w=001$, or  $w=111$ and $\ell<\dpoly-1$.
On the other hand,
$Z_{\mathfrak s}=Z$ must hold,
and so $\C^s$ is not branching at $\mathfrak s$.
As branching at $\mathfrak s$ is required in condition 
{\bf (pb1)}
of being \pbr, it follows that $\mathfrak s$ is not \pbr{}, and so it is \syic{} in $\C^s$.

{\bf (g3)}
There exist some $k$ with $4\le k\le 4\dpoly+11$ and ${\ast}\in\{0,1\}$ such that $\ga$ is the type $AA$ gadget implementing $\fnobr^k_\ast$ (cf.\ Sec.~\ref{sbr}). 
Suppose, say, that ${\ast}=0$ (the case when ${\ast}=1$ is similar).
By the \itype{} of $\fnobr^k_\ast$, $\inputgs=(\avec{e}^{\mathfrak s},b^{\mathfrak s})$, where 
$\avec{e}^{\mathfrak s}$ is the $k$-long uppath and $b^{\mathfrak s}$ is a $1$-long downpath.
Also, $b^{\mathfrak s}=0$ and
$\avec{e}^{\mathfrak s}$ is the reverse of a sequence of the form $001{\ast}(111{\ast})^\ell w$, where
either $w$ is empty and $0<\ell<\dpoly$,
or $w=1$, or $w=11$, or $w=00$.
As having a $0$-child is forbidden in condition 
{\bf (pb2)} of being \pbr,
it follows that $\mathfrak s$ is not \pbr{}, and so it is \syic{} in $\C^s$.

{\bf (g4)}
There exists some $k$ with $4\le k\le 4\dpoly+11$ such that $\ga$ is the type $AA$ gadget implementing $\fnobr^k$ (cf.\ Sec.~\ref{sbr}). 
By the \itype{} of $\fnobr^k$, $\inputgs=(\avec{e}^{\mathfrak s},b^{\mathfrak s}_1,b^{\mathfrak s}_2)$, where 
$\avec{e}^{\mathfrak s}$ is the $k$-long uppath and each of $b^{\mathfrak s}_1$ and $b^{\mathfrak s}_2$ is a $1$-long downpath.
Also, $b^{\mathfrak s}_1\ne b^{\mathfrak s}_2$, and
$\avec{e}^{\mathfrak s}$ is the reverse of a sequence of the form $001{\ast}(111{\ast})^\ell w$, where
$w=111$ and $\ell=\dpoly-1$.
As having two children is forbidden in condition 
{\bf (pb4)} of being \pbr,
it follows that $\mathfrak s$ is not \pbr{}, and so $\mathfrak s$ is \syic{} in $\C^s$.

{\bf (g5)}
$\ga$ is the type $AA$ gadget implementing $\fdelta$. 
By the \itype{} of $\fdelta$, $\inputgs$ should have gathered data about some $\lor$-configur\-ation $c$ and
its two `subsequent' $\lor$-configurations $c_0,c_1$.
As explained in Sec.~\ref{sdelta}, $(c,c_0,c_1)$ is inconsistent with $\delta$,
and so $\mathfrak s$ is not \pcomp{} in $\C^s$. Thus, it is  \syic{} in $\C^s$.

{\bf (g6)}
$\ga$ is the type $AA$ gadget implementing $\finit$. 
By the \itype{} of $\finit$, $\inputgs$ should have gathered data about the $8$-long uppath and some $\lor$-configuration $c$.
As explained in Sec.~\ref{sinit},
the part of $\inputgs$ gathered from the $8$-long uppath 
is the reverse of some pattern $111{\ast} 001{\ast}$, but $c\ne\cinit$.
Thus, $\mathfrak s$ is not \pinit{} in $\C^s$, and so it is \syic{} in $\C^s$.

{\bf (g7)}
$\ga$ is the type $AA$ gadget implementing $\freject$. 
As explained in Sec.~\ref{sreject},
by the \itype{} of $\freject$, $\inputgs$ should have gathered data about some
state $q$, and $q=\qreject$ must hold.
Therefore,
$\mathfrak s$ represents a $\qreject$-configuration in $\C^s$,
%
as required.


{\bf (leaf)} $(\Leftarrow)$ 
Again, we have cases {\bf (g1)}--{\bf (g7)}. In each case, we have a formula $\foga$ for which some input
$\inputgs$ can be gathered from `around' $\mathfrak s$ in $\C^s$ according to its
\itype{} and for which $\foga[\inputgs]=1$. So by Claim~\ref{c:trigger}, there is 
a homomorphism $h \colon \qmxfa \to \C$ mapping $\qmxfa$ into $\mathfrak s$ and triggering $\ga$ at $\mathfrak s$.


{\bf (\locality)} 
Suppose there is a homomorphism $h \colon \qmxfa \to \C$ mapping $\qmxfa$ into some non-leaf segment $\mathfrak s$.
Then some gadget $\ga$ is triggered by $h$ at $\mathfrak s$. 
By our assumption on $\mathfrak s$ and by the proof of the $(\Rightarrow)$ direction of 
{\bf (leaf)} above,
it follows that $\ga$ can only be the $Z$ type gadget
implementing $\fbr^k$, where $Z\in\{AT,TA\}$ is such that $\mathfrak s=\q^-_Z$. 
An inspection of Fig.~\ref{fig:gadgets} shows that 
$h(\tleft)=\tleft$ and $h(\tright) \ne \tleft$ whenever $Z=TA$,
and $h(\tright)=\tright$ and $h(\tleft) \ne \tright$ whenever $Z=AT$.
Therefore, {\bf (\locality)} holds.

This completes the proof that $\q$ satisfies {\bf (leaf)} and {\bf (\locality)}.


\subsection{OMQs with Schema.org and DL-Lite$_\textsl{bool}$}

Schema.org, founded by Google, Microsoft, Yahoo and Yandex and developed by an open community process, comprises a set of rules $P(\avec{x}) \leftarrow Q(\avec{x})$, for unary or binary predicates $P$ and $Q$, together with domain and range constraints such as
\begin{align}\label{schemarule0}
T(x) \lor F(x) &\leftarrow S(x,y)\\\label{schemarule}
T(y) \lor F(y) &\leftarrow R(x,y)
\end{align}
For example, according to the Schema.org ontology, the range of the binary relation $\mathsf{musicBy}(x,y)$ is covered by the union of $\mathsf{MusicGroup}(y)$ and $\mathsf{Person}(y)$. 
In the syntax of description logic $\DLb$~\cite{ACKZ09}, rules~\eqref{schemarule0} and~\eqref{schemarule} are written as
$$
\exists S \sqsubseteq T \sqcup F \qquad \text{and} \qquad \exists R^- \sqsubseteq T \sqcup F
$$
Given any d-sirup $(\Delta_\q, \G)$, denote by $\Delta'_\q$ the `Schema.org ontology' obtained by replacing~\eqref{d-sirup1} in $\Delta_\q$ with rule~\eqref{schemarule}, for a fresh $R$. 

\begin{proposition}\label{prop:sch}
A d-sirup $(\Delta_\q, \G)$ is FO-rewritable iff $(\Delta'_\q, \G)$ is FO-rewritable.
\end{proposition}
\begin{proof}
$(\Rightarrow)$ Suppose $\Phi$ is a UCQ-rewriting of $(\Delta_\q, \G)$ and $\Phi'$ the result of replacing every $A(y)$ in $\Phi$ with $\exists x \, R(x,y)$. We claim that $\Phi'$ is an FO-rewriting of $(\Delta'_\q, \G)$. Indeed, suppose $\A'$ is any data instance for $(\Delta'_\q, \G)$. Without loss of generality we may assume that it does not contain atoms $A(a)$. Let $\A$ be the result of adding $A(b)$ to $\A'$ whenever $R(a,b) \in \A'$. Then $\Delta_\q,\A \models \G$ iff $\Delta'_\q,\A' \models \G$, and also $\A \models \Phi$ iff $\A' \models \Phi'$, from which $\Delta'_\q,\A' \models \G$ iff $\A' \models \Phi'$.

$(\Leftarrow)$ Suppose $\Phi'$ is a UCQ-rewriting of $(\Delta'_\q, \G)$ and $\Phi$ is the result of replacing every $R(x,y)$ in $\Phi'$ with $A(y)$. Let $\A$ be a data instance for $(\Delta_\q, \G)$. Without loss of generality we may assume that it does not contain atoms of the form $R(a,b)$. Let $\A'$ be the result of adding $R(a,b)$, for a fresh $a$, to $\A$ whenever $A(b) \in \A$. Then $\Delta_\q,\A \models \G$ iff $\Delta'_\q,\A' \models \G$, and also $\A \models \Phi$ iff $\A' \models \Phi'$, from which $\Delta'_\q,\A' \models \G$ iff $\A' \models \Phi'$.
\end{proof}

As a consequence of Theorem~\ref{thm:2-exp} and Proposition~\ref{prop:sch}, we obtain the following theorem, which is an improvement on~\cite[Theorem 11]{DBLP:conf/ijcai/HernichLOW15} showing \PSpace-hardness of deciding FO-rew\-ritability of UCQs mediated by Schema.org ontologies.

\begin{theorem}\label{thm:2-exp-sch}
Deciding FO-rewritability of CQs mediated by a Schema.org or $\DLb$ ontology is 2\Exp-hard.
\end{theorem}


\section{Monadic d-sirups with a ditree CQ}\label{sec:ditrees}

The high lower bound obtained in the previous section can be regarded as a formal  confirmation of the empirical fact that finding transparent syntactic, let alone practical criteria of FO-rewritability for sufficiently general classes of monadic (d-)sirups is a notoriously difficult problem. The only positive results in this direction we know of are the syntactic NC/\PTime{} dichotomy of binary \emph{chain} sirups~\cite{DBLP:journals/jacm/AfratiP93} (see also~\cite{DBLP:journals/tcs/AfratiGT03}) and the complete $\ACz$/\NL/\PTime/\coNP{} tetrachotomy of monadic \emph{path} d-sirups without twins~\cite{DBLP:conf/kr/GerasimovaKKPZ20}. 

The CQs used in the proof of Theorem~\ref{thm:2-exp} were directed acyclic graphs with one solitary $F$-node, two solitary $T$-nodes, and multiple $FT$-twins. The question we try to answer in this section is whether 
the restriction of the set of CQs admitted in d-sirups to those that are   \emph{rooted directed trees} as graphs (\emph{ditree CQs}, for short)
makes deciding FO-rewritability of d-sirups $(\Delta_\q,\G)$ easier, having in mind a complete syntactic classification of such d-sirups as an ultimate (possibly unrealistic) aim. Note for starters that, by Example~\ref{ex-comple}, the data complexity of evaluating d-sirups with a ditree CQ ranges from $\ACz$ to \L, \NL, \PTime, and \coNP. 

A CQ $\q$ is \emph{minimal} if there is no $\q\to\q'$ homomorphism, for any proper subCQ $\q'$ of $\q$. As well-known, checking minimality of tree-shaped CQs can be done in polynomial time; see, e.g.,~\cite{DBLP:journals/tcs/ChekuriR00}. 
We denote the root node of $\q$ by $\rr$ and write $x \preceq_{\q} y$ to say that there is a (directed) path from $x$ to $y$ in $\q$, and $x\prec_{\q} y$ if $x\preceq_{\q} y$ and $x\ne y$. A pair $(x,y)$ is $\prec_{\q}$-\emph{comparable} if either $x \preceq_{\q} y$ or $y \preceq_{\q} x$, 
otherwise $(x,y)$ is $\prec_{\q}$-\emph{incomparable\/}. If $x \preceq_{\q} y$ then $\delta_{\q}(x,y)$ is the number of edges between $x$ and $y$.
The \emph{distance} between any $x$ and $y$ is 
$\partial_{\q}(x,y)=\delta_{\q}\bigl(\inf_{\q}(x,y),x\bigr)+\delta_{\q}\bigl(\inf_{\q}(x,y),y\bigr)$, where  
$\inf_{\q}(x,y)$ is the unique node such that $\inf_{\q}(x,y)\preceq_{\q} x$, $\inf_{\q}(x,y)\preceq_{\q} y$ and 
$z\preceq_{\q}\inf_{\q}(x,y)$ whenever $z\preceq_{\q} x$ and $z\preceq_{\q} y$. The subscript $\q$ in $\preceq_{\q}$,
$\prec_{\q}$, $\delta_{\q}$, $\inf_{\q}$ and $\partial_{\q}$ will be dropped if understood. 

If $\ct$ is a solitary $T$-node and $\cf$ is solitary $F$-node, we call $(\ct,\cf)$ a 
\emph{\solitarypair\/}.
We say that a \solitarypair{} $(\ct,\cf)$ is of \emph{minimal distance}, if $\partial(t,f)\geq\partial(\ct,\cf)$ for any \solitarypair{}
$(t,f)$. A $\prec$-incomparable \solitarypair{} $(\ct,\cf)$ is called \emph{symmetric} if the CQ obtained by removing the labels $F$, $T$ from $\cf$, $\ct$ and cutting the branches below them is symmetric with respect to $\rr$ (see $\q_4$ in Example~\ref{ex-comple}).
A ditree $\q$ is \emph{quasi-symmetric} if it has no $\prec$-comparable \solitarypair s, and every \solitarypair{} $(\ct,\cf)$ of minimal distance is symmetric. 
%

As follows from~\cite{DBLP:conf/kr/GerasimovaKKPZ20} (where $F$ and $T$ are interchangeable), 
\begin{description}
\item[\rm (a)] if $\q$ has no solitary $F$, then $(\Delta_\q,\G)$ is FO-rewritable;

\item[\rm (b)] if $\q$ has one solitary $F$, then $(\Delta_\q,\G)$ is datalog-rewritable (and so in \PTime{} for data complexity);

\item[\rm (c)] if $\q$ has one solitary $F$ and one solitary $T$, then $(\Delta_\q,\G)$ is linear-datalog-rewritable (and so in \NL); \label{reference}

\item[\rm (d)] if $\q$ has one solitary $F$, one solitary $T$ and is quasi-symmetric, then $(\Delta_\q,\G)$ is symmetric-linear-datalog-rewritable (and so in \L).
\end{description}
The following result identifies a large and tractable class of d-sirups with a ditree CQ whose evaluation is \NL-hard:

\begin{theorem}\label{thm:tree1}
Suppose $\q$ is a minimal ditree CQ with at least one solitary $F$, at least one solitary $T$ and such that either 
\begin{itemize}
\item[$(i)$] there is a $\prec$-comparable \solitarypair{} $(\ct,\cf)$ or

\item[$(ii)$] $\q$ is not quasi-symmetric and has no $FT$-twins.
\end{itemize}
Then evaluating the d-sirup $(\Delta_\q,\G)$ is \NL-hard. 
\end{theorem}
\begin{proof}
The proof is by reduction of the \NL-complete reachability problem for dags. Given a dag $G = (V,E)$ with nodes $\snode,\tnode \in V$, we construct a data instance $\A_G$ as follows.
%
We pick a \solitarypair{} $(\ct,\cf)$ such that, 
in case $(i)$, $(\ct,\cf)$ is $\prec$-comparable and there is no solitary $T$- or $F$-node between $\ct$ and $\cf$;
and, in case $(ii)$, $(\ct,\cf)$ is of minimal distance, $\prec$-incomparable, and not symmetric. 
%
%
%
Then, in both cases, we replace each $e = (\unode,\vnode) \in E$ by a fresh copy $\q^e$ of $\q$ in which $\ct^e$ is renamed to $\unode$ with $T(\unode)$ replaced by $A(\unode)$, and $\cf^e$ is renamed to $\vnode$ with $F(\vnode)$ replaced by $A(\vnode)$. The dag $\A_G$ comprises the $\q^e$, for $e \in E$, as well as $T(\snode)$ and $F(\tnode)$.
We show that $\snode \to_G \tnode$ iff the answer to $(\Delta_\q,\G)$ over $\A_G$ is `yes'\!. 
		
$(\Rightarrow)$ If $\snode=\vnode_0, \dots, \vnode_n = \tnode$ is a path in $G$ with $e_i =(\vnode_i,\vnode_{i+1}) \in E$, for $i < n$, then for
		any model $\I$ of $\Delta_\q$ and $\A_G$, there is some $i < n$ such that  $\I \models T(\vnode_i)$ and $\I \models F(\vnode_{i+1})$, and so the identity map from $\q$ to its copy $\q^{e_i}$ is a $\q\to\I$ homomorphism.

$(\Leftarrow)$ If $\snode \not\to_G \tnode$, we define a model $\I$ of $\Delta_\q$ and $\A_G$ by labelling with $T$ the $A$-nodes in $\A_G$ that (as nodes of $G$) are reachable from $\snode$ (via a directed path in $G$) and with $F$ the remaining ones. 
We claim that if one of $(i)$ or $(ii)$ holds, then there is no homomorphism from $\q$ to $\I$, and so the answer to $(\Delta_\q,\G)$ over $\A_G$ is `no'\!. Indeed, take any map $h$ from $\q$ to $\I$, and
consider the substructure $\wormh$ of $\A_G$ comprising those copies $\q^{e_1},\dots,\q^{e_n}$ of $\q$  that have a non-empty intersection with $h(\q)$. To simplify notation, we set $\q^j = \q^{e_j}$. Then $\I$ can be regarded as a model of
$\wormh$.
The (quite arduous case-distinction) proof in Appendix~\ref{a:treeproof}
shows that $h$ cannot be a homomorphism from $\q$ to $\I$. 

%


Here, we only sketch the proof for case $(ii)$ when $\q$ is not quasi-symmetric, and we may also assume that $\q$ has no $\prec$-comparable \solitarypair s. 
Suppose $h \colon \q \to \I$ is a homomorphism, and
let $a$ be such that $h(\rr) \in \q^a$. As $(\ct,\cf)$ is $\prec$-incomparable, $\wormh$ consists of (at most) three copies
$\q^a$, $\q^{a-1}$ and $\q^{a+1}$ of $\q$, and looks as shown in the picture below:\\ 
%
\centerline{   
\begin{tikzpicture}[line width=0.7pt,scale =.67]
\node[]  at (-1.7,4) {$\wormh$};
\node[point,scale=0.6] (a1) at (0,0) {};
\node[point,fill=white,scale=0.6,label =below:$\ct^a$,label =above:\ \ \ $\cf^{a-1}$] (a2) at (3,0) {};
\node[point,scale=0.2] (a3) at (2,1) {};
\node[point,scale=0.6,label =above:$\rr^{a-1}$\  ] (a4) at (2,3) {};
\node[]  at (0,3.4) {$\q^{a-1}$};
\draw[] (a2) -- (1.8,-1) -- (2.4,-1) -- cycle; 
\draw[] (a1) to (a3);
\draw[] (a2) to (a3);
\draw[] (a4) to (a3);
\draw[] (a1) -- (-.8,-1) -- (-.3,-1) -- cycle;
\draw[] (2.4,.6) -- (2,.1) -- (2.6,.1) -- cycle;
\draw[] (.6,.3) -- (.4,-.2) -- (1,-.2)-- cycle;
\draw[] (1.6,.8) -- (1,.2) -- (1.5,.2)-- cycle;
\draw[] (2,1.8) -- (1.3,1.3) -- (1.8,1.3) -- (2,1.8);
\draw[] (2,2.7) -- (2.2,2.1) -- (2.6,2.1) -- (2,2.7);
\draw[thin,dashed,rounded corners] (3,-1.3) -- (3,4) -- (-1,4) -- (-1,-1.3) -- cycle;
\draw[] (a2) -- (3.8,-1) -- (4.4,-1) -- cycle; 
\node[point,scale=0.6] (b1) at (3,0) {};
\node[point,fill=white,scale=0.6,label =below:$\cf^a$,label =above:\ \ \ $\ct^{a+1}$] (b2) at (6,0) {};
\node[point,scale=0.2] (b3) at (5,1) {};
\node[point,scale=0.6,label =above:\ \ $\rr^{a}$] (b4) at (5,3) {};
\node[]  at (3.5,3.4) {$\q^{a}$};
\draw[] (b2) -- (4.8,-1) -- (5.4,-1) -- cycle; 
\draw[] (b1) to (b3);
\draw[] (b2) to (b3);
\draw[] (b4) to (b3);
\draw[] (5.4,.6) -- (5,.1) -- (5.6,.1) -- cycle;
\draw[] (3.6,.3) -- (3.4,-.2) -- (4,-.2)-- cycle;
\draw[] (4.6,.8) -- (4,.2) -- (4.5,.2)-- cycle;
\draw[] (5,1.8) -- (4.3,1.3) -- (4.8,1.3) -- (5,1.8);
\draw[] (5,2.7) -- (5.2,2.1) -- (5.6,2.1) -- (5,2.7);
\draw[] (b2) -- (6.8,-1) -- (7.4,-1) -- cycle; 
\node[point,scale=0.6] (c1) at (6,0) {};
\node[point,fill=white,scale=0.6] (c2) at (9,0) {};
\node[point,scale=0.2] (c3) at (8,1) {};
\node[point,scale=0.6,label =above:\quad $\rr^{a+1}$] (c4) at (8,3) {};
\node[]  at (6.8,3.4) {$\q^{a+1}$};
\draw[] (c2) -- (9.4,-1) -- (9.8,-1) -- cycle; 
\draw[] (c1) to (c3);
\draw[] (c2) to (c3);
\draw[] (c4) to (c3);
\draw[] (8.4,.6) -- (8,.1) -- (8.6,.1) -- cycle;
\draw[] (6.6,.3) -- (6.4,-.2) -- (7,-.2)-- cycle;
\draw[] (7.6,.8) -- (7,.2) -- (7.5,.2)-- cycle;
\draw[] (8,1.8) -- (7.3,1.3) -- (7.8,1.3) -- (8,1.8);
\draw[] (8,2.7) -- (8.2,2.1) -- (8.6,2.1) -- (8,2.7);
\draw[thin,dashed,rounded corners] (6,-1.3) -- (6,4) -- (3,4) -- (3,-1.3) -- cycle;
\draw[thin,dashed,rounded corners] (10,-1.3) -- (10,4) -- (6,4) -- (6,-1.3) -- cycle;
\node[point,fill=white,scale=0.6,label =above:$\ct^{a-1}$] (a1) at (0,0) {};
\node[point,fill=white,scale=0.6,label =below:$\ct^a$,label =above:\ \ \ $\cf^{a-1}$] (a2) at (3,0) {};
\node[point,fill=white,scale=0.6,label =below:$\cf^a$,label =above:\ \ \ $\ct^{a+1}$] (b2) at (6,0) {};
\node[point,fill=white,scale=0.6,label =above right:\!\!\!\!$\cf^{a+1}$] (c2) at (9,0) {};
\end{tikzpicture}}
%
As $(\ct,\cf)$ is not symmetric, $\I$ is such that the `contacts' between the $\q$-copies are either both in $F^\I$ or both in $T^\I$.

%
The following `structural' claim (tracking the possible locations of $h(\cf)$ and $h(\ct)$) is proved in 
Appendix~\ref{a:treeproof}
(it is also used in the proof of Theorem~\ref{t:tree-1-dich} below): 


\begin{lclaim}\label{c:hinf}
Suppose $(\ct,\cf)$ is $\prec$-incomparable and of minimal distance
\textup{(}though $\q$ might contain $FT$-twins\textup{)}.
If $\itf=\inf_{\q}(\ct,\cf)$ then
$h(\itf)$ is in $\q^a$, and one of the following holds\textup{:}
\begin{enumerate}
\item
$\itf^a\prec_{\q^a}h(\itf)\prec_{\q^a}\ct^a$, $h(\ct)$ is in $\q^{a-1}$ with $\cf^{a-1}\prec_{\q^{a-1}}h(\ct)$, and $h(\cf)=\ct^a$\textup{;}

\item
$\itf^a\prec_{\q^a}h(\itf)\prec_{\q^a}\cf^a$, $h(\cf)$ is in $\q^{a+1}$ with $\ct^{a+1}\prec_{\q^{a+1}}h(\cf)$, and $h(\ct)=\cf^a$\textup{;}

\item
$h(\itf)=\itf^a$, $h(\cf)=\cf^a$, and $h(\ct)$ is in $\q^a$ with $h(\ct)\prec_{\q^a}\cf^a$\textup{;}

\item
$h(\itf)=\itf^a$, $h(\ct)=\ct^a$, and $h(\cf)$ is in $\q^a$ with $h(\cf)\prec_{\q^a}\ct^a$. 
\end{enumerate}
\end{lclaim}
\noindent
However, if $\q$ contains neither $FT$-twins nor $\prec$-comparable 
\solitarypair s,
none of (1)--(4) in Claim~\ref{c:hinf} can happen.
%
\end{proof}

Denote by $\Delta^+_\q$ the d-sirup $(\Delta_\q,\G)$ extended by an extra rule $\bot \leftarrow T(x), F(x)$ saying that the predicates $F$ and $T$ are disjoint, and so $\Delta^+_\q$ with $\q$ containing an $FT$-twin is inconsistent.
As shown in~\cite{DBLP:conf/kr/GerasimovaKKPZ20}, $(\Delta_\q,\G)$ is \L-hard when $\q$ has at least one solitary $F$ and at least one solitary $T$ but no $FT$-twins. So we have:

\begin{corollary}
Every d-sirup $(\Delta^+_\q,\G)$ with a ditree $\q$ is either FO-rewritable \textup{(}if $\q$ contains $FT$-twins\textup{)},
or \L-hard  \textup{(}if $\q$ is quasi-symmetric without $FT$-twins\textup{)},
 or \NL-hard  \textup{(}otherwise\textup{)}.
\end{corollary}

The non-quasi-symmetric CQs $\q$ that are outside the scope of Theorem~\ref{thm:tree1} are those that have $FT$-twins and only contain $\prec$-incomparable \solitarypair s. 
That Theorem~\ref{thm:tree1} does not hold for such CQs is demonstrated by $\q_5$ in Example~\ref{ex-comple} (cf.~Claim~\ref{c:hinf}~(3)), $\q_6$ in Example~\ref{ex:one more} (cf.~Claim~\ref{c:hinf}~(4)), and $\q_7$, $\q_8$ below (cf.~Claim~\ref{c:hinf}~(1)), for all of which $(\Delta_\q,\G)$ is FO-rewritable. As before, the omitted labels on the arrows are all $R$.\\
\centerline{
\begin{tikzpicture}[>=latex,line width=0.8pt,rounded corners, scale = 0.64]
\node (0) at (-2.1,0) {$\q_7$};
\node[point,scale = 0.6,label=above:{\scriptsize $T$}] (0) at (-1.5,0) {};
\node[point,scale = 0.6,label=above:{\scriptsize $FT$}] (1) at (0,0) {};
\node[point,scale = 0.6] (m) at (1.5,0) {};
\node[point,scale = 0.6,label=above:{\scriptsize $FT$}] (2) at (3,0) {};
\node[point,scale = 0.6] (3) at (4.5,0) {};
\node[point,scale = 0.6] (4) at (6,0) {};
\node[point,scale = 0.6,label=above:{\scriptsize $F$}] (5) at (7.5,0) {};
\node[point,scale = 0.6,label=above:{\scriptsize $FT$}] (6) at (9,0) {};
\node[point,scale = 0.6,label=above:{\scriptsize $FT$}] (7) at (10.5,0) {};
\draw[<-,right] (0) to node[below] {}  (1);
\draw[<-,right] (1) to node[below] {}  (m);
\draw[<-,right] (m) to node[below] {} (2);
\draw[<-,right] (2) to node[below] {} (3);
\draw[->,right] (3) to node[below] {} (4);
\draw[->,right] (4) to node[below] {} (5);
\draw[->,right] (5) to node[below] {} (6);
\draw[->,right] (6) to node[below] {} (7);
\end{tikzpicture}}\\
%
%
Our next result, used in tandem with Theorem~\ref{thm:tree1}, gives an FO/\L-hardness dichotomy for d-sirups $(\Delta_\q,\G)$ with a ditree 1-CQ $\q$ (having a single solitary $F$). If such a $\q$ has $k$-many solitary $T$-nodes, each of which is $\prec$-incomparable with the $F$-node, we call it a \emph{$\Lambda$-CQ of span $k$}.  

\begin{theorem}\label{thm:dich}
$(i)$ For any $\Lambda$-CQ $\q$, either the d-sirup $(\Delta_\q,\G)$ is FO-rewritable or evaluating it is \LogSpace-hard.

$(ii)$ For $\Lambda$-CQs of span $k$, deciding this FO/\LogSpace-dichotomy can be done in time $p(|\q|)2^{p'(k)}$, for some polynomials $p$ and $p'$. Thus, deciding FO-rewritability of d-sirups with a $\Lambda$-CQ is 
fixed-parameter tractable, if the $\Lambda$-CQ's span is regarded as a parameter. 
%
\end{theorem}
\begin{proof}
Let $\q$ be a $\Lambda$-CQ with solitary $T$-nodes $T(y_1),\dots,T(y_k)$. By Prop.~\ref{thmequi}, $(\Delta_\q,\G)$ is FO-rewritable iff there exists $d < \omega$ such that any cactus (for $\q$) contains a homomorphic image of some cactus of depth $\le d$. 
The \emph{neighbourhood} of a segment $\segs$ in a cactus $\C$ consists of $\segs$ itself and those segments that, in the skeleton $\C^s$, are the children, parent and siblings of $\segs$---at most $2k + 1$ segments in total. Since $\q$ is a ditree, in which the $F$-node is $\prec$-incomparable with any $T$-node, the following holds for any cactuses $\C$, $\C'$ (for $\q$):

\begin{lclaim}\label{c:segs}
Suppose $h \colon \C \to \C'$ is a homomorphism that maps the root  of a segment $\segs$ in $\C$ to a node in a segment $\segs'$ in $\C'$. Then the nodes in $\segs$ are mapped by $h$ to nodes in the neighbourhood of $\segs'$.
\end{lclaim}
\begin{example}\label{e:three}
\em
Consider the $\Lambda$-CQ $\q_8$ of span $1$ below. We invite the reader to verify that there is a homomorphism $h \colon \C_2 \to \C_i$, for $i \ge 3$ (where $\C_i$ is obtained by $i$-many applications of \textbf{(bud)} to $\C_0 = \q_8$) such that the $h$-image of the leaf segment in $\C_2$ intersects three segments in $\C_i$. It is not hard to see that $(\Delta_\q,\G)$ is FO-rewritable to $\exists \avec{z}\,\big(\C_0\lor\C_1\lor\C_2\big)$.\\ 
\centerline{ 
\begin{tikzpicture}[decoration={brace,mirror,amplitude=7},line width=0.8pt,scale =.95]
\node (0) at (-0.3,0) {$\q_8$};
\node[point,scale=0.6] (r) at (0,0) {};
\node[point,scale=0.6,label=right:{\scriptsize $FT$}] (missed) at (0,-.5) {};
\node[point,scale=0.6] (11) at (1,0) {};
\node[point,scale=0.6] (12) at (2,0) {};
\node[point,scale=0.6] (13) at (3,0) {};
\node[point,scale=0.6] (20) at (0,-1) {};
\node[point,scale=0.6] (21) at (1,-1) {};
\node[point,scale=0.6,label=above:{\scriptsize $F$}] (22) at (2,-1) {};
\node[point,scale=0.6,label=above:{\scriptsize $FT$}] (23) at (3,-1) {};
\node[point,scale=0.6,label=above:{\scriptsize $FT$}] (24) at (4,-1) {};
\node[point,scale=0.6] (25) at (5,-1) {};
\node[point,scale=0.6,label=above:{\scriptsize $FT$}] (26) at (6,-1) {};
\node[point,scale=0.6,label=above:{\scriptsize $FT$}] (27) at (7,-1) {};
\node[point,scale=0.6,label=above:{\scriptsize $FT$}] (14) at (4,0) {};
\node[point,scale=0.6,label=above:{\scriptsize $FT$}] (15) at (5,0) {};
\node[point,scale=0.6,label=above right:{\scriptsize $FT$}] (30) at (0,-2) {};
\node[point,scale=0.6] (31) at (1,-2) {};
\node[point,scale=0.6,label=above:{\scriptsize $FT$}] (32) at (2,-2) {};
\node[point,scale=0.6,label=above:{\scriptsize $T$}] (33) at (3,-2) {};
\node[point,scale=0.6] (34) at (4,-2) {};
\node[point,scale=0.6,label=above:{\scriptsize $FT$}] (35) at (5,-2) {};
\node[point,scale=0.6,label=above:{\scriptsize $FT$}] (36) at (6,-2) {};
\draw[->,right] (r) to node[below] { } (11);
\draw[->,right] (11) to node[below] { } (12);
\draw[->,right] (12) to node[below] { } (13);
\draw[->,right] (13) to node[below] { } (14);
\draw[->,right] (14) to node[below] { } (15);
\draw[->,right] (r) to node[below] { } (missed);
\draw[->,right] (missed) to node[below] { } (20);
\draw[->,right] (20) to node[below] { } (21);
\draw[->,right] (21) to node[below] { } (22);
\draw[->,right] (22) to node[below] { } (23);
\draw[->,right] (23) to node[below] { } (24);
\draw[->,right] (24) to node[below] { } (25);
\draw[->,right] (25) to node[below] { } (26);
\draw[->,right] (26) to node[below] { } (27);
\draw[->,right] (20) to node[below] { } (30);
\draw[->,right] (30) to node[below] { } (31);
\draw[->,right] (31) to node[below] { } (32);
\draw[->,right] (32) to node[below] { } (33);
\draw[->,right] (33) to node[below] { } (34);
\draw[->,right] (34) to node[below] { } (35);
\draw[->,right] (35) to node[below] { } (36);
\end{tikzpicture}
}
\end{example}
In any skeleton $\C^s$, we label by $i\in \{1, \dots, k\}$ every edge that results from budding $T(y_i)$.
The neighbourhood of any segment $\segs$ is given by the triple $t = (P,i_\segs,C)$,  where $P \subseteq \{1,\dots,k\}$ comprises the labels on the edges from the parent $\segs'$ of $\segs$, $i_\segs$ is the label on $(\segs',\segs)$, and $C \subseteq \{1,\dots,k\}$ are the labels on the edges to $\segs$'s children. If $\segs$ is the root of $\C^s$, $P = \emptyset$ and we set  $i_\segs = 0$; if $\segs$ is a leaf, $C=\emptyset$. We refer to $\segs$ as the \emph{central segment of} $t$, and to $t$ as the \emph{type} of $\segs$; we call it a \emph{root type} if $\segs$ is the root of $\C^s$, and a \emph{leaf type} if $\segs$ is a leaf. A cactus $\C$ is \emph{acyclic} if none of the branches in $\C^s$ has two nodes of the same type.   

Let $\mathfrak G$ be the digraph whose nodes are all possible types and there is an edge $(t,t')$ labelled by $i\in \{1, \dots, k\}$ iff some skeleton $\C^s$ has an edge $(\segs,\segs')$ labelled by $i$ with $\segs$ being of type $t$ and $\segs'$ of type $t'$. Let $\chi_{\C}$ be the \emph{canonical homomorphism} of $\C^s$ to $\mathfrak G$ (mapping the segments of $\C^s$ to their types). 
For a subgraph $\mathfrak H$ of $\mathfrak G$ denote by $\bar{\mathfrak H}$ the 
result of replacing the types in $\mathfrak H$ with their central segments and glueing them at $A$-nodes as indicated by the types and edges in $\mathfrak H$, mimicking \textbf{(bud)}. We call this operation the $\bar\cdot$-\emph{closure} of $\mathfrak H$. 

A node $v$ of type $(P,i,C)$ in a subgraph $\mathfrak H$ of $\mathfrak G$ is  \emph{realisable in $\mathfrak H$} if $v$ has exactly one outgoing edge labelled by $j$ in $\mathfrak H$, for each $j \in C$.
We call $\mathfrak H$ \emph{realisable} if it has exactly one \emph{source} (a node without incoming edges) of root type and all nodes in $\mathfrak H$ are realisable.

A \emph{periodic structure} is a triple $\mathfrak P = (\pre,\per,\post)$ satisfying the following conditions. 
First, we take some realisable subgraph $\mathfrak H$ of $\mathfrak G$ and define the \emph{pre-periodic} part $\pre$ to be  the subgraph of $\mathfrak H$ induced by those nodes $v$ in $\mathfrak H$, for which there are no arbitrarily long paths from the source to $v$. The \emph{periodic} part $\per$ is induced by the remaining nodes in $\mathfrak H$. 
%
Finally, the \emph{post-periodic part} $\post$ comprises a set $R$ of nodes in $\per$,  intersecting any directed cycle in $\per$, and a family of \emph{acyclic} subgraphs $\mathfrak H_v$, $v\in R$, of $\mathfrak G$ with unique source $v$ and such that all of the $\mathfrak H_v$'s nodes are realisable in $\mathfrak H_v$, for any $v \in  R$. 
Denote by $\bar{\mathfrak P} = (\bar\pre,\bar\per,\bar\post)$ the triple obtained by taking the $\bar\cdot$-closure of the components $\pre$, $\per$ and $\post$ in $\mathfrak P$.

To illustrate, for $k=1$, in the only periodic structure with non-empty $\per$ shown below, $\pre$ comprises the root segment $\segs_r$, $\per$ the segment $\segs$ (with two $A$-nodes), and $\post$ the leaf segment $\segs_l$. There are also three `degenerate' periodic structures with empty $\per$ and $\post$.\\ 
\centerline{
		\begin{tikzpicture}[line width=0.8pt]
		\node[point,scale=0.6,draw = white] (00) at (1.5,-.45) {$1$};
		\node[point,fill=black,scale=0.5,label=above:{$\segs_r$}] (1) at (0,0) {};
		\node[point,fill=black,scale=0.5,label=above:{$\segs$}] (2) at (1.5,0) {};
		\node[point,fill=black,scale=0.5,label=above:{$\segs_l$}] (3) at (3,0) {};
		\draw[->,right] (1) to node[below] {\scriptsize $1$} (2);
		\draw[->,right] (2) to node[below] {\scriptsize $1$} (3);
		\draw[->,right, scale =2] (2) to [out=-130,in=-50,loop] (2);
\end{tikzpicture}}\\ 
%
%
%
The \emph{acyclic version} of a rooted digraph $G$ is constructed as follows. We consider each path $\pi$ starting in the root and ending at the first repeating node $v$ on $\pi$ with the last edge $(u,v)$. For all such $\pi$, $v$ and $(u,v)$, we add to $G$ a fresh node $v'$ and replace $(u,v)$ by $(u,v')$.

%
%
The proof of the following criterion can be found in 
Appendix~\ref{a:FOlemmaproof}:

\begin{lclaim}\label{lemmaFO}
The d-sirup $(\Delta_\q,\G)$ is FO-rewritable iff, for any periodic structure $\mathfrak P = (\pre,\per,\post)$ with $\per\ne\emptyset$, one of the following holds:
\begin{enumerate}
\item[{\bf (h1)}] there is a homomorphism from some cactus to the $\bar\cdot$-closure of the acyclic version of $\pre \cup \per$\textup{;}

\item[{\bf (h2)}] there is a homomorphism from the root segment of some cactus to $\bar \per$\textup{;} 

\item[{\bf (h3)}] there is a homomorphism from the root segment of one of the $\bar{\mathfrak H}_v$ to $\bar\post$. 
\end{enumerate}
\end{lclaim}

On the other hand, we have the following claim, which is proved 
in Appendix~\ref{ap:L-h}
and establishes an FO/L-hardness dichotomy of d-sirups with a $\Lambda$-CQ:

\begin{lclaim}\label{l:L-hard}
If none of conditions {\bf (h1)}--{\bf (h3)} holds, then evaluating $(\Delta_\q,\G)$ is $\L$-hard. 
\end{lclaim}

In Appendix~\ref{a:complexity}, 
we show that checking the criterion of Claim~\ref{lemmaFO} for $\Lambda$-CQs of span $k$ can be done in time $p(|\q|)2^{p'(k)}$, for some polynomials $p$ and $p'$. 
\end{proof}

As a consequence of Theorems~\ref{thm:tree1} and~\ref{thm:dich}, we obtain the  dichotomy:

\begin{corollary}
Any d-sirup $(\Delta_\q,\G)$ with a ditree 1-CQ $\q$ is either FO-rewritable or \L-hard. Deciding this dichotomy, parameterised by the number of solitary $T$-nodes in CQs, is fixed-parameter tractable.
\end{corollary}

This result is in sharp contrast to 2\ExpTime-completeness of deciding FO-rewritability of d-sirups $(\Delta_\q,\G)$ with a dag 1-CQ $\q$ having two solitary $T$-nodes. We hope that, using the techniques of~\cite{DBLP:conf/ijcai/LutzS17,DBLP:conf/kr/GerasimovaKKPZ20}, this dichotomy can be extended to a complete FO/L/NL/P-tetrachotomy of all d-sirups with a ditree 1-CQ. As a first step, we obtain the following trichotomy:

\begin{theorem}\label{t:tree-1-dich}
For any a ditree CQ $\q$ with one solitary $F$ and one solitary $T$, $(\Delta_\q,\G)$ is either FO-rewritable, or $\L$-complete, or \NL-complete. Deciding this trichotomy can be done in polynomial time. 
\end{theorem}
\begin{proof}
We use the results of~\cite{DBLP:conf/kr/GerasimovaKKPZ20} listed as items (c) and (d) on page~\pageref{reference}.
Let $\ct$ and $\cf$ be the solitary $T$- and $F$-nodes in $\q$. If $(\ct,\cf)$ is $\prec$-comparable then $(\Delta_\q,\G)$ is \NL-complete by (c) and  Theorem~\ref{thm:tree1}~$(i)$. If $\q$ is quasi-symmetric, then $(\Delta_\q,\G)$ is in \L{} by (d);
\L-hardness is shown 
in Appendix~\ref{a:trich}
by a reduction of graph reachability (using a construction that is similar to the one in the proof of Theorem~\ref{thm:tree1}). 

Otherwise, we consider two models $\I$ over the structure $\wormh$ (defined in the proof-sketch of Theorem~\ref{thm:tree1}~$(ii)$): one has both `contacts' in $F^{\I}$, the other in $T^{\I}$. We check whether there exists a homomorphism from $\q$ to either of these models: If neither, then  
$(\Delta_\q,\G)$ is \NL-hard by the the proof of Theorem~\ref{thm:tree1}~$(ii)$.
If at least one of them is possible, then $(\Delta_\q,\G)$ is FO-rewritable by Prop.~\ref{thmequi} (as one can use the $\q\to\I$ homomorphism to define homomorphisms from some depth $\leq 2$ cactus to any larger cactus).
Details can be found in Appendix~\ref{a:trich}.
\end{proof}


\section{Conclusions}

In this paper, we settled the long-standing open problem on the complexity of deciding boundedness of monadic single rule datalog programs. Namely, we proved this problem to be 2\Exp-complete---that is, as hard as deciding program boundedness of arbitrary monadic datalog programs~\cite{DBLP:conf/stoc/CosmadakisGKV88}. The main innovation of our proof is that we look at the computations of ATMs and the expansions of sirups through the lens of Boolean circuits and show how these circuits  can be `implemented' in dag-shaped CQs to verify the correctness of computations encoded by the expansions.  

We obtained this result while trying to classify a somewhat different type of basic recursive programs called monadic disjunctive sirups. The disjunctive rule $F(x) \lor T(x) \leftarrow A(x)$ can make answering a Boolean CQ it mediates in the d-sirup range between \ACz{} and \coNP. Deciding FO-rewritability of monadic d-sirups (as well as of Schema.org and $\DLb$  ontology-mediated queries) was shown to be between 2\Exp{} and 2\NExpTime, and so a complete classification of monadic d-sirups according to their data complexity can be as illusory as the classification of monadic sirups, which has been challenging the datalog community since the 1980s. 

On the other hand, this paper shows that d-sirups with ditree CQs are less impenetrable, and we believe a complete classification is possible, though it could be quite tricky and laborious. This problem as well as pinpointing the exact complexity of deciding FO-rewritability of monadic d-sirups (2\Exp{} vs 2NexpTime) are left for future work.


\begin{acks}
This work was supported by the UK EPSRC grant EP/S032282, HSE University Basic Research Program, and Russian Science Foundation 20-11-20203 (Section 4). Thanks are due to the anonymous reviewers for their comments and constructive suggestions.
\end{acks}


\newpage




\newpage

\appendix


\section{Proof of Lemma~\ref{l:dagq} from {\bf (foc)}, {\bf (leaf)} and {\bf (\locality)}}\label{a:lproof}

Let $\q$ be a 1-CQ having two solitary $T$-nodes $\tleft$ and $\tright$ such that, for every $\C,\C'\in \mathfrak K_{\omqmx}$,
\begin{description}
\item[(foc)] if $h\colon\C\to\C'$ is a homomorphism then $h$ maps 
the root segment of $\C$ into 
the root segment of $\C'$;

\item[(leaf)] there is a homomorphism $h \colon \qmxfa \to \C$ mapping $\qmxfa$ into some non-leaf segment $\mathfrak s$ of $\C$ iff either $\mathfrak s$ is \syic{} or $\mathfrak s$ 
represents a $\qreject$-configuration in $\C^s$;

\item[(\locality)] if $h$ maps $\qmxfa$ into a non-leaf segment $\mathfrak s$ that is not \pbr{} in $\C^s$ due to violating 
{\bf (pb1)},
but $\mathfrak s$ is \syc{} in $\C^s$ according to the other properties, then
\begin{itemize}
\item[--]
$h(\tleft)=\tleft$ and $h(\tright) \ne \tleft$, if $\mathfrak s=\q^-_{TA}$;
\item[--]
$h(\tright)=\tright$ and $h(\tleft) \ne \tright$, if $\mathfrak s=\q^-_{AT}$.
\end{itemize}
\end{description}
We prove the following:

\smallskip
\  {\sc Lemma}~\ref{l:dagq}.
\emph{
$\atm$ rejects $\w$ iff there is $K<\omega$ such that every $\C\in \mathfrak K_{\omqmx}$ contains a homomorphic image of some $\C^-\in \mathfrak K_{\omqmx}$ of depth $\le K$.
}

\smallskip
$(\Leftarrow)$	
Suppose $\atm$ accepts $\w$, and so there is an accepting com\-pu\-tation-tree $\T_{\textit{accept}}$. We 
 construct an ideal tree $\beta^\infty$
by using only $\beta_{\T_{\textit{accept}}}^+$ in every step, and then take a subtree $\beta^\infty_{\nodea}$ whose root $\nodea$ is the main node of some $\cinit$.
Given  $K<\omega$, we take the $(K+1)$-cut $\beta_K$ of the \whatisit{} tree $\beta^\infty_{\nodea}$.
By the ($\Rightarrow$) direction of Claim~\ref{c:syc},
every node of depth $\le K$ in $\beta_K$ is \syc{} in $\beta_K$.
By the construction of $\beta^\infty$ from the accepting \otree{} $\T_{\textit{accept}}$, no node in $\beta_K$ 
represents a $\qreject$-configuration. 
Let $\C_K$ be the 
cactus such that $\C_K^s=\beta_K$. Then, by the ($\Rightarrow$) direction of {\bf (leaf)},
there is no homomorphism from $\qmxfa$ to $\C_K$ mapping $\qmxfa$ into some non-leaf segment of $\C_K$.

It follows that no 
cactus $\C^-$ of depth $\le K$ can be homomorphically mapped to $\C_{K}$.
Indeed, 
suppose on the contrary  that there is a homomorphism $h$ from $\C^-$ to $\C_K$. 
By {\bf (foc)}, $h$ must map the root segment of $\C^-$ into the root segment of $\C_K$.
Thus, any leaf segment $\mathfrak s$ of $\C^-$
(and so $\qmxfa$) should be mapped by $h$ into some segment $\mathfrak s'$ of $\C_K$ whose depth is $\le K$  in $\C_K^s$, and so $\mathfrak s'$ is a non-leaf segment of $\C_K$, which is a contradiction.

\smallskip
$(\Rightarrow)$		
Suppose $\atm$ rejects $\w$, and so every \otree{} is rejecting.
Let $K=\dexp+8\dpoly+19$
and take some cactus $\C$ of depth $>K$. 
We will cut each of the long branches of $\C^s$ at some depth $\le K$, and show that the resulting cactus $\C^-$ can be
mapped homomorphically into $\C$.

To this end, take a branch $\mathcal{B}$ of $\C^s$ longer than $K$. We will cut $\mathcal{B}$ (and possibly some other
branches) at some depth $\le K$, and show that the resulting cactus $\C'$ can be mapped homomorphically into $\C$.
There are two cases: either the $4\dpoly+11$-long prefix of $\mathcal{B}$ does not contain a $001{\ast}$-sequence,
or it does. In the former case, the segment $\mathfrak s$ at the end of the $4\dpoly+11$-long prefix of $\mathcal{B}$
(which is a  non-leaf segment of $\C^s$) is not \good{} in $\C^s$, and so it is \syic{} in $\C^s$.
Thus, by the ($\Leftarrow$) direction of {\bf (leaf)},
there is a homomorphism $h\colon\qmxfa\to\C$ mapping $\qmxfa$ into $\mathfrak s$.
We cut $\mathcal{B}$ at $\mathfrak s$. Let $\mathfrak s'$ denote the leaf segment
corresponding to $\mathfrak s$ in the resulting cactus $\C'$. Then by mapping 
$\mathfrak s'=\qmxfa$ according to $h$ and taking the
isomorphism on any other segment of $\C'$, we obtain a homomorphism from $\C'$ to $\C$.

Now consider the latter case. We take some $001{\ast}$-sequence in the $4\dpoly+11$-long
prefix of $\mathcal{B}$,
and let $\mathfrak v$ be the segment at the end of this $001{\ast}$-sequence. 
We consider the subtree $\C_{\mathfrak v}^s$ of $\C^s$ with root $\mathfrak v$.
We claim that there is a segment $\mathfrak s$ in $\C_{\mathfrak v}^s$ whose depth 
 in $\C^s$ is $\le K$ and such that
\begin{description}
\item[(correct)]
every segment on the path from $\mathfrak v$ to $\mathfrak s$ is
\syc{} in $\C^s$, and there is a homomorphism $h_0 \colon \qmxfa \to \C$ mapping $\qmxfa$ into $\mathfrak s$.
\end{description}
%
%
Indeed, denote by $d_{\mathfrak v}$ the depth of $\mathfrak v$ in $\C^s$ (then $d_{\mathfrak v}\leq 4\dpoly+11$).
There are two cases:

$(i)$
There is some segment of depth $<K-d_{\mathfrak v}$ 
in $\C_{\mathfrak v}^s$ that is \syic{} in $\C_{\mathfrak v}^s$. 
Then we choose such a segment $\mathfrak s$ for which every segment on the path from $\mathfrak v$ to $\mathfrak s$ is \syc{} in $\C_{\mathfrak v}^s$ (and so in $\C^s$). If $\mathfrak s$ is a leaf of $\C_{\mathfrak v}^s$ (and so of $\C^s$), then 
the isomorphism from $\qmxfa$ to $\mathfrak s$ is a homomorphism from $\qmxfa$ to $\C$ mapping $\qmxfa$ into $\mathfrak s$.
Otherwise, by 
the ($\Leftarrow$) direction of {\bf (leaf)},
there is a homomorphism from $\qmxfa$ to $\C$ mapping $\qmxfa$ into the non-leaf segment $\mathfrak s$
(whose depth in $\C^s$ is $\leq K$).

$(ii)$
All segments of depth $< K-d_{\mathfrak v}$ in $\C_{\mathfrak v}^s$ are \syc{} in $\C_{\mathfrak v}^s$.\\ 
%
\centerline{
\begin{tikzpicture}[decoration={brace,mirror,amplitude=7},line width=1pt,scale =.52]
\node[point,fill=black,scale=0.7,label = left:$\mathfrak s$] (1) at (0,0) {};
\node[]  at (.9,-1) {$\gamma_{\creject}$};
\draw[thin] (1) -- (2.5,-1.5) -- (0,-1.5) -- (1);
\node[point,fill=black,scale=0.7,label = left:$\mathfrak v'$] (2) at (-1,3) {};
\node[]  at (0,1) {$\vdots$};
\node[]  at (0,2) {$\gamma_{\cinit}$};
\draw[thin] (2) -- (2,1.5) -- (-1,1.5) -- (2);
\node[point,fill=black,scale=0.3] (b1) at (-1,4) {};
\node[point,fill=black,scale=0.3] (b2) at (-1,5) {};
\node[point,fill=black,scale=0.3] (b3) at (-1,6) {};
\node[point,fill=black,scale=0.3] (b4) at (-1,7) {};
\node[point,fill=black,scale=0.7,label = left:$\mathfrak v$] (3) at (-2,8.5) {};
\node[]  at (-1.5,7.5) {$\gamma_c$};
\draw[thin] (3) -- (.5,7) -- (-2,7) -- (3);
\node[point,fill=black,scale=0.3] (d1) at (-2,9.5) {};
\node[point,fill=black,scale=0.3] (d2) at (-2,10.5) {};
\node[point,fill=black,scale=0.3] (d3) at (-2,11.5) {};
\node[point,fill=black,scale=0.3] (d4) at (-2,12.5) {};
\node[point,fill=black,scale=0.3] (d5) at (-2,15) {};
\draw (2) to node[right] {$0$} (b1);
\draw (b1) to node[right] {$1$} (b2);
\draw (b2) to node[right] {$0$} (b3);
\draw (b3) to node[right] {$0$} (b4);
\draw (3) to node[right] {$\ast$} (d1);
\draw (d1) to node[right] {$1$} (d2);
\draw (d2) to node[right] {$0$} (d3);
\draw (d3) to node[right] {$0$} (d4);
\draw (d4) to  (d5);
\draw[dotted,rounded corners=10] (2) -- (2.7,1.5) -- (3.1,0) -- (3.3,-1.8) -- (-3,-1.8) -- (-2.5,0) -- (2);
 \node[point,scale=0.7,draw=white] (c1) at (-1.5,.5) {\Large $\dots$};  
 \node[]  at (1.5,2.6) {$\beta_{\T}$};
 \draw[thin,dashed,rounded corners=10] (3) -- (1.3,7) -- (2.8,4) -- (4.1,1) -- (5,-4);
 \node[]  at (-5,-4.4) {$\mathcal{B}$}; 
\node[]  at (-2.8,14.5) {$\C^s$}; 
\draw[thin] (-6.3,-2.5) -- (5.3,-2.5); 
\draw[thin] (-6.5,-4) -- (5.3,-4);  
 \draw[thin,rounded corners=10] (d5) -- (-.25,11) -- (1.8,7) -- (3,4) -- (4.4,1) -- (5.3,-4);
 \draw[thin,rounded corners=10] (d5) -- (-4,7) -- (-6.5,-4); 
\draw[very thick,rounded corners] (3) -- (-2.2,7.5) -- (-2.4,4) -- (-4,1) -- (-5,-4);
 \draw[thin,dashed,rounded corners=10] (3) -- (-3,7) -- (-5.8,-4);
 \node[]  at (.2,8) {$\C^s_{\mathfrak v}$}; 
%
\draw[thick,gray] (-6.6,-2.5) -- (-6.8,-2.5) -- (-6.8,15) -- (-6.6,15); 
 \node[gray]  at (-7.2,8) {$K$};
\draw[thick,gray] (5.4,-2.5) -- (5.6,-2.5) -- (5.6,3) -- (5.4,3); 
 \node[gray]  at (6.2,.8) {$\geq\dexp$};
 \draw[thick,gray] (3.8,3) -- (4,3) -- (4,8.5) -- (3.8,8.5);
\node[gray]  at (5,5.8) {$4\dpoly+8$};
 \draw[thick,gray] (1.8,8.5) -- (2,8.5) -- (2,15) -- (1.8,15);
\node[gray]  at (3.8,12) {$d_{\mathfrak v}\leq 4\dpoly+11$};  
 \draw[thick,gray] (7.1,-2.5) -- (7.3,-2.5) -- (7.3,8.5) -- (7.1,8.5);
\node[gray]  at (8,4) {$\C_K^s$};   
\end{tikzpicture}}
Then let $\C_K^s$ be the 
$(K-d_{\mathfrak v})$-cut  of $\C_{\mathfrak v}^s$.
As $\C_K^s$ is a substructure of $\C_{\mathfrak v}^s$, all segments of depth $< K-d_{\mathfrak v}$ in $\C_K^s$ are \syc{} in $\C_K^s$. So, 
by the ($\Leftarrow$) direction of Claim~\ref{c:syc}, $\C_K^s$ is isomorphic to the 
$(K-d_{\mathfrak v})$-cut 
of some \whatisit{} tree, and therefore $\mathfrak v$ represents some configuration $c$. 
Whichever configuration $c$ is,
there is a segment $\mathfrak v'$ of depth $\le 4\dpoly+8$ in $\C_K^s$ that is the \mnode{}
of $\cinit$. 
As $K-d_{\mathfrak v}-(4\dpoly+8)\geq \dexp$, 
it follows that 
there is a \otree{} $\T$ such that $\beta_{\T}$ is a substructure of the subtree of $\C_K^s$ with root $\mathfrak v'$
(as the depth of each $\beta_{\T}$ is $\dexp$).
Thus, 
there is a segment $\mathfrak s$ in $\C_K^s$ 
representing a 
$\qreject$-configuration $\creject$ in $\C_K^s$ (because every \otree{} is rejecting, and so $\T$ is rejecting).
As $\mathfrak s$ is a non-leaf segment in $\C_K^s$ and $\C_K^s$ is a substructure of $\C^s$, $\mathfrak s$ is a non-leaf segment  representing $\creject$ in $\C^s $ whose depth is $\le K$ in $\C^s$.
Thus, by 
the ($\Leftarrow$) direction of {\bf (leaf)},
there is a homomorphism from $\qmxfa$ to $\C$ mapping $\qmxfa$ into $\mathfrak s$. 

So in both cases $(i)$ and $(ii)$, we have shown that there is a segment $\mathfrak s$ of depth $\le K$ in $\C^s$ 
such that 
{\bf (correct)} holds.
%
%
  However, $\mathfrak s$ is not necessarily in the branch $\mathcal{B}$.
 Let $\mathfrak s_m$ be the last ancestor of $\mathfrak s$ in $\mathcal{B}$, and list the segments
 $\mathfrak s=\mathfrak s_0,\mathfrak s_1,\dots,\mathfrak s_m$ on the path leading upwards from $\mathfrak s$
 to $\mathfrak s_m$. Let $\C'$ be obtained from $\C$ by cutting
 at $\mathfrak s_i$ every branch of $\C^s$ going through $\mathfrak s_i$ other than the one going to $\mathfrak s$,
  for every $i\leq m$.
(In particular, $\mathcal{B}$ is cut at $\mathfrak s_m$ which is of depth $\le K$.)
 Let $\mathfrak s_{i}^\star$ denote the segment corresponding to $\mathfrak s_{i}$ in $\C'$.
 Then $\mathfrak s_{0}^\star$ is a leaf in $\C'$, and so $\mathfrak s_{0}^\star=\qmxfa$. Also, for each $i>0$,
\begin{itemize}
\item[--] either $\mathfrak s_{i}^\star=\mathfrak s_{i}$ 

\item[--] or $\mathfrak s_{i}= \q^-_{AA}$ and $\mathfrak s_{i}^\star$ is either $\q^-_{AT}$ or $\q^-_{TA}$.
\end{itemize}

 %
%
We claim that, for every $i\le m$, there is some homomorphism $h_i\colon\mathfrak s_{i}^\star\to\C$
 mapping $\mathfrak s_{i}^\star$ into $\mathfrak s_{i}$ and such that
 \begin{multline}\label{localh}
 \mbox{if $\mathfrak s_{i-1}^\star$ is the $j$-child of $\mathfrak s_{i}^\star$, for $j=0,1$, then}\\
\mbox{$h_i$ maps the $t_j$-node of $\mathfrak s_{i}^\star$ to the $t_j$-node of $\mathfrak s_{i}$.}
\end{multline}
This will be enough for building a homomorphism from $\C'$ to $\C$: we take these $h_i$ on 
 each $\mathfrak s_{i}^\star$, and the isomorphism on any other segment.
 
 Indeed, if $i=0$ then the $h_0$ in {\bf (correct)} is suitable.
 If $i>0$ and $\mathfrak s_i^\star=\mathfrak s_i$, then the isomorphism is suitable for $h_i$. So suppose that 
 $\mathfrak s_i^\star\ne\mathfrak s_i$ (so $\mathfrak s_i= \q^-_{AA}$).
 We consider the case when $\mathfrak s_{i}^\star=\q^-_{AT}$, that is, $\mathfrak s_{i-1}^\star$ is a
 $0$-child of $\mathfrak s_{i}^\star$
 (the case when $\mathfrak s_{i}^\star=\q^-_{TA}$ is similar).
 Let $\C_i$ be obtained from $\C$ by cutting at $\mathfrak s_i$ the branch leading to $\mathfrak s$.
 Let $\mathfrak s_{i}^\dag$ denote the segment corresponding to $\mathfrak s_{i}$ in $\C_i$, that is, 
 $\mathfrak s_{i}^\dag=\q^-_{TA}$.\\
%
%
\centerline{
\begin{tikzpicture}[decoration={brace,mirror,amplitude=7},line width=1pt,scale =.2]
\node[]  at (1.7,7.5) {$\C$};   
 \draw[very thick] (-2.8,-8) -- (2,6);
 \node[]  at (-3.7,-7) {$\mathcal{B}$};  
\node[point,fill=black,scale=0.4,label = left:$\mathfrak s_m$] (1) at (0,0) {};
\node[point,fill=black,scale=0.3] (2) at (1,-.5) {};
\node[point,fill=black,scale=0.3] (3) at (2,-1) {};
\node[point,fill=black,scale=0.3] (4) at (3,-1.5) {};
\node[point,fill=black,scale=0.3,label = above:$\mathfrak s_i$] (5) at (4,-2) {};
\node[point,fill=black,scale=0.3] (6) at (5,-2.5) {};
\node[point,fill=black,scale=0.3] (7) at (6,-3) {};
\node[point,fill=black,scale=0.3] (8) at (7,-3.5) {};
\node[point,fill=black,scale=0.3,label = above:$\qquad{\mathfrak s_0=\mathfrak s}$] (9) at (8,-4) {};
\draw[thin] (1) to (9);
\draw[thin] (2) to (1,-6);
\draw[thin] (4) to (3,-4);
\draw[thin] (5) to (4,-7);
\draw[thin] (6) to (5,-3.8);
\draw[thin] (8) to (7,-6.8);
\draw[thin] (9) to (8,-7.5);
\draw[thin,dashed] (12.4,-8) to (12.4,7.5);
\end{tikzpicture}
\hspace*{-.4cm}
\begin{tikzpicture}[decoration={brace,mirror,amplitude=7},line width=1pt,scale =.2]
\node[]  at (1.7,7.5) {$\C'$};   
\draw[very thick] (1) -- (2,6);
\node[]  at (2,3) {$\mathcal{B}$};  
 \node[]  at (-3.7,-7) {$\phantom{\mathcal{B}}$};  
\node[point,fill=black,scale=0.4,label = left:$\mathfrak s_m^\star$] (1) at (0,0) {};
\node[point,fill=black,scale=0.3] (2) at (1,-.5) {};
\node[point,fill=black,scale=0.3] (3) at (2,-1) {};
\node[point,fill=black,scale=0.3] (4) at (3,-1.5) {};
\node[point,fill=black,scale=0.3,label = above:$\mathfrak s_i^\star$] (5) at (4,-2) {};
\node[point,fill=black,scale=0.3] (6) at (5,-2.5) {};
\node[point,fill=black,scale=0.3] (7) at (6,-3) {};
\node[point,fill=black,scale=0.3] (8) at (7,-3.5) {};
\node[point,fill=black,scale=0.3,label = above:$\ \ \ \mathfrak s_0^\star$] (9) at (8,-4) {};
\draw[thin] (1) to (9);
\draw[thin,dashed] (10,-8) to (10,7.5);
\end{tikzpicture}
\begin{tikzpicture}[decoration={brace,mirror,amplitude=7},line width=1pt,scale =.2]
\node[]  at (1.7,7.5) {$\C_i$};   
 \draw[very thick] (-2.8,-8) -- (2,6);
 \node[]  at (-3.7,-7) {$\mathcal{B}$};  
\node[point,fill=black,scale=0.4] (1) at (0,0) {};
\node[point,fill=black,scale=0.3] (2) at (1,-.5) {};
\node[point,fill=black,scale=0.3] (3) at (2,-1) {};
\node[point,fill=black,scale=0.3] (4) at (3,-1.5) {};
\node[point,fill=black,scale=0.3,label = above:$\mathfrak s_i^\dag$] (5) at (4,-2) {};
\draw[thin] (1) to (5);
\draw[thin] (2) to (1,-6);
\draw[thin] (4) to (3,-4);
\draw[thin] (5) to (4,-7);
\end{tikzpicture}}
 By {\bf (correct)}, $\mathfrak s_i$ is \syc{} in $\C^s$, and so $\mathfrak s_i$ is \pbr{}
 in $\C^s$. Thus,
 $\mathfrak s_i^\dag$ is \syic{} in $\C_i^s$ 
 because it violates condition {\bf (pb1)}
 in Sec.~\ref{ss:check}. On the other hand, $\mathfrak s_i^\dag$ is 
\syc{} in $\C_i^s$ in all the other aspects (this is because
 apart from $\mathfrak s_i$ and some of its descendants, every other segment is the same in both
 cactuses $\C$ and $\C_i$).
 Therefore, by
 the ($\Leftarrow$) direction of {\bf (leaf)},
 there is a homomorphism $h_{i}\colon\qmxfa\to\C_{i}$ mapping $\qmxfa$ into $\mathfrak s_i^\dag$.
Also, by  {\bf (\locality)},
 the same $h_i$ is a  homomorphism  from $\mathfrak s_i^\star$ to $\C$, mapping $\mathfrak s_i^\star$ to $\mathfrak s_i$ and such that \eqref{localh} holds:\\
%
\centerline{
\begin{tikzpicture}[decoration={brace,mirror,amplitude=7},line width=1pt,scale =.4]
\node[point,scale=0.6,label = above:$\tleft$,label = below:$T$] (a1) at (0,0) {};
\node[point,scale=0.6,label = above:$\tright$,label = below:$T$] (a2) at (1.5,0) {};
 \draw[thin,dashed,rounded corners=10] (-1,1.5) -- (2.5,1.5) -- (2.5,-1.5) -- (-1,-1.5) -- cycle;
\node[point,scale=0.6,label = above:$\tleft$,label = below:$T$] (b1) at (4,0) {};
\node[point,scale=0.6,label = above:$\tright$,label = below:$A$] (b2) at (5.5,0) {};
\node[point,scale=0.6,label = below:$FT$] (b3) at (7,0) {};
 \draw[thin,dashed,rounded corners=10] (3,1.5) -- (8,1.5) -- (8,-1.5) -- (3,-1.5) -- cycle;
 \node[]  at (9.5,0) {$\leadsto\ \ $};
\draw[->,gray,very thick] (0,-1.2) to [out=-70,in=-120] (4,-1.2); 
\draw[->,gray,very thick] (1.5,-1.2) to [out=-50,in=-140] (7,-1.2); 
 \node[gray]  at (3.1,-2.7) {$h_i$};
 \node[]  at (.8,2.3) {$\qmxfa$}; 
 \node[]  at (5.5,2.3) {$\mathfrak s_{i}^\dag=\q^-_{TA}$};  
\end{tikzpicture}
\begin{tikzpicture}[decoration={brace,mirror,amplitude=7},line width=1pt,scale =.4]
\node[point,scale=0.6,label = above:$\tleft$,label = below:$A$] (a1) at (0,0) {};
\node[point,scale=0.6,label = above:$\tright$,label = below:$T$] (a2) at (1.5,0) {};
 \draw[thin,dashed,rounded corners=10] (-1,1.5) -- (2.5,1.5) -- (2.5,-1.5) -- (-1,-1.5) -- cycle;
\node[point,scale=0.6,label = above:$\tleft$,label = below:$A$] (b1) at (4,0) {};
\node[point,scale=0.6,label = above:$\tright$,label = below:$A$] (b2) at (5.5,0) {};
\node[point,scale=0.6,label = below:$FT$] (b3) at (7,0) {};
 \draw[thin,dashed,rounded corners=10] (3,1.5) -- (8,1.5) -- (8,-1.5) -- (3,-1.5) -- cycle;
\draw[->,gray,very thick] (0,-1.2) to [out=-70,in=-120] (4,-1.2); 
\draw[->,gray,very thick] (1.5,-1.2) to [out=-50,in=-140] (7,-1.2); 
 \node[gray]  at (3.1,-2.7) {$h_i$};
 \node[]  at (.8,2.3) {$\mathfrak s_{i}^\star=\q^-_{AT}$}; 
 \node[]  at (5.5,2.3) {$\mathfrak s_{i}=\q^-_{AA}$};   
\end{tikzpicture}}

 So in any case we showed that there exists a $\C'\to\C$ homomorphism,
 for some subcactus $\C'$ of $\C$ where branch $\mathcal{B}$ is cut at some depth $\le K$.
If $\C'$ still has branches longer than $K$, we repeat the above process for a long branch in $\C'$ to obtain
 a $\C''\to\C'$ homomorphism for some $\C''$, and so on. At the end, we obtain a cactus $\C^-$ of depth $\le K$ homomorphically mapping into $\C$, which completes the proof of Lemma~\ref{l:dagq}.




\section{Proof of Claim~\ref{c:trigger}}\label{a:triggerproof}

$(\Rightarrow)$
Suppose that, for some $\C$ and $\mathfrak s$,  a gadget $\ga$ implementing a formula $\foga(y_1,\dots,y_n)$ 
is triggered at $\mathfrak s$. Then there is a  \mbox{$h\colon\qmxfa\to\C$} homomorphism
mapping the $\iga$-block in $\qmxfa$ to the $\fga$-block in $\mathfrak s$.
In particular, $h(\inode)=\midnode$,  
and so 
$h(\pnode)=\rnode$.
Thus, for every $i\le n$, the $B_i$-node in $\iga$ must also be mapped to one of the
two $B_i$-nodes in the $\fga$-block of $\mathfrak s$ (either $\beta_i^T$ or $\beta^F$). 
However, which of these two $B_i$-nodes is the image depends on the truth-value $b^{\mathfrak s}_i$ of the gathered input $\inputgs=(b^{\mathfrak s}_1,\dots,b^{\mathfrak s}_n)$ on the variable $y_i$. We claim that
\begin{itemize}
\item[$(i)$]
if $b^{\mathfrak s}_i=0$, 
then the $B_i$-node in $\iga$ is mapped by $h$ to $\beta^F$;

\item[$(ii)$]
if $b^{\mathfrak s}_i=1$, 
then the $B_i$-node in $\iga$ is mapped by $h$ to $\beta_i^T$.
\end{itemize}
%
Instead of proving $(i)$ and $(ii)$, here we give an illustrative example.
Suppose $\foga(y_1,\dots,y_5)$ is such that $(y_1,y_2,y_3)$ should be gathered from the $3$-long uppath,
and $(y_4,y_5)$ from a $2$-long downpath. Suppose the `environment' of $\mathfrak s$
in $\C^s$ looks like this:\\
\centerline{
\setlength{\unitlength}{.06cm}
\begin{picture}(35,70)
\thicklines
\multiput(0,10)(20,0){2}{\circle*{1}}
\multiput(10,20)(20,0){2}{\circle*{1}}
\multiput(20,30)(0,10){4}{\circle*{1}}
\put(30,10){\circle*{1}}
\put(22,30){$\mathfrak s$}
\multiput(20,63)(0,2){3}{\circle*{.5}}
\multiput(22,2)(0,2){3}{\circle*{.5}}

\multiput(19,29)(-10,-10){2}{\vector(-1,-1){8}}
\multiput(21,29)(-10,-10){2}{\vector(1,-1){8}}
\multiput(20,59)(0,-10){3}{\vector(0,-1){8}}
\put(30,19){\vector(0,-1){8}}

\put(16,55){$1$}
\put(16,45){$1$}
\put(16,35){$0$}
\put(12,26){$0$}
\put(25,26){$1$}
\put(2,16){$0$}
\put(15,16){$1$}
\put(32,15){$0$}
\end{picture}
}\\ 
%
%
Then if $h$ is a homomorphism triggering $\ga$ at $\mathfrak s$, then the possible inputs $\inputgs$ that can be gathered are $01100$, $01101$, or $01110$, because $h$ should map the pattern
%
%
\begin{center}
\begin{tikzpicture}[decoration={brace,mirror,amplitude=7},line width=0.8pt,scale =.51]
\node[point,scale=0.6,label = left:$\eta_{1}$] (1) at (0,0) {};
\node[point,scale=0.6] (2) at (0,-1) {};
\node[point,scale=0.6] (3) at (0,-2) {};
\node[point,scale=0.6] (4) at (0,-3) {};
\node[point,scale=0.6,label = left:--,label = right:--] (5) at (0,-4) {};
\node[point,scale=0.6] (6) at (0,-5) {};
\node[point,scale=0.6] (7) at (0,-6) {};
\node[point,scale=0.6] (8) at (0,-7) {};
\node[point,scale=0.6,label = left:--,label = right:--] (9) at (0,-8) {};
\node[point,scale=0.6] (10) at (0,-9) {};
\node[point,scale=0.6] (11) at (0,-10) {};
\node[point,scale=0.6] (12) at (0,-11) {};
\node[point,scale=0.6,label = left:$\gamma_{1}$,label = right:--] (13) at (0,-12) {};
\node[point,scale=0.6] (14) at (0,-13) {};
\node[point,scale=0.6,label = below:$B_{1}$] (15) at (0,-14) {};
\draw[->] (1) to (2);
\draw[->] (2) to (3);
\draw[->] (3) to (4);
\draw[->] (4) to (5);
\draw[->] (5) to (6);
\draw[->] (6) to (7);
\draw[->] (7) to (8);
\draw[->] (8) to (9);
\draw[->] (9) to (10);
\draw[->] (10) to (11);
\draw[->,line width=.5mm] (11) to node[left] {$S$} (12);
\draw[->] (12) to (13);
\draw[->] (13) to (14);
\draw[->] (14) to (15);

\node[point,scale=0.6,label = left:$\eta_{2}$] (t1) at (2,0) {};
\node[point,scale=0.6] (t2) at (2,-1) {};
\node[point,scale=0.6] (t3) at (2,-2) {};
\node[point,scale=0.6] (t4) at (2,-3) {};
\node[point,scale=0.6,label = left:--,label = right:--] (t5) at (2,-4) {};
\node[point,scale=0.6] (t6) at (2,-5) {};
\node[point,scale=0.6] (t7) at (2,-6) {};
\node[point,scale=0.6] (t8) at (2,-7) {};
\node[point,scale=0.6,label = left:--,label = right:--] (t9) at (2,-8) {};
\node[point,scale=0.6] (t10) at (2,-9) {};
\node[point,scale=0.6] (t11) at (2,-10) {};
\node[point,scale=0.6] (t12) at (2,-11) {};
\node[point,scale=0.6,label = left:$\gamma_{2}$,label = right:--] (t13) at (2,-12) {};
\node[point,scale=0.6] (t14) at (2,-13) {};
\node[point,scale=0.6,label = below:$B_{2}$] (t15) at (2,-14) {};
\draw[->] (t1) to (t2);
\draw[->] (t2) to (t3);
\draw[->] (t3) to (t4);
\draw[->] (t4) to (t5);
\draw[->] (t5) to (t6);
\draw[->] (t6) to (t7);
\draw[->,line width=.5mm] (t7) to node[left] {$S$} (t8);
\draw[->] (t8) to (t9);
\draw[->] (t9) to (t10);
\draw[->] (t10) to (t11);
\draw[->] (t11) to (t12);
\draw[->] (t12) to (t13);
\draw[->] (t13) to (t14);
\draw[->] (t14) to (t15);

\node[point,scale=0.6,label = left:$\eta_{3}$] (h1) at (4,0) {};
\node[point,scale=0.6] (h2) at (4,-1) {};
\node[point,scale=0.6] (h3) at (4,-2) {};
\node[point,scale=0.6] (h4) at (4,-3) {};
\node[point,scale=0.6,label = left:--,label = right:--] (h5) at (4,-4) {};
\node[point,scale=0.6] (h6) at (4,-5) {};
\node[point,scale=0.6] (h7) at (4,-6) {};
\node[point,scale=0.6] (h8) at (4,-7) {};
\node[point,scale=0.6,label = left:--,label = right:--] (h9) at (4,-8) {};
\node[point,scale=0.6] (h10) at (4,-9) {};
\node[point,scale=0.6] (h11) at (4,-10) {};
\node[point,scale=0.6] (h12) at (4,-11) {};
\node[point,scale=0.6,label = left:$\gamma_{3}$,label = right:--] (h13) at (4,-12) {};
\node[point,scale=0.6] (h14) at (4,-13) {};
\node[point,scale=0.6,label = below:$B_{3}$] (h15) at (4,-14) {};
\draw[->] (h1) to (h2);
\draw[->] (h2) to (h3);
\draw[->,line width=.5mm] (h3) to node[left] {$S$} (h4);
\draw[->] (h4) to (h5);
\draw[->] (h5) to (h6);
\draw[->] (h6) to (h7);
\draw[->] (h7) to (h8);
\draw[->] (h8) to (h9);
\draw[->] (h9) to (h10);
\draw[->] (h10) to (h11);
\draw[->] (h11) to (h12);
\draw[->] (h12) to (h13);
\draw[->] (h13) to (h14);
\draw[->] (h14) to (h15);

\node[point,scale=0.6,label = left:$\eta_{4}$] (r5) at (6,-4) {};
\node[point,scale=0.6] (r6) at (6,-5) {};
\node[point,scale=0.6] (r7) at (6,-6) {};
\node[point,scale=0.6] (r8) at (6,-7) {};
\node[point,scale=0.6,label = left:--,label = right:--] (r9) at (6,-8) {};
\node[point,scale=0.6] (r10) at (6,-9) {};
\node[point,scale=0.6] (r11) at (6,-10) {};
\node[point,scale=0.6] (r12) at (6,-11) {};
\node[point,scale=0.6,label = left:$\gamma_{4}$,label = right:--] (r13) at (6,-12) {};
\node[point,scale=0.6] (r14) at (6,-13) {};
\node[point,scale=0.6,label = below:$B_{4}$] (r15) at (6,-14) {};
\draw[->] (r13) to (r12);
\draw[->] (r12) to (r11);
\draw[->,line width=.5mm] (r11) to node[right] {$S$} (r10);
\draw[->] (r11) to (r10);
\draw[->] (r10) to (r9);
\draw[->] (r9) to (r8);
\draw[->] (r8) to (r7);
\draw[->] (r7) to (r6);
\draw[->] (r6) to (r5);
\draw[->] (r13) to (r14);
\draw[->] (r14) to (r15);

\node[point,scale=0.6,label = right:$\eta_{5}$] (q5) at (8,-4) {};
\node[point,scale=0.6] (q6) at (8,-5) {};
\node[point,scale=0.6] (q7) at (8,-6) {};
\node[point,scale=0.6] (q8) at (8,-7) {};
\node[point,scale=0.6,label = left:--,label = right:--] (q9) at (8,-8) {};
\node[point,scale=0.6] (q10) at (8,-9) {};
\node[point,scale=0.6] (q11) at (8,-10) {};
\node[point,scale=0.6] (q12) at (8,-11) {};
\node[point,scale=0.6,label = left:$\gamma_{5}$,label = right:--] (q13) at (8,-12) {};
\node[point,scale=0.6] (q14) at (8,-13) {};
\node[point,scale=0.6,label = below:$B_{5}$] (q15) at (8,-14) {};
\draw[->] (q13) to (q12);
\draw[->] (q12) to (q11);
\draw[->] (q11) to (q10);
\draw[->] (q10) to (q9);
\draw[->] (q9) to (q8);
\draw[->] (q8) to (q7);
\draw[->,line width=.5mm] (q7) to node[right] {$S$} (q6);
\draw[->] (q6) to (q5);
\draw[->] (q13) to (q14);
\draw[->] (q14) to (q15);

\node[point,scale=0.6,label = above:$W$] (w) at (7,-3) {};
\draw[->] (r5) to (w);
\draw[->] (q5) to (w);

\node[point,scale=0.6,draw=white,label = below:{\Large $\vdots$}] (vv) at (3,-14) {};
\node[point,scale=0.6,draw=white,label = above:$\iga$] (iga) at (-3,-7) {};
\end{tikzpicture}
\end{center}
%
to the pattern shown below:
%
\begin{center}
\begin{tikzpicture}[decoration={brace,mirror,amplitude=7},line width=0.8pt,xscale =.55,yscale =.6]
\node[point,scale=0.6,label = above:$\vdots$] (1) at (0,0) {};
\node[point,scale=0.6] (2) at (0,-1) {};
\node[point,scale=0.6] (3) at (0,-2) {};
\node[point,scale=0.6] (4) at (0,-3) {};
\node[point,scale=0.6,label = left:--,label = right:--] (5) at (0,-4) {};
\node[point,scale=0.6] (6) at (0,-5) {};
\node[point,scale=0.6] (7) at (0,-6) {};
\node[point,scale=0.6] (8) at (0,-7) {};
\node[point,scale=0.6,label = left:--,label = right:--] (9) at (0,-8) {};
\node[point,scale=0.6] (10) at (0,-9) {};
\node[point,scale=0.6] (11) at (0,-10) {};
\node[point,scale=0.6] (12) at (0,-11) {};
\node[point,scale=0.6,label = left:$A$] (13) at (0,-12) {};
\node[point,scale=0.6,label = left:$\xi$] (14) at (0,-13) {};
\node[point,scale=0.6,label = left:$\midnode$] (15) at (0,-14) {};
\draw[->] (1) to (2);
\draw[->] (2) to (3);
\draw[->,bend right=30] (3) to (4);
\draw[->,bend left=30,line width=.5mm] (3) to node[right] {$\ S$} (4);
\draw[->] (4) to (5);
\draw[->] (5) to (6);
\draw[->] (6) to (7);
\draw[->,bend right=30] (7) to (8);
\draw[->,bend left=30,line width=.5mm] (7) to node[right] {$\ S$} (8);
\draw[->] (8) to (9);
\draw[->] (9) to (10);
\draw[->] (10) to (11);
\draw[->] (11) to (12);
\draw[->,bend right=30] (12) to (13);
\draw[->,bend left=30,line width=.5mm] (12) to node[right] {$\ S$} (13);
\draw[->] (13) to (14);
\draw[->] (14) to (15);

\node[point,scale=0.6] (l16) at (-1,-15) {};
\node[point,scale=0.6,label = left:$A$] (l17) at (-2,-16) {};
\node[point,scale=0.6,label = right:{$\quad\dots$}] (l18) at (-2,-17) {};
\node[point,scale=0.6] (l19) at (-2,-18) {};
\node[point,scale=0.6] (l20) at (-3,-19) {};
\node[point,scale=0.6] (l21) at (-1,-19) {};
\node[point,scale=0.6,label = left:$A$] (l22) at (-4,-20) {};
\node[point,scale=0.6,label = right:$A$] (l23) at (0,-20) {};
\node[point,scale=0.6,label = below:$\vdots$] (l24) at (-4,-21) {};
\node[point,scale=0.6,label = below:$\vdots$] (l25) at (0,-21) {};
\node[point,scale=0.6,label = below:$W$] (w1) at (-5,-21) {};
\node[point,scale=0.6,label = below:$W$] (w2) at (-1,-21) {};
\draw [decorate] ([xshift = -29mm, yshift = -1mm]13.west) --node[left=3mm]{$\mathfrak s$} ([xshift = -12mm]l17.east);
\draw[->] (15) to (l16);
\draw[->,bend right=30] (l16) to (l17);
\draw[->,bend left=30,line width=.5mm] (l16) to node[right] {$\ S$} (l17);
\draw[->] (l17) to (l18);
\draw[->] (l18) to (l19);
\draw[->] (l18) to (-1,-16.5);
\draw[->] (l18) to (-1,-17.5);
\draw[->] (l19) to (l20);
\draw[->,bend right=30] (l19) to (l21);
\draw[->,bend left=30,line width=.5mm] (l19) to node[right] {$\ S$} (l21);
\draw[->,bend right=30] (l20) to (l22);
\draw[->,bend left=30,line width=.5mm] (l20) to node[right] {$\ S$} (l22);
\draw[->] (l21) to (l23);
\draw[->] (l22) to (l24);
\draw[->] (l23) to (l25);
\draw[->] (l22) to (w1);
\draw[->] (l24) to (w1);
\draw[->] (l23) to (w2);
\draw[->] (l25) to (w2);

\node[point,scale=0.6] (r16) at (1,-15) {};
\node[point,scale=0.6,label = right:$A$] (r17) at (2,-16) {};
\node[point,scale=0.6,label = right:{$\quad\dots$}] (r18) at (3,-17) {};
\node[point,scale=0.6] (r19) at (4,-18) {};
\node[point,scale=0.6] (r20) at (3,-19) {};
\node[point,scale=0.6] (r21) at (5,-19) {};
\node[point,scale=0.6,label = left:$A$] (r22) at (2,-20) {};
\node[point,scale=0.6,label = right:$T$] (r23) at (6,-20) {};
\node[point,scale=0.6,label = below:$\vdots$] (r24) at (2,-21) {};
\node[point,scale=0.6,label = below:$W$] (w3) at (3,-21) {};
\draw[->,bend right=30] (15) to (r16);
\draw[->,bend left=30,line width=.5mm] (15) to node[right] {$\ S$} (r16);
\draw[->] (r16) to (r17);
\draw[->] (r17) to (r18);
\draw[->] (r18) to (r19);
\draw[->] (r18) to (4,-16.5);
\draw[->] (r18) to (4,-17.5);
\draw[->] (r19) to (r20);
\draw[->,bend right=30] (r19) to (r21);
\draw[->,bend left=30,line width=.5mm] (r19) to node[right] {$\ S$} (r21);
\draw[->,bend right=30] (r20) to (r22);
\draw[->,bend left=30,line width=.5mm] (r20) to node[right] {$\ S$} (r22);
\draw[->] (r21) to (r23);
\draw[->] (r22) to (r24);
\draw[->] (r22) to (w3);
\draw[->] (r24) to (w3);

\node[point,scale=0.6,label = above:{\scriptsize $\beta_1^T\quad B_1$}] (b1) at (4,-8.8) {};
\node[point,scale=0.6,label = above:{\scriptsize $B_2$}] (b2) at (4,-10.1) {};
\node[point,scale=0.6,label = above:{\scriptsize $B_3$}] (b3) at (4,-11.4) {};
\node[point,scale=0.6,label = above:{\scriptsize $B_4$}] (b4) at (4,-12.7) {};
\node[point,scale=0.6,label = below:{\scriptsize $\beta_5^T\quad B_5$}] (b5) at (4,-14) {};
\node[point,scale=0.6,label = above :{\scriptsize $\qquad\quad B_1,\dots,B_5$},label = below:{\scriptsize $\beta^F$}] (b) at (8,-12) {};

\draw[->] (14) to (b1);
\draw[->] (14) to (b2);
\draw[->] (14) to (b3);
\draw[->] (14) to (b4);
\draw[->] (14) to (b5);
\draw[->] (b1) to (b);
\draw[->] (b2) to (b);
\draw[->] (b3) to (b);
\draw[->] (b4) to (b);
\draw[->] (b5) to (b);
\end{tikzpicture}
\end{center}
%
%
(We are also using that the parts of gadgets that are not depicted above do not contain $W$-nodes, so
the $h$-image cannot `stray' there when taking a downpath.)

It remains to see how $h$ maps the remaining part of the $\iga$-block into the $\fga$-block of $\mathfrak s$.
We claim that for every non-leaf gate $g$ in $\foga$, if $g_{ij}^\ell$ is an occurrence of $g$ on some branch,
then the end-node $p_{ij}^\ell$ of the $RSR$-pattern corresponding to  $g_{ij}^\ell$ in $\iga$
is mapped 
in such a way that
\begin{itemize}
\item[$(iii)$]
$h(p_{ij}^\ell)$ is the $\out$-node of the gadget for $g$, whenever the value of $g$ under $\inputgs$ is $0$;

\item[$(iv)$]
$h(p_{ij}^\ell)$ is the $(D)$-node of the gadget for $g$, whenever the value of $g$ under $\inputgs$ is $1$.
\end{itemize}
We prove this by induction on the tree-structure of $\foga$, going from leaves to root. 
Take some gate $g$, and let $g_{ij}^\ell$ be an occurrence of $g$. 

First, suppose that $g$ is an AND-gate. There are many cases, depending on the truth-values of $g$ and its two inputs $g_1$ and $g_2$ under $\inputgs$, and also on whether each of the $g_i$ is a leaf gate or not. We consider just two cases, the other ones are similar.
\begin{itemize}
\item[--]
Suppose that the value of $g$ under $\inputgs$ is $0$, $\ell=1$ (and so $g_1$ is a leaf labelled by $y_i$),  and 
$b^{\mathfrak s}_i=1$.
Suppose that $g_2$ is also a leaf gate, and so $g_2$ has value $0$ under $\inputgs$.
Let  $g_{i'j'}^{1}$ be an occurrence of $g_2$. By $(ii)$, the $B_{ij}$-node in $\iga$
is mapped by $h$ to the upper $B_{ij}$-node in
the $\fga$-block of $\mathfrak s$. So the first $R$-edge of the $RSR$-pattern corresponding to 
$g_{ij}^1$ is mapped to the $R$-edge connecting the two $B_{ij}$-nodes.
Thus,
the $S$-edge of the $RSR$-pattern corresponding to
$g_{ij}^1$ must be mapped to an $S$-edge starting at the $\inp_1$-node of the $g$-gadget.
Similarly, by $(i)$, the $B_{i'j'}$-node in $\iga$
is mapped by $h$ to the lower $B_{i'j'}$-node in
the $\fga$-block of $\mathfrak s$. So the $S$-edge of the $RSR$-pattern corresponding to
$g_{i'j'}^1$ must be mapped to an $S$-edge following an $R$-edge starting at the $\inp_2$-node of the $g$-gadget. As $h$ preserves $E$, the end-nodes of these two $S$-edges in the $g$-gadget must coincide,
and so it must be node $c_1$. So $h(p_{ij}^1)$ is the $\out$-node of the $g$-gadget.

\item[--]
Suppose that the value of $g$ under $\inputgs$ is $1$, and both of its inputs are non-leaf gates
having value $1$ under $\inputgs$.
Suppose $g_{ij}^{\ell-1}$ is an occurrence of $g_1$ and $g_{i'j'}^{\ell'}$ is an occurrence of $g_2$.
By the IH, $h(p_{ij}^{\ell-1})$ is the $(D)$-node of the gadget for $g_1$, and  
$h(p_{i'j'}^{\ell'})$ is the $(D)$-node of the gadget for $g_2$. 
Then the $S$-edges of the $RSR$-patterns corresponding to
$g_{ij}^{\ell-1}$ and $g_{i'j'}^{\ell'}$ must be mapped, respectively, to $S$-edges starting at the $\inp_1$- and $\inp_2$-nodes of the $g$-gadget. As $h$ preserves $E$, the end-nodes of these two $S$-edges in the $g$-gadget must coincide,
and so it must be node $b$. So $h(p_{ij}^\ell)$ is the $(D)$-node of the $g$-gadget, as required.
\end{itemize}
The case when  $g$ is a NOT-gate can be handled similarly, thereby completing
the proof of $(iii)$ and $(iv)$. As $h$ preserves $D$, it follows that $\foga[\inputgs]=1$.

$(\Leftarrow)$ 
If there is $\inputgs$ such that $\inputgs$ is gathered from `around' $\mathfrak s$ in $\C^s$ according to the
\itype{} for $\foga$ and $\foga[\inputgs]=1$, then 
we define a function $h\colon\qmxfa\to\C$ by taking
\begin{itemize}
\item[--]
$h(\midnode)=\fnode$ for the $\fnode$-node of $\mathfrak s$,
\item[--]
$h(\inode)=\midnode$ for the $\midnode$-node of $\mathfrak s$,
\end{itemize}
and mapping
\begin{itemize}
\item[--] the $\iga$-block to the $\fga$-block of $\mathfrak s$ following the structure of $\inputgs$ and $\foga$ as 
described above,
\item[--]
the $\igai$-block of every gadget $\ga_i$ different from $\ga$ to the $\igai$-block of $\mathfrak s$.
\item[--]
the $\fgai$-block of every gadget $\ga_i$ to the $\fgai'$-block of $\mathfrak s$, and
\item[--]
the $\fgai'$-block of every gadget $\ga_i$ also to the $\fgai'$-block of $\mathfrak s$.
\end{itemize}
Using the interaction-regulating mechanism between different gadgets described in Sec.~\ref{qstructure},
it is easy to see  that $h$ is a homomorphism, and $\ga$ is triggered by $h$ at $\mathfrak s$.


\section{Proof of Theorem~\ref{thm:tree1}}\label{a:treeproof}

We prove the following:

\smallskip
\  {\sc Theorem}~\ref{thm:tree1}.
\emph{Suppose $\q$ is a minimal ditree CQ with at least one solitary $F$, at least one solitary $T$ and such that either}
\begin{itemize}
\item[$(i)$] \emph{there is a $\prec$-comparable \solitarypair{} $(\ct,\cf)$ or}

\item[$(ii)$] \emph{$\q$ is not quasi-symmetric and has no $FT$-twins.}
\end{itemize}
\emph{Then evaluating the d-sirup $(\Delta_\q,\G)$ is \NL-hard.}

\smallskip
The proof is by reduction of the \NL-complete reachability problem for dags. Given a dag $G = (V,E)$ with nodes $\snode,\tnode \in V$, we construct a data instance $\A_G$ as follows.
We pick a \solitarypair{} $(\ct,\cf)$ such that,
in case $(i)$, $(\ct,\cf)$ is $\prec$-comparable and there is no solitary $T$- or $F$-node between $\ct$ and $\cf$;
and, in case $(ii)$, $(\ct,\cf)$ is of minimal distance, $\prec$-incomparable, and not symmetric.
 Then, in both cases, we replace each $e = (\unode,\vnode) \in E$ by a fresh copy $\q^e$ of $\q$ in which $\ct^e$ is renamed to $\unode$ with $T(\unode)$ replaced by $A(\unode)$, and $\cf^e$ is renamed to $\vnode$ with $F(\vnode)$ replaced by $A(\vnode)$. The dag $\A_G$ comprises the $\q^e$, for $e \in E$, as well as $T(\snode)$ and $F(\tnode)$.
We show that $\snode \to_G \tnode$ iff the answer to $(\Delta_\q,\G)$ over $\A_G$ is `yes'\!. 
		
$(\Rightarrow)$ If $\snode=\vnode_0, \dots, \vnode_n = \tnode$ is a path in $G$ with $e_i =(\vnode_i,\vnode_{i+1}) \in E$, for $i < n$, then for
		any model $\I$ of $\Delta_\q$ and $\A_G$, there is some $i < n$ such that  $\I \models T(\vnode_i)$ and $\I \models F(\vnode_{i+1})$, and so the identity map from $\q$ to its copy $\q^{e_i}$ is a $\q\to\I$ homomorphism.
		
$(\Leftarrow)$ If $\snode \not\to_G \tnode$, we define a model $\I$ of $\Delta_\q$ and $\A_G$ by labelling with $T$ the $A$-nodes in $\A_G$ that (as nodes of $G$) are reachable from $\snode$ (via a directed path in $G$) and with $F$ the remaining ones. We call these $A$-nodes \emph{contacts\/}.
We claim that if one of $(i)$ or $(ii)$ holds, then there is no homomorphism from $\q$ to $\I$, and so the answer to $(\Delta_\q,\G)$ over $\A_G$ is `no'\!. Indeed, take any map $h$ from $\q$ to $\I$, and
consider the substructure $\wormh$ of $\A_G$ comprising those copies $\q^{e_1},\dots,\q^{e_n}$ of $\q$  that have a non-empty intersection with $h(\q)$. To simplify notation, we set $\q^j = \q^{e_j}$. Then $\I$ can be regarded as a model of
$\wormh$. We show that $h$ cannot be a homomorphism from $\q$ to $\I$. We prove the two cases $(i)$ and $(ii)$ separately. 
Throughout, for any ditree CQ $\q'$ and node $x$ in it, we denote by $\q'_x$ the sub-ditree of $\q'$ with root $x$.



In case $(i)$, we picked a 
\solitarypair{} $(\ct,\cf)$ such that it is $\prec$-compa\-rable and 
there is no solitary $T$- or $F$-node between $\ct$ and $\cf$. 
Suppose that $\ct\prec\cf$ (the other case is similar).
Then $\wormh$ is as follows:

\centerline{\includegraphics[scale=0.6]{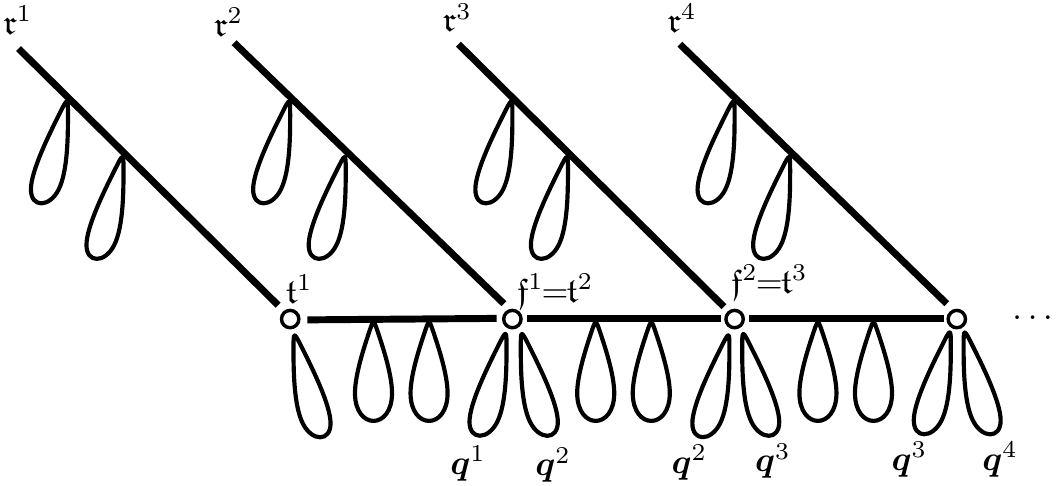}}

\noindent
We may assume that the model
$\I$ over $\wormh$ is such that the contacts between the $\q$-copies are either all in $F^\I$ or all in $T^\I$. 
We track the location of $h(\ct)$. Let $a$ be such that $h(\ct)$ is in $\q^a$ and $h(\ct)\ne\cf^a$.
We consider several cases, and show that none of them is possible.
\begin{enumerate}
\item $h(\ct)\prec_{\q^a}\ct^a$.\\
Then $h(\rr)$ is also in $\q^a$, and so there is not enough room for $h(\q)$ in $\q^a$.

\item $h(\ct)$ and $\ct^a$ are $\prec_{\q^a}$-incomparable.\\
Then there is $x$ such that $\rr\preceq x\prec \ct$ and $x^a=\inf_{\q^a}\bigl(h(\ct),\ct^a\bigr)$. 
So $h(x)\prec_{\q^a}h(\ct)$.
If $h(x)\prec_{\q^a}x^a$ as well, then there is not enough room for $h(\q)$ in $\q^a$. So $x^a\preceq_{\q_a}h(x)\prec_{\q^a}h(\ct)$. If $x^a\prec_{\q_a}h(x)$ then $\q_x$ is mapped by $h$ into $\q_{h(x)}^a$, but there is no room for that.
So it follows that $h(x)=x^a$.
Let $y$ be the child of $x$ in $\q$ with $y\preceq\ct$. Then $h(y)$ is the child of $x^a$ in $\q^a$ that is $\preceq_{\q^a}h(\ct)$, and so $y^a\ne h(y)$. Thus, $\q^a_{y^a}$ and $\q^a_{h(y)}$ are disjoint subtrees of $\q^a$, and so
$\q_{y}$ and $\q_{\iota^a\bigl(h(y)\bigr)}$ are disjoint subtrees of $\q$. Therefore, the following function $f$ is
a $\q\to(\q\setminus\q_y)$ homomorphism:
\[
f(z)=\left\{
\begin{array}{ll}
\iota^a\bigl(h(z)\bigr), & \mbox{if $z$ is in $\q_y$},\\
z & \mbox{otherwise},
\end{array}
\right.
\]
contrary to the minimality of $\q$.

\item $h(\ct)=\ct^a$.\\
Then all contacts are in $T^\I$, and so $h(\cf)\ne\cf^a$. Thus, there is $x$ such that $\ct\preceq x\prec \cf$ and $h(x)=x^a=\inf_{\q^a}\bigl(\cf^a,h(\cf)\bigr)$.  Let $y$ be the child of $x$ with $y\preceq\cf$.
Then $h(y)$ is the child of $x^a$ in $\q^a$, that is $\preceq_{\q^a}h(\cf)$, and so $y^a\ne h(y)$. Thus, $\q^a_{y^a}$ and $\q^a_{h(y)}$ are disjoint subtrees of $\q^a$, and so
$\q_{y}$ and $\q_{\iota^a\bigl(h(y)\bigr)}$ are disjoint subtrees of $\q$. Therefore, the following function $f$ is
a $\q\to(\q\setminus\q_y)$ homomorphism:
\[
f(z)=\left\{
\begin{array}{ll}
\iota^a\bigl(h(z)\bigr), & \mbox{if $z$ is in $\q_y$},\\
z & \mbox{otherwise},
\end{array}
\right.
\]
contrary to the minimality of $\q$.

\item $\ct^a\prec_{\q^a} h(\ct)$.\\
Since there are no $T$-nodes between $\ct^a$ and $\cf^a$, there are three cases: 
either $(a)$ $\cf^a\preceq_{\q^a}h(\ct)$,
or $(b)$  $(h(\ct),\cf^a)$ is $\prec_{\q^a}$-incom\-parable,
or $(c)$ $h(\ct)=\vartheta^a_\ast$ for some $FT$-twin $\vartheta_\ast$ with $\ct\prec\vartheta_\ast\prec\cf$.
%
\begin{itemize}
\item[$(a)$]
Then $\q_{\ct}$ is mapped by $h$ into $\q_{\cf}^a$, but there is no room for that as $\ct\prec\cf$. 

\item[$(b)$]
 Then $\q_{\ct}$ is mapped by $h$ into $\q_{x}^a$ for some $x$ with $\ct\prec x$, but there is no room for that.
 
 \item[$(c)$]
Then the set $X=\{\vartheta^\ell\in h(\q) \mid \vartheta\mbox{ is an $FT$-twin with $\ct\prec\vartheta$}\}$ is not empty.
 We define a function $\shift\colon h(\q)\to h(\q)$ by taking 
 $\shift(x)=h\bigl(\iota^\ell(x)\bigr)$ whenever $x$ is a node in $\q^\ell$, where we  consider each contact $\cc=\cf^i=\ct^{i+1}$, for $1\le i<n$, as a node in $\q^{i+1}$, that is, $\shift(\cc)=\shift(\ct^{i+1})=h\bigl(\iota^{i+1}(\ct^{i+1})\bigr)=h(\ct)$.
		Throughout, we use the following obvious `shift' property of $\shift$: for every $\ell$, $1\le \ell\le n$,
\begin{align}
&\text{if  $y,z$ are both in the same copy $\q^\ell$,} \notag\\ 
&\text{$y,z\ne\cf^\ell$ whenever $\ell<n$, and $y\preceq_{\q^\ell} z$, } \notag\\
& \quad \text{then } \shift(y)\preceq_{h(\q)}\shift(z) \notag\\ \label{shift}
& \qquad \text{and}\ \delta_{\q^\ell}(y,z)\ =\
		\delta_{h(\q)}\bigl(\shift(y),\shift(z)\bigr).
\end{align}	
It is straightforward to see that the restriction $g\!\mid_X$ of $\shift$ to $X$ is an $X\to X$ function.
As $X$ is finite, there exists a `fixpoint' of $g\!\mid_X$: a node $\vartheta$ in $X$ and a number $N > 0$ such that 
$(g\!\mid_X)^N(\vartheta)=\vartheta$. 
		We regard this fixpoint-cycle as a fixpoint-cycle of $\shift$, and will `shift  it to the left.' More precisely,  we claim that 
		%
		\begin{equation}\label{cfixpoint}
		\mbox{there is a contact $\cc$ with $\shift^N(\cc)=\cc$.}
		\end{equation}
		Indeed,
		let $y_0=\vartheta, y_1=\shift(\vartheta), y_2=\shift^2(\vartheta),\dots, y_{N-1}=\shift^{N-1}(\vartheta)$. 
		For every $j<N$, we have
		$y_j\in X$, and so there is a contact $\ct^{\ell_j}$ with $\ct^{\ell_j}\prec_{\q^{\ell_j}} y_j$.  We let $d_j=\delta_{\q^{\ell_j}}\bigl(\ct^{\ell_j},y_j\bigr)$.
		Let $K< N$ be such that 
		$
		d_K=\min\{ d_j\mid j< N\},
		$
		and set $\cc=\cf^{\ell_K-1}=\ct^{\ell_K}$.
		As each $y_j$ is an $FT$-twin, 
		we have $y_j\ne\cf^{\ell_j}$ whenever $j<N$. Thus,
		by the definition of $K$ and \eqref{shift}, for every $j\le N$,
		$\shift^j(\cc)$ belongs to the same copy $\q^{\ell_{K+ j}}$ as $y_{K+j}$,
		$\shift^j(\cc)\preceq_{\q^{\ell_{K+j}}}y_{K+j}$,
		and 
		\[
		d_K=
		\delta_{\q^{\ell_{K+j}}}\bigl(\shift^j(\cc),y_{K+j}\bigr).
		\]
		It follows, in particular, that $\shift^N(\cc)$ belongs to the same copy $\q^{\ell_{K}}$ as $y_{K}$, and 
		\[
		\delta_{\q^{\ell_{K}}}\bigl(\shift^N(\cc),y_K\bigr)=\delta_{\q^{\ell_{K}}}(\cc,y_{K}).
		\]
		Therefore, $\shift^N(\cc)=\cc$, as required in \eqref{cfixpoint}.
		
		It remains to show that \eqref{cfixpoint}
		leads to a contradiction. Indeed, we have either $\cc\notin F^\I$ or $\cc\notin T^\I$.
		On the other hand, by our assumption $\shift^j(\cc)=h(\ct)=\vartheta_\ast^a$ for some $FT$-twin $\vartheta_\ast$,
		and so $\shift^j(\cc)$ is an $FT$-twin for every $j$ with $1\leq j\leq N$.

 \end{itemize}
\end{enumerate}

\medskip
In case $(ii)$, 
$\q$ is not quasi-symmetric, and we may also assume that $\q$ has no $\prec$-comparable \solitarypair s.
In this case, we picked a \solitarypair{} $(\ct,\cf)$ such that it is of minimal distance, $\prec$-incomparable, and not symmetric. 
Suppose $h \colon \q \to \I$ is a homomorphism, and
let $a$ be such that $h(\rr) \in \q^a$. As $(\ct,\cf)$ is $\prec$-incomparable, $\wormh$ consists of (at most) three copies
$\q^a$, $\q^{a-1}$ and $\q^{a+1}$ of $\q$, and looks as shown in the picture below:\\ 
%
\centerline{   
\begin{tikzpicture}[line width=0.7pt,scale =.7]
\node[]  at (-1.2,4.4) {$\wormh$};
\node[point,scale=0.6] (a1) at (0,0) {};
\node[point,fill=white,scale=0.6,label =below:$\ct^a$,label =above:\ \ \ $\cf^{a-1}$] (a2) at (3,0) {};
\node[point,scale=0.2] (a3) at (2,1) {};
\node[point,scale=0.6,label =above:$\rr^{a-1}$\  ] (a4) at (2,3) {};
\node[]  at (0,3.4) {$\q^{a-1}$};
\draw[] (a2) -- (1.8,-1) -- (2.4,-1) -- cycle; 
\draw[] (a1) to (a3);
\draw[] (a2) to (a3);
\draw[] (a4) to (a3);
\draw[] (a1) -- (-.8,-1) -- (-.3,-1) -- cycle;
\draw[] (2.4,.6) -- (2,.1) -- (2.6,.1) -- cycle;
\draw[] (.6,.3) -- (.4,-.2) -- (1,-.2)-- cycle;
\draw[] (1.6,.8) -- (1,.2) -- (1.5,.2)-- cycle;
\draw[] (2,1.8) -- (1.3,1.3) -- (1.8,1.3) -- (2,1.8);
\draw[] (2,2.7) -- (2.2,2.1) -- (2.6,2.1) -- (2,2.7);
\draw[thin,dashed,rounded corners] (3,-1.3) -- (3,4) -- (-1,4) -- (-1,-1.3) -- cycle;
\draw[] (a2) -- (3.8,-1) -- (4.4,-1) -- cycle; 
\node[point,scale=0.6] (b1) at (3,0) {};
\node[point,fill=white,scale=0.6,label =below:$\cf^a$,label =above:\ \ \ $\ct^{a+1}$] (b2) at (6,0) {};
\node[point,scale=0.2] (b3) at (5,1) {};
\node[point,scale=0.6,label =above:\ \ $\rr^{a}$] (b4) at (5,3) {};
\node[]  at (3.5,3.4) {$\q^{a}$};
\draw[] (b2) -- (4.8,-1) -- (5.4,-1) -- cycle; 
\draw[] (b1) to (b3);
\draw[] (b2) to (b3);
\draw[] (b4) to (b3);
\draw[] (5.4,.6) -- (5,.1) -- (5.6,.1) -- cycle;
\draw[] (3.6,.3) -- (3.4,-.2) -- (4,-.2)-- cycle;
\draw[] (4.6,.8) -- (4,.2) -- (4.5,.2)-- cycle;
\draw[] (5,1.8) -- (4.3,1.3) -- (4.8,1.3) -- (5,1.8);
\draw[] (5,2.7) -- (5.2,2.1) -- (5.6,2.1) -- (5,2.7);
\draw[] (b2) -- (6.8,-1) -- (7.4,-1) -- cycle; 
\node[point,scale=0.6] (c1) at (6,0) {};
\node[point,fill=white,scale=0.6] (c2) at (9,0) {};
\node[point,scale=0.2] (c3) at (8,1) {};
\node[point,scale=0.6,label =above:\quad $\rr^{a+1}$] (c4) at (8,3) {};
\node[]  at (6.8,3.4) {$\q^{a+1}$};
\draw[] (c2) -- (9.4,-1) -- (9.8,-1) -- cycle; 
\draw[] (c1) to (c3);
\draw[] (c2) to (c3);
\draw[] (c4) to (c3);
\draw[] (8.4,.6) -- (8,.1) -- (8.6,.1) -- cycle;
\draw[] (6.6,.3) -- (6.4,-.2) -- (7,-.2)-- cycle;
\draw[] (7.6,.8) -- (7,.2) -- (7.5,.2)-- cycle;
\draw[] (8,1.8) -- (7.3,1.3) -- (7.8,1.3) -- (8,1.8);
\draw[] (8,2.7) -- (8.2,2.1) -- (8.6,2.1) -- (8,2.7);
\draw[thin,dashed,rounded corners] (6,-1.3) -- (6,4) -- (3,4) -- (3,-1.3) -- cycle;
\draw[thin,dashed,rounded corners] (10,-1.3) -- (10,4) -- (6,4) -- (6,-1.3) -- cycle;
\node[point,fill=white,scale=0.6,label =above:$\ct^{a-1}$] (a1) at (0,0) {};
\node[point,fill=white,scale=0.6,label =below:$\ct^a$,label =above:\ \ \ $\cf^{a-1}$] (a2) at (3,0) {};
\node[point,fill=white,scale=0.6,label =below:$\cf^a$,label =above:\ \ \ $\ct^{a+1}$] (b2) at (6,0) {};
\node[point,fill=white,scale=0.6,label =above right:\!\!\!\!$\cf^{a+1}$] (c2) at (9,0) {};
\end{tikzpicture}}
%
As $(\ct,\cf)$ is not symmetric, $\I$ is such that the contacts between the $\q$-copies are either both in $F^\I$ or both in $T^\I$.

%
The following `structural' claim (tracking the possible locations of $h(\cf)$ and $h(\ct)$) is 
also used in the proof of Theorem~\ref{t:tree-1-dich}: 


\smallskip
\ {\sc Claim}~\ref{c:hinf}.
\emph{Suppose $(\ct,\cf)$ is $\prec$-incomparable and of minimal distance
\textup{(}though $\q$ might contain $FT$-twins\textup{)}.
If $\itf=\inf_{\q}(\ct,\cf)$ then
$h(\itf)$ is in $\q^a$, and one of the following holds\textup{:}}
\begin{enumerate}
\item
\emph{$\itf^a\prec_{\q^a}h(\itf)\prec_{\q^a}\ct^a$, $h(\ct)$ is in $\q^{a-1}$ with $\cf^{a-1}\prec_{\q^{a-1}}h(\ct)$, and $h(\cf)=\ct^a$\textup{;}}

\item
\emph{$\itf^a\prec_{\q^a}h(\itf)\prec_{\q^a}\cf^a$, $h(\cf)$ is in $\q^{a+1}$ with $\ct^{a+1}\prec_{\q^{a+1}}h(\cf)$, and $h(\ct)=\cf^a$\textup{;}}

\item
\emph{$h(\itf)=\itf^a$, $h(\cf)=\cf^a$, and $h(\ct)$ is in $\q^a$ with $h(\ct)\prec_{\q^a}\cf^a$\textup{;}}

\item
\emph{$h(\itf)=\itf^a$, $h(\ct)=\ct^a$, and $h(\cf)$ is in $\q^a$ with $h(\cf)\prec_{\q^a}\ct^a$.}
\end{enumerate}

%

\begin{proof}
If $h(\itf)\in\q^{a-1}\setminus\q^a$ then $\q_{\itf}$ is mapped by $h$ into $\q^{a-1}_{\cf^{a-1}}$, but there is no room for that.
Similarly, $h(\itf)\in\q^{a+1}\setminus\q^a$ then $\q_{\itf}$ is mapped by $h$ into $\q^{a+1}_{\ct^{a+1}}$, but there is no room for that. So $h(\itf)$ is in $\q^a$. We consider several cases:
\begin{itemize}
\item[--]
$h(\itf)\prec_{\q^a}\itf^a$.\\
Then there is not enough room for $h(\q)$ in $\q^a$.

\item[--]
$(h(\itf),\itf^a)$ is $\prec_{\q^a}$-incomparable.\\
Then let $x=\inf(h(\itf),\itf^a)$. If $h(x)=x^a$ then $\q$ is not minimal. If $x^a\prec_{\q^a} h(x)\prec_{\q^a} h(\itf)$ then $\q_x$
is mapped by $h$ into $\q^a_{h(x)}$, but there is no room for that.

\item[--]
$\itf^a\prec_{\q^a}h(\itf)$, and
both $(h(\itf),\ct^a)$ and $(h(\itf),\cf^a)$ are $\prec_{\q^a}$-incomparable.\\
Then $\q_{\itf}$ is mapped by $h$ into $\q^a_{h(\itf)}$, but there is no room for that.

\item[--]
$h(\itf)=\ct^a$.\\
Then $h(\ct)$ cannot be in $\q^a$, as there is no room for that. So $h(\ct)$ is in $\q^{a-1}$, and so
depth of $\q_{\cf}$ $>$ depth of $\q_{\ct}$.
Also, $h(\cf)$ cannot be in $\q^{a-1}$, as there is no room for that. So $h(\cf)$ is in $\q^{a}$, and so
depth of $\q_{\ct}$ $>$ depth of $\q_{\cf}$, a contradiction.

\item[--]
$h(\itf)=\cf^a$.\\
This is similar to the previous case.
\end{itemize}
It follows that either $\itf^a\prec_{\q^a}h(\itf)\prec_{\q^a}\ct^a$ or $\itf^a\preceq_{\q^a}h(\itf)\prec_{\q^a}\cf^a$.
Consider first the case when $\itf^a\prec_{\q^a}h(\itf)\prec_{\q^a}\ct^a$.
First, we track the location of $h(\ct)$. 
\begin{itemize}

\item[--] $h(\ct)$ is in $\q^a$ and $(h(\ct),\ct^a)$ is $\prec_{\q^a}$-incomparable.\\
Then let $u$ be such that $u^a=\inf(h(\ct),\ct^a)$. As $u^a\prec_{\q^a} h(u)\prec_{\q^a}h(\ct)$,
$\q_u$ is mapped by $h$ into $\q_{h(u)}^a$, but there is no room for that.


\item[--] $h(\ct)$ is in $\q^a$, and $\ct^a\prec_{\q^a} h(\ct)$.\\
Then $\q_{\ct}$ is mapped by $h$ into $\q_{h(\ct)}^a$, but there is no room for that.
\end{itemize}
Therefore, it follows that $h(\ct)$ is in $\q^{a-1}$ with $\cf^{a-1}\prec_{\q^{a-1}}h(\ct)$. Thus,
\begin{equation}\label{height}
\mbox{depth of $\q_{\cf}\ >$ depth of $\q_{\ct}$.}
\end{equation}
We claim that
\begin{equation}\label{fcontact}
h(\cf)=\ct^a,
\end{equation}
as required in item (1) of Claim~\ref{c:hinf}.
Indeed, suppose $h(\cf)\ne\ct^a$. Then 
\begin{equation}\label{hft}
\mbox{$h(\cf)$ is an  $FT$-twin,}
\end{equation}
as there are no $\prec$-comparable solitary $T$- and $F$-nodes, and $(\ct,\cf)$ is of minimal distance 
by assumption.
We cannot have that $h(\cf)\in\q^{a-1}$ with $\cf^{a-1}\prec_{\q^{a-1}}h(\cf)$, as then there is no room for $h(\q_{\cf})$.
We cannot have that $h(\cf)\in\q^{a}_{\ct^a}$ by \eqref{height}.
Thus, there is $x$ such that $x^a=\inf\bigl(\ct^a,h(\cf)\bigr)$ and $\itf^a\preceq_{\q^a} x^a\prec_{\q^a}\ct^a$.
We will show that this case is not possible either.
 We define a function $\shift\colon h(\q)\to h(\q)$ by taking 
 $\shift(x)=h\bigl(\iota^\ell(x)\bigr)$ whenever $x$ is a node in $\q^\ell$, where we  consider the contact $\cc=\ct^a=\cf^{a-1}$ as a node in $\q^{a-1}$, that is, $\shift(\cc)=h\bigl(\iota^{a-1}(\cf^{a-1})\bigr)=h(\cf)$ (the definition of
 $\shift(\cc')$ when the other contact $\cc'$ is in $h(\q)$ does not matter).
 Let $g$ be the restriction of $\shift$ to $\q^a_{x^a}\cup\q^{a-1}_{\cf^{a-1}}$.
 We will use the following obvious `shift' property of $g$: for every $\ell\in\{a,a-1\}$,
\begin{multline}\label{shiftagain}
\mbox{if  $y,z$ are both in the same copy $\q^\ell$, $z\ne\ct^a$, and $y\preceq_{\q^\ell} z$, }\\
\mbox{then } g(y)\preceq_{h(\q)}g(z)\ \mbox{and}\ \delta_{\q^\ell}(y,z)\ =\
		\delta_{h(\q)}\bigl(g(y),g(z)\bigr).
\end{multline}	
Now it is not hard to see that 
 $g\bigl(\q^a_{x^a}\cup\q^{a-1}_{\cf^{a-1}}\bigr)\subseteq\q^a_{x^a}\cup\q^{a-1}_{\cf^{a-1}}$.
 Moreover, using \eqref{shiftagain} one can also show that, for every $FT$-twin $\vartheta$,
 \begin{align}
 \label{upshift}
& \mbox{if $\vartheta^a,h(\vartheta)\in\q^a_{x^a}$, then $\delta_{\q^a}(x^a,\vartheta^a)<\delta_{\q^a}\bigl(x^a,h(\vartheta)\bigr)$, and}\\
\nonumber
& \mbox{if $\vartheta^{a-1},h(\vartheta)\in\q^{a-1}_{\cf^{a-1}}$,}\\
\label{downshift}
&\hspace*{2cm}\mbox{then $\delta_{\q^{a-1}}(\cf^{a-1},\vartheta^{a-1})>\delta_{\q^{a-1}}\bigl(\cf^a,h(\vartheta)\bigr)$.}
\end{align}
Now let $X=\{\vartheta^\ell\in \q^a_{x^a}\cup\q^{a-1}_{\cf^{a-1}} \mid \vartheta\mbox{ is an $FT$-twin}\}$.
Then the restriction $g\!\mid_X$ of $g$ to $X$ is an $X\to X$ function.
As $X$ is finite, there exists a `fixpoint' of $g\!\mid_X$: a node $\vartheta^\ell$ in $X$ and a number $N > 0$ such that 
$(g\!\mid_X)^N(\vartheta^\ell)=g^N(\vartheta^\ell)=\vartheta^\ell$. 
Now it follows from \eqref{upshift} and \eqref{downshift} that $N>1$ and 
there is some $j\leq N$ with $g^j(\vartheta^\ell)\in\q^{a-1}$.
Let $j\le N$ be such that the distance $\delta_{\q^{a-1}}\bigl(\cf^a,g^j(\vartheta^\ell)\bigr)$ is minimal, that is,
$g^{j+1}(\vartheta^\ell)\notin\q^{a-1}$, and let $\eta^{a-1}= g^j(\vartheta^\ell)\bigr)$. We `shift up' the fixpoint-cycle with the distance $\delta_{\q^{a-1}}\bigl(\cf^a,g^j(\vartheta^\ell)\bigr)$: By \eqref{shiftagain}, we have $g^N(\cf^{a-1})=\cf^{a-1}$ 
and $g(\cf^{a-1})=h(\cf)$ is in $\q^a$. By \eqref{hft}, $h(\cf)$ is an $FT$-twin, and so $g^N(\cf^{a-1})=g^{N-1}\bigl(g(\cf^{a-1})\bigr)$ is an $FT$-twin as well. But $\cf^{a-1}$ is a contact, and so it cannot be both in $F^{\I}$ and $T^{\I}$, 
 a contradiction, proving \eqref{fcontact}.


The case when $\itf^a\prec_{\q^a}h(\itf)\prec_{\q^a}\cf^a$ is similar, and it follows that 
$h(\cf)$ is in $\q^{a+1}$ with $\ct^{a+1}\prec_{\q^{a+1}}h(\cf)$, and 
$h(\ct)=\cf^a$, as required in item (2) of Claim~\ref{c:hinf}.


So suppose that $h(\itf)=\itf^a$. It is easy to see that $h(\ct)$ cannot be such that it is in $\q^a$ but
both $(h(\ct),\ct^a)$ and $(h(\ct),\cf^a)$ are $\prec_{\q^a}$-incomparable (as otherwise $\q$ would not be minimal).
Also, $h(\ct)$ cannot be in $\q^{a+1}\setminus\q^a$, as there is no room for the $h$-image of $\q_{\ct}$ there.
Similarly, $h(\cf)$ cannot be such that it is in $\q^a$ but
both $(h(\cf),\ct^a)$ and $(h(\cf),\cf^a)$ are $\prec_{\q^a}$-incomparable.
Also, $h(\cf)$ cannot be in $\q^{a-1}\setminus\q^a$, as there is no room for the $h$-image of $\q_{\cf}$ there.
It also follows that $\delta(\itf,\ct)\ne\delta(\itf,\cf)$ (as either both contacts are in $T^{\I}$ or both are in $F^{\I}$, and
so $h(\ct)$ and $h(\cf)$ cannot both be contacts).

Thus, both $h(\ct)$ and $h(\cf)$ are in $\q^a$, 
either $(h(\ct),\ct^a)$ or $(h(\ct),\cf^a)$ is $\prec_{\q^a}$-comparable, and
either $(h(\cf),\ct^a)$ or $(h(\cf),\cf^a)$ is $\prec_{\q^a}$-comparable.

Next, we show that if $(h(\ct),\cf^a)$ is $\prec_{\q^a}$-comparable, then $h(\ct)\prec_{\q^a}\cf^a$.
Indeed, suppose on the contrary that $\cf^a\prec_{\q^q}h(\ct)$.
Let $x$ be such that $\itf\prec x\prec\ct$ and $h(x)=\cf^a$. If $x$ is not labelled by $T$ in $\q$ then $h$ maps a proper subCQ of $\q$ to $\q$, contradicting the minimality of $\q$.
If $x$ is labelled by $T$ then $\cf^a\in T^{\I}$ should hold.
Thus, $h(\cf)\ne\cf^a$, and so $h(\cf)=x$. Therefore, $x$ must be labelled by $F$ too, so it is an $FT$-twin. But then
$h(x)$ cannot be a contact. 

Similarly, it can be shown that if $(h(\cf),\ct^a)$ is $\prec_{\q^a}$-comparable, then $h(\cf)\prec_{\q^a}\ct^a$.

Now, suppose first that $(h(\ct),\cf^a)$ is $\prec_{\q^a}$-comparable. Then it follows that $(h(\cf),\cf^a)$ is $\prec_{\q^a}$-comparable, and so $h(\cf)=\cf^a$, that is, item 3.\ of Claim~\ref{c:hinf} holds.
On the other hand, if $(h(\ct),\cf^a)$ is $\prec_{\q^a}$-incomparable, then  $(h(\ct),\ct^a)$ is $\prec_{\q^a}$-comparable,
and so $h(\ct)=\ct^a$. Thus, $(h(\cf),\cf^a)$ is $\prec_{\q^a}$-incomparable, and so $(h(\cf),\ct^a)$ is $\prec_{\q^a}$-comparable. Therefore, $h(\cf)\prec_{\q^a}\ct^a$ holds, as required in item 4.\ of Claim~\ref{c:hinf}.
\end{proof}

Now we can complete  the proof of Theorem~~\ref{thm:tree1}~$(ii)$ by observing that
none of the cases in Claim~\ref{c:hinf} is possible, whenever $\q$ contains neither $FT$-twins nor
$\prec$-comparable \solitarypair s. 


\section{Proof of Claim~\ref{lemmaFO}}\label{a:FOlemmaproof}

We prove the following:

\smallskip
\ {\sc Claim}~\ref{lemmaFO} {\em 
The d-sirup $(\Delta_\q,\G)$ is FO-rewritable iff, for any periodic structure $\mathfrak P = (\pre,\per,\post)$ with $\per\ne\emptyset$, one of the following holds:
\begin{enumerate}
\item[{\bf (h1)}] there is a homomorphism from some cactus to the $\bar\cdot$-closure of the acyclic version of $\pre \cup \per$\textup{;}

\item[{\bf (h2)}] there is a homomorphism from the root segment of some cactus to $\bar \per$\textup{;} 

\item[{\bf (h3)}] there is a homomorphism 
from the root segment of one of the $\bar{\mathfrak H}_v$ to $\bar\post$.
\end{enumerate}}

We say that a skeleton $\C^s$ \emph{fits} a periodic structure $\mathfrak P = (\pre,\per,\post)$ if its $ \chi_\C$-image in $\mathfrak G$ is such that each branch of $\C^s$ consists of two consecutive parts: the first one is mapped by $\chi_\C$ into $\mathfrak H$ ($= \pre \cup \per$) and the second one into some $\mathfrak H_v$ in $\post$, for $v \in R$.
If $\C^s$ fits $\mathfrak P$, we denote by $d(\C^s,\mathfrak P)$ the minimum depth of the nodes in $\C^s$ that belong to the second part of these branches. 
 Note that, for any periodic structure $\mathfrak P$ with non-empty $\per$ and any $d < \omega$, there is a cactus $\C$ such that $\C^s$ fits $\mathfrak P$ and $d(\C^s,\mathfrak P) > d$. (If the source of $\mathfrak H$ is $(\emptyset,0,C)$, then in the root segment of $\C$, we bud $T(y_j)$ iff $j \in C$. We continue budding as prescribed by $\mathfrak H$, using the cycles in $\per$ to make sure that $d(\C^s,\mathfrak P) > d$, and then acyclic $\mathfrak H_v$, which end in leaf types that give rise to leaf segments in $\C$ with all of the $T(y_i)$ unbudded.)

$(\Rightarrow)$ Suppose $d<\omega$ is such that every cactus contains a homomorphic image of some cactus of depth $\le d$. 
Let $\mathfrak P = (\pre,\per,\post)$ be a periodic structure with $\per \ne \emptyset$  and let $\C^s$ fit $\mathfrak P$ with `sufficiently large' $d(\C^s,\mathfrak P)$. Take  a minimal cactus $\C'$ of depth $\le d$, for which there is a homomorphism $h' \colon \C' \to \C$. Set $\bar h = \bar \chi_\C h'$ and let $\segs'$ be the root segment in $\C'$.


\emph{Case} {\bf (h1)}: $\bar h(\segs') \cap \bar\pre \ne \emptyset$. Then, since $d(\C^s,\mathfrak P)$ is sufficiently large, $\bar h(\C') \subseteq \bar\pre \cup \bar\per$. If $\bar h(\C')$ goes for more than one cycle in $\bar\pre \cup \bar\per$, we could cut out one cycle and obtain a homomorphism of a smaller cactus $\C''$ into $\C$, contrary to our assumption.
This gives us a homomorphism 
from a small depth cactus to the $\bar\cdot$-closure of the acyclic version of $\pre \cup \per$. 

%

\emph{Case} {\bf (h2)}: $\bar h(\segs') \cap \bar\pre = \emptyset$ and $\bar h(\segs') \cap \bar\per \ne \emptyset$. If $\bar h(\segs') \not\subseteq \bar\per$, then using the fact that $\mathfrak H$ is realisable, we can modify $\bar h$ and obtain the required homomorphism from the root segment $\segs'$ to $\bar\per$.

\emph{Case} {\bf (h3)}: if neither of the previous two cases holds, then we have $\bar h(\segs') \subseteq \bar{\mathfrak H_v}$, for some $\mathfrak H_v$ in $\post$ and $v \in R$.

$(\Leftarrow)$ Consider an arbitrary cactus $\C$, its skeleton $\C^s$ and the canonical homomorphism $\chi_\C$. Our aim is to define a periodic structure $\mathfrak P = (\pre, \per, \post)$, using which we could construct a `small' cactus that is homomorphically embeddable into $\C$. 

We start from the root and move along each branch of $\C^s$ and the $\chi_\C$-images of its nodes in $\mathfrak G$ until the same type repeats twice, that is, we visit the same node twice in $\mathfrak G$. This will give us the pre-periodic $\pre$ and periodic parts $\per$ of $\mathfrak P$. 
We next associate with $\C^s$ a certain $\post$. For each branch of $\C^s$, consider the node after which the $\chi_\C$-image of the branch leaves $\per$. Let $R$ be the set of all of such nodes in $\mathfrak G$. For each $v \in R$, pick some branch in $\C^s$ whose $\chi_\C$-image leaves $\per$ at $v$. Let $p$ be the node on this branch such that $\chi_\C(p) =v \in \per$ but the $\chi_\C$-image of the child of  $p$ on this branch is not in $\per$. Consider the $\chi_\C$-image in $\mathfrak G$ of the subtree $\C^s_p$ of $\C^s$ with root $p$. 
Cut out from $\C^s_p$ the segments between repeating types, if any, so that the $\chi_\C$-image $\mathfrak H_p$ of the remaining part is acyclic. The constructed pairs $(p,\mathfrak H_p)$ form the post-periodic part $\post$. It follows from the construction that the resulting $\mathfrak P = (\pre, \per, \post)$ is a periodic structure. 

\begin{itemize}
\item[--]
If {\bf (h1)} is satisfied, we are done because some cactus can be homomorphically embedded into the $\bar\cdot$-closure of 
the acyclic version of $\pre \cup \per$, and so into some initial part of $\C$.

\item[--]
If {\bf (h2)} holds, then there is a homomorphism from some root segment $\segs$ into a segment in $\bar \per$. This gives a homomorphism from $\segs$ into $\C$, with the image of the root of $\segs$ being in a non-root segment $\segs'$ in $\C$. Using the fact that $\per$ is periodic, we extract from $\C$ a cactus, starting from $\segs'$, of smaller depth that is homomorphically embeddable into $\C$, to which we apply the same argument. 

\item[--]
If {\bf (h3)} holds, we are done by the definition of $\mathfrak H_v$.
\end{itemize}



\section{Proof of Claim~\ref{l:L-hard}}\label{ap:L-h}

We prove the following:

\smallskip
\ {\sc Claim}~\ref{l:L-hard}.
\emph{If none of conditions {\bf (h1)}--{\bf (h3)} holds, then evaluating $(\Delta_\q,\G)$ is $\L$-hard.}

\smallskip
Let $\mathfrak P = (\pre,\per,\post)$ be a periodic structure.  
%
%
%
%
Consider the following additional condition: 
\begin{enumerate}
\item[{\bf (h4)}] there is a homomorphism from the leaf segment to $\bar\pre \cup \bar\per$. 
\end{enumerate}
We claim that
%
\begin{align}
\nonumber
& \mbox{if $\mathfrak P = (\pre,\per,\post)$ is periodic structure with $\per\ne\emptyset$ for which}\\
\nonumber
&\mbox{none of {\bf (h1)}--{\bf (h3)} holds, then there is a periodic structure $\mathfrak P'$}\\
\label{lem:minimal_tree}
&\mbox{such that none of {\bf (h1)}--{\bf (h4)} holds for it.}
\end{align}
%
%
%
Indeed,
let $\mathfrak P$ be a minimal periodic structure, for which none of {\bf (h1)}--{\bf (h3)} holds.  Suppose there is a homomorphism $h$ from the leaf segment $\segs$ into $\bar\pre \cup \bar\per$. Let $a$ be the $A$-node (focus) in $\segs$. Then $h(a)$ is also an $A$-node. It corresponds to an edge $(\segs',\segs'')$ in $\mathfrak P$. We cut this edge, remove from $\mathfrak P$ those segments that are only reachable from $\segs''$, and change $\segs'$ to the segment obtained by replacing the label $A$ on $h(a)$ in $\segs'$ with $T$. The resulting $\mathfrak P'$ is a periodic structure. It is not hard to see that it satisfies none of {\bf (h1)}--{\bf (h3)}, contrary to the minimality of $\mathfrak P$, which proves \eqref{lem:minimal_tree}.
%

We are now in a position to show $\L$-hardness of evaluating $(\Delta_\q,\G)$ with a periodic structure, for which none of {\bf (h1)}--{\bf (h4)} holds. The proof is by reduction of  undirected reachability. Suppose we are given a graph $G$, two vertices $s$ and $t$ in it, and we want to decide whether $t$ is reachable from $s$. 

Let $\mathfrak P$ be a periodic structure, for which none of {\bf (h1)}--{\bf (h4)} holds. 
Consider its blow-up $\bar{\mathfrak P}$. Replace each node $v$ in $G$ by a fresh copy $\bar\per_v$ of $\bar\per$ in $\mathfrak P$. For any two connected nodes $u$ and $v$, if there is a directed edge between nodes $\alpha$ and $\beta$ in $\bar\per$, we add the same edge between $\alpha_u$ and $\beta_v$, and also between $\alpha_v$ and $\beta_u$. We additionally attach $\bar\pre$ to $\bar\per_s$ and $\bar\post$ to $\bar\per_t$ in the same way as in $\mathfrak P$. We regard the resulting structure as a data instance $\A$ and show that there is a path from $s$ to $t$ in $G$ iff the answer to $(\Delta_\q,\G)$ over $\A$ is `yes'\!.

$(\Rightarrow)$ Suppose there is a path from $s$ to $t$. Consider a sufficiently large cactus $\C$ with the periodic structure $\mathfrak P$. Let $h$ be a homomorphism from $\C$ to $\A$, which canonically maps the pre-periodic part of $\C$ to $\bar\pre$ and then, having entered $\bar\per$, it follows the path from $s$ to $t$ in $G$ via the edges of the form $(\alpha_u,\beta_v)$ and uses $\bar\post$ at $t$. The cactus $\C$ should be large enough to reach $t$ in this way. Now, taking an arbitrary assignment of $T$ and $F$ to $A$, we see that we always have a segment in $\A$ (under this assignment) isomorphic to $\q$.

$(\Leftarrow)$ Suppose now there is no path from $s$ to $t$ in $G$. Denote by $V_s$ the set of nodes reachable from $s$ and by $V_t$ the remainder in $G$ (so there are no edges between $V_s$ and $V_t$). Consider an assignment under which all $A$-labels in $V_s$ become $F$ and $A$ in $V_t$ become $T$. We show that there is no homomorphism from $\q$ to the resulting data instance.

Suppose there is $h \colon \q  \to V_s$. This means that there was a homomorphism from  either $\q$, or the leaf segment into the data instance before the assignment. Both of these cases are impossible by \eqref{lem:minimal_tree}.

Suppose there is $h \colon \q \to V_t$. This means that there was a homomorphism from  either $\q$, or a root segment into the data instance before the assignment, which is again a contradiction with  \eqref{lem:minimal_tree}.


\section{Complexity of checking the criterion of Claim~\ref{lemmaFO}}\label{a:complexity}

\newcommand{\perstr}{\mathfrak P}

\begin{lemma}
For $\Lambda$-CQs of span $k$, one can check the criterion of Claim~\ref{lemmaFO} in time $p(|\q|)2^{p'(k)}$, for some polynomials $p$ and $p'$. Thus, deciding FO-rewritability of d-sirups with a $\Lambda$-CQ is fixed-parameter tractable, if the CQ's span is regarded as a parameter. 
\end{lemma}
\begin{proof}
We slightly modify the criterion of Claim~\ref{lemmaFO}. Suppose we have a periodic structure $\perstr$ and a homomorphism from a cactus $\C$ into $\bar \perstr$. We say that this homomorphism is \emph{anchored} if the image of the root segment in $\C$ includes some  $A$-node of the root segment in $\perstr$. We also use the same notion for a homomorphism from a root segment: it is \emph{anchored} if its image includes some $A$-node of the root segment in $\perstr$.

The modified criterion is as follows: $(\Delta_\q,\G)$ is FO-rewritable iff, for any periodic structure $\perstr$, there is either $(i)$ an unanchored homomorphism from a root segment to $\bar{\perstr}$, or $(ii)$ an anchored homomorphism from an acyclic cactus to $\bar{\perstr}$. The proof of this criterion remains essentially the same as the proof of Claim~\ref{lemmaFO}. Our aim is to show that it can be checked in time $p(|\q|)2^{p'(k)}$.

As well-known (see, e.g.,~\cite{DBLP:journals/jacm/Grohe07} and references therein), we can check whether there is a homomorphism from a segment into a type in polynomial time. 
Clearly, there are $2^{O(k)}$ nodes in the digraph $\mathfrak G$. 

We call a node (type) $t$ in $\mathfrak G$ \emph{black}, if we can homomorphically map some  root segment into the blow-up $\bar t$ of $t$. We can find all black nodes in $\mathfrak G$ in time $p(|\q|)2^{p'(k)}$. 

We call a node $v$ \emph{blue} if, for every $\mathfrak H_v$ starting from it, there is a homomorphism of a root segment into $\bar{\mathfrak H}_v$. We can find all blue nodes in time polynomial in the size of $\mathfrak G$. Indeed, consider a game between two players. They move a pebble over the nodes of $\mathfrak G$. In each node, the first player picks one outgoing edge for each possible label. The second player picks one of these edges and the players proceed along this edge to the next node. Black positions are winning for the second player. Non-black leaf-types are winning for the first player. It is easy to see that, for a given $v$, there is $\bar{\mathfrak H}_v$ with no root segment embeddings iff the node $v$ is winning for the first player ($\mathfrak H_v$ describes a winning strategy). So, to check whether a node is blue, it is enough to check the winner in this node. 

If a periodic structure has a black or a blue node, which is not a child of the root,   case $(i)$ holds  for this structure. 
Now, we need to check whether every acyclic structure, for which there are no such nodes, satisfies $(ii)$. 

By definition, edges in $\mathfrak G$ correspond to (budded) $A$-nodes in $\bar{\mathfrak G}$. For an edge $(u,v)$ in $\mathfrak G$, we say that we can \emph{cut it at depth $d$} if, for every periodic structure starting from this edge (more precisely, from the corresponding $A$ node), there is a focused homomorphism (a homomorphism is focused if the focus of a cactus, with $F$ renamed to $A$, is mapped into the focus of the periodic structure) into it from a cactus $\C^{F\to A}$ of depth at most $d$. The image of $\C^{F\to A}$ here is allowed to use any nodes of $\bar u$ save its focus. The property `to be cuttable at depth $d$' can be checked by recursion on $d$. For $d=1$, it just means that there is a homomorphism of a leaf segment into $u$ or $v$ using the $A$-node corresponding to $(u,v)$. For the recursion step, we need to try all possible extensions of $v$, not using coloured nodes, and to check all possible homomorphisms of a segment in $u$ or $v$ that map the focus node of a segment into $(u,v)$ and that does not use the focus node of $u$. If, for any extension of $v$,  there is a homomorphism of a segment that, apart from $(u,v)$, uses only $A$ nodes cuttable at depth $d-1$, then $(u,v)$ is cuttable at depth $d$.
We keep finding all nodes that can be cut at depth $d = 1, 2, 3, \dots$ until we find $d$, for which there are no new such nodes. Fix this $d$. The operation described above can be done in time $p(|\q|)2^{p'(k)}$.

Next, for each root segment in $\bar{\mathfrak G}$, we consider all of its neighbourhoods of depth 1 and check whether there is a homomorphism of a root segment in such a neighbourhood such that all $A$-nodes used by it can be cut at depth $d$. If this is the case, then in any periodic structure, in which there is no unanchored homomorphism  of a root segment, there is an anchored homomorphism of a cactus of depth $d$ (that can be made acyclic in the usual way).

Otherwise, we claim that there is a periodic structure for which neither $(i)$ nor $(ii)$  holds. To show this, we start from a root and its depth one neighbourhood that cannot be cut at depth $d$ and, in each $A$-node that cannot be cut at depth $d$, we proceed in such a way that one of the next nodes cannot be cut at depth $d$.
Suppose there is a homomorphism of a cactus of depth into the resulting periodic structure. Consider the smallest cactus with this property. Take an $A$-node of depth $d+1$ in this cactus. It is mapped into some edge in our periodic structure. Due to the minimality of the cactus, this $A$-node cannot be cut at depth $d$. However, all of its children in the cactus can be cut at depth $d$, which is a contradiction.
\end{proof}


\section{Proof of Theorem~\ref{t:tree-1-dich}}\label{a:trich}

We prove the following:

\smallskip
\  {\sc Theorem}~\ref{t:tree-1-dich}.
\emph{For any a ditree CQ $\q$ with one solitary $F$ and one solitary $T$, $(\Delta_\q,\G)$ is either FO-rewritable, or $\L$-complete, or \NL-complete. Deciding this trichotomy can be done in polynomial time.}

\smallskip
First, we show that if $\q$ is a ditree CQ  with a single solitary $F$-node $\cf$ and a single solitary $T$-node $\ct$, and $\q$ is quasi-symmetric, then $(\Delta_\q,\G)$ is \L-hard. 
%
%
The proof is by an FO-reduction  of the \L-complete reachability problem for undirected graphs.
Given an undirected graph $G = (V,E)$ with nodes $\snode,\tnode \in V$, we construct a data instance $\A_G$ as follows.
We replace each $e = (\unode,\vnode) \in E$ by a fresh copy $\q^e$ of $\q$ 		
such that node $\ct^e$ in $\q^e$ is renamed to $\unode$ with $T(\unode)$ replaced by $A(\unode)$, and node $\cf^e$ is renamed to $\vnode$ with $F(\vnode)$ replaced by $A(\vnode)$. Then $\A_G$ comprises all such $\q^e$, for $e \in E$, as well as $T(\snode)$ and $F(\tnode)$.
We show that $\snode \to_G \tnode$ iff the certain answer to $(\Delta_\q,\G)$ over $\A_G$ is `yes'\!. 

$(\Rightarrow)$ Suppose there is a path $\snode=\vnode_0,  \dots, \vnode_n = \tnode$ in $G$, in which $e_i =(\vnode_i,\vnode_{i+1}) \in E$, for $i < n$.  Then, for
		any model $\I$ of $\Delta_\q$ and $\A_G$, there is some $i < n$ such that  $\I \models T(\vnode_i)$ and $\I \models F(\vnode_{i+1})$. Thus, the identity map from $\q$ to its copy $\q^{e_i}$ is a $\q\to\I$ homomorphism, as required.
		
$(\Leftarrow)$ Suppose $\snode \not\to_G \tnode$. Define a model $\I$ of $\Delta_\q$ and $\A_G$ by labelling with $T$ the $A$-nodes in $\A_G$ that (as nodes in $G$) are reachable from $\snode$ (via an undirected path in $G$) and with $F$ the remaining ones. We claim
		that there is no homomorphism $h \colon \q \to \I$. Indeed, suppose to the contrary  that such $h$ exists. Consider the sub-structure $\wormh$ of $\A_G$ comprising those copies of $\q$  that have a non-empty intersection with $h(\q)$. As $\q$ is quasi-symmetric, $(\ct,\cf)$ is $\prec$-incomparable, and so $\wormh$ looks like in the proof of Theorem~\ref{thm:tree1}~$(ii)$. Thus, we have the four cases of Claim~\ref{c:hinf}. It follows that $\delta(\itf,\cf)\ne\delta(\itf,\ct)$. On the
other hand, as $\q$ is quasi-symmetric and $(\ct,\cf)$ is its only \solitarypair, $(\ct,\cf)$ is symmetric. Thus, we must have $\delta(\itf,\cf)=\delta(\itf,\ct)$, which is a contradiction.		

\smallskip
Next, suppose that $(\ct,\cf)$ is not $\prec$-comparable and $\q$ is not quasi-symmetric.
We consider two 
models $\I$ over the structure $\wormh$ (defined in the proof of Theorem~\ref{thm:tree1}~$(ii)$): one has both contacts in $F^{\I}$, the other in $T^{\I}$. We check whether there exists a homomorphism from $\q$ to either of these models: If neither, then  
$(\Delta_\q,\G)$ is \NL-hard by the proof of Theorem~\ref{thm:tree1}~$(ii)$.
If at least one of them is possible, then we show that $(\Delta_\q,\G)$ is FO-rewritable:
We will use the homomorphism $h\colon\q\to\I$ to define homomorphisms from some depth $\leq 2$ cactus 
to any larger cactus, and then apply the criterion of Prop.~\ref{thmequi}.

By Claim~\ref{c:hinf}, we have four cases for $h$. When $\q$ contains a single solitary $F$- and a single solitary $T$-node,
then case (2) is a `symmetrical' version of case (1), and case (4) is a `symmetrical' version of case (3). So below we deal
with cases (1) and (3) only. In each of these cases,
we claim that
\begin{multline}\label{cembedone}
\mbox{for any $n> 2$, either there is a homomorphism $g\colon\C_1\to\C_n$,}\\
\mbox{or there is a homomorphism $g\colon\C_2\to\C_n$,}
\end{multline}
where $\C_n$ is the unpruned $\q$-cactus (budded at the $T$-node) whose skeleton is of depth $n$.

Indeed,  we define $g$ in each of the cases.
Suppose that the segments of $\C_n$ are $\mathfrak s_0^n$ (root), $\dots,\mathfrak s_n^n$ (leaf).
\begin{enumerate}
\item
There are two cases, depending on whether $h(\q)$ does not intersect with $\q^{a+1}-\{\cf^a\}$ (like in $\q_7$ above), or it does (like in $\q_8$ of Example~\ref{e:three}). 
In the former case, we can map $\C_1$ to any larger cactus, while in the latter case
only $\C_2$. The map $g\colon\C_i\to\C_n$ for $i=1,2$, is defined as follows:
Each of the non-leaf segments of $\C_i$ is mapped to the same segment in $\C_n$ by its identity map, and
the remaining points in $\mathfrak s_i^i$ are mapped  by taking
\[
g(z)=\left\{
\begin{array}{ll}
\iota^a\bigl(h(z)\bigr) \mbox{ in }\mathfrak s_{i-1}^n, & \mbox{ if $h(z)\in\q^a$},\\
\iota^{a-1}\bigl(h(z)\bigr) \mbox{ in }\mathfrak s_{i}^n, & \mbox{ if $h(z)\in\q^{a-1}$},\\
\iota^{a+1}\bigl(h(z)\bigr) \mbox{ in }\mathfrak s_{i-2}^n, & \mbox{ if $i=2$ and $h(z)\in\q^{a+1}$}.
\end{array}
\right.
\]
%
It is straightforward to check that $g$ is well-defined and it is a homomorphism.

\item[(3)]
We can map $\C_1$ to any larger cactus by the map $g\colon\C_1\to\C_n$ defined as follows:
%
The root segment $\mathfrak s^1_0$ of $\C_1$ is mapped to the root segment $\mathfrak s_0^n$ of $\C_n$ by its identity map, and
the remaining points in $\mathfrak s_1^1$ are mapped  by taking
\[
g(z)=\left\{
\begin{array}{ll}
\iota^a\bigl(h(z)\bigr) \mbox{ in }\mathfrak s_{1}^n, & \mbox{ if $h(z)\in\q^a$},\\
\iota^{a-1}\bigl(h(z)\bigr) \mbox{ in }\mathfrak s_{2}^n, & \mbox{ if $h(z)\in\q^{a-1}$},\\
\iota^{a+1}\bigl(h(z)\bigr) \mbox{ in }\mathfrak s_{0}^n, & \mbox{ if $h(z)\in\q^{a+1}$}.
\end{array}
\right.
\]
\end{enumerate}
Again, it is straightforward to check that $g$ is well-defined and it is a homomorphism.


\end{document}